\documentclass[a4paper,11pt,DIV=calc]{scrartcl}
\usepackage[pdfpagelabels]{hyperref}
\usepackage[utf8]{inputenc}
\usepackage{amsmath,amsfonts,amssymb,amsthm,mathtools}
\usepackage{amsmath}
\usepackage{amsfonts}
\usepackage{amssymb}
\usepackage{dsfont}
\usepackage{mathrsfs}
\usepackage[T1]{fontenc}
\usepackage[ngerman,english]{babel}
\usepackage[inline]{enumitem}
\usepackage{latexsym,graphicx}
\usepackage[all]{xy}
\usepackage{mathtools}
\usepackage{color}
\usepackage{authblk}
\usepackage{amscd}
\usepackage{microtype}
\usepackage{mathrsfs,cite}
\numberwithin{equation}{section}

\usepackage{txfonts}

\DeclareMathOperator{\tr}{Tr}
\DeclareMathOperator{\trs}{tr}

\newtheorem{theorem}{Theorem}[section]
\newtheorem{proposition}{Proposition}[section]
\newtheorem{lemma}{Lemma}[section]
\newtheorem{corollary}{Corollary}[section]
\theoremstyle{definition}
\newtheorem{remark}{Remark}[section]

\newtheorem{assumption}{Assumptions}[section]

\newcommand{\dda}{\mathrm{d}}
\newcommand{\de}{\,\dda}

\renewcommand\rho\varrho
\renewcommand\epsilon\varepsilon

\theoremstyle{definition}

\newcommand{\beq}{\begin{equation}}
\newcommand{\eeq}{\end{equation}}
\usepackage{pdfsync}

\begin{document}

\title{Semiclassical approximation and  critical temperature shift for weakly interacting trapped bosons}

\author{Andreas Deuchert and Robert Seiringer}

\date{\today}

\maketitle

\begin{abstract} 
We consider a system of $N$ trapped bosons with repulsive interactions in a combined semiclassical mean-field limit at positive temperature.  We show that the free energy is well approximated by the minimum of the Hartree free energy functional -- a natural extension of the Hartree energy functional to positive temperatures. The Hartree free energy functional converges in the same limit to a  semiclassical free energy functional, and we show that the system displays Bose--Einstein condensation if and only if it occurs in the semiclassical free energy functional. This  allows us to show that for weak coupling the critical temperature decreases due to the repulsive interactions.
\end{abstract}

\setcounter{tocdepth}{2}
\tableofcontents

\section{Introduction and main results}
\subsection{Background and summary}
\label{sec:background}
The rigorous mathematical analysis of quantum many-particle systems has a long history, dating back to the early days of quantum mechanics. In case of (dilute) Bose gases, there has been a period of renewed interest since the first experimental observation of Bose--Einstein condensation (BEC) in trapped alkali gases in 1995 \cite{WieCor1995,Kett1995} and the breakthrough work of Lieb and Yngvason in 1998 \cite{LiYng1998}, who proved a lower bound for the ground state energy of the dilute Bose gas in the thermodynamic limit. In combination with the matching upper bound that Dyson had proven in 1957 \cite{Dyson}, this established its leading order asymptotics. By now, the techniques of Lieb and Yngvason have been significantly extended to prove related results for the ground state energy of the two-dimensional Bose gas \cite{LiYng2001}, the free energy of two- and three-dimensional Bose gases \cite{Sei2008,Yin2010,me3,MaySei}, as well as the ground state energy  \cite{LiSeiSol2005,HP20} and  pressure  \cite{RobertFermigas} of the Fermi gas. Recently, also the next to leading order correction to the ground state energy of the dilute Bose gas, that is, the Lee--Huang--Yang formula, could be established \cite{YauYin2009,FournSol2019}.

While the thermodynamic limit is appropriate to describe samples of macroscopic size, the Gross--Pitaevskii (GP) limit is relevant for the study of (mesoscopic) dilute trapped Bose gases as prepared in typical experiments with cold atoms. In such a situation the ground state energy of the interacting system is to leading order given by the minimum of the GP energy functional \cite{RobertGPderivation,LiSei2002,LiSei2006,NamRouSei2016}, and a convex combination of projections onto its minimizers approximates the one-particle density matrix (1-pdm) of any approximate minimizer of the energy. In case of a unique minimizer of the GP energy functional, this, in particular, proves complete BEC for approximate ground states. Also in the GP limit the next to leading order correction to the ground state energy predicted by Bogoliubov in 1947 could be justified \cite{BogGP}. The accuracy reached in this work allows for an approximate computation of the ground state wave function and for a characterization of the low energy excitations of the system, consisting of sound waves. Apart from equilibrium properties of dilute Bose gases, also their dynamics after the trapping potential has been switched off is  important  for the interpretation of modern experiments. The dynamics of initially fully Bose--Einstein condensed systems in the GP limit can be described by the time-dependent GP equation, see \cite{ErdSchlYau2009,ErdSchlYau2010,BenOlivSchl2015,Pickl2015}. For a more extensive list of references to the mathematical analysis of dilute Bose gases we refer to \cite{Themathematicsofthebosegas,Rou2015,BenPorSchl2015}.

The ground state energy and the ground state wave function yield a good approximation of the system for very low temperatures. If positive temperatures are relevant one needs to consider the free energy and the related Gibbs state. Recently it was  shown that also for a system of trapped Bosons at positive temperature the GP energy functional turns out to be the relevant effective theory. More precisely, it was shown in \cite{me} that for a system in a harmonic trap in the GP limit the difference between the free energies of the interacting and the ideal gas is given by the minimum of a GP energy functional to leading order. Additionally, the 1-pdm of any approximate minimizer of the Gibbs free energy functional is to leading order given by the one of the ideal gas, where the condensate wave function has been replaced by the minimizer of the GP energy functional. The result shows that the interaction can be seen to leading order only in the condensate, which is related to the fact that the energy per particle in the thermal cloud (all particles outside the condensate) is much larger than in the condensate. In a trap with soft walls the thermal cloud is therefore even more dilute than the condensate. Modern experimental techniques also allow for the study of dilute Bose gases in box potentials \cite{Gauntetal2013}. The computation of the free energy for such a system is mathematically more challenging than for a system in a power law trap because all interactions are relevant to leading order (there is no separation of length scales). With techniques based on the ones that were introduced for the analysis of the free energy in the thermodynamic limit \cite{Sei2008}, a proof of the BEC phase transition for this system was  given in \cite{me2}.

Although the GP limit is the most relevant for the description of experiments with cold quantum gases, there has been a considerable interest in systems in the mean-field (MF) (or Hartree) limit of weak and long-range interactions, or in limits that interpolate between the MF and the GP limit. See e.g.  \cite{LieYau1987,LewinNamRougerie20141,LewinNamRougerie20142,Bogoliubov_Robert_1,GrechSei2013,Bogoliubov_Mathieu_1,Robert_Nam_1} for works concerned with ground state properties of such systems, \cite{Spohn1980,ErdYau2001,Adami2007,ElgSchl2007,FroeKnoPiz2007,AmmNier2009,FroeKnowSchw2009,KnoPick2010,NaNap20171,NaNap20172,LeNaSchl2015,GriMach2010,GriMach2011,KiSchlStaf2011,AmmFalcPawl2016,CheHolm2016,AnaHott2016,GriMach2017} for their dynamics and \cite{LewinNamRougerie2015,LewinNamRougerie2017,LewinNamRougerie2018,FKSS2017,FKSS2018,FKSS2020} for the analysis of systems at temperatures slightly above the critical temperature for BEC. For a more extensive list of references concerning the MF limit we refer to \cite{Lewin2015,BenPorSchl2015}. 
The MF limit is  mathematically easier to handle than the GP limit, and therefore allows for the development of new techniques that can later be applied to study the more complicated GP regime. Note, moreover,  that long range interactions among bosons have been successfully implemented in recent experiments \cite{explongrange}.
 
A {\em semiclassical MF limit} for fermions was  considered in \cite{BenPorSchl2014,BenJaPoSaSchl,FouLeSol2018,LewinMadsenTriay}. This regime is also of relevance for bosons at temperatures of the order of the critical temperature and has been considered in \cite{BaNaTh83}. It shall also be our main concern here. We prove optimal bounds for the difference between the free energy of the system and the minimum of the Hartree free energy functional. Moreover, we show that the Hartree free energy functional can be related to a novel semiclassical free energy functional, whose critical temperature for BEC can be characterized explicitly in the weak coupling regime. This allows us to show that the repulsive interactions decrease the  critical temperature for BEC, at least for weak coupling. 

We also consider a second natural scaling limit, whose relevance stems from the fact that the Hartree energy functional emerges in the zero temperature limit. We provide optimal bounds quantifying the difference between the free energy and its Hartree approximation, and discuss implications for approximate minimizers of the Gibbs free energy functional.
\subsection{The model}
\label{sec:model}
We are interested in a trapped bosonic many-particle system at positive temperature and we start by introducing the model. Let $h$ denote the harmonic oscillator Hamiltonian with oscillator frequency $\omega$ acting on the one-particle Hilbert space $L^2(\mathbb{R}^3)$, 
\begin{equation}
	h = - \hbar^2 \Delta + \frac{\omega^2 x^2}{4}.
	\label{eq:harmonicoscillator}
\end{equation}
Due to the factor $1/4$ multiplying the trapping potential, the level spacing of $h$ equals $\hbar \omega$. Without loss of generality we could set $\omega$ equal to $1$, but we prefer  to keep it general in order to explicitly display physical units. The particle mass is set equal to $1/2$, and Planck's constant $\hbar$ is chosen to depend on the particle number $N$ of the system as
\begin{equation}\label{hbar}
\hbar = N^{-1/3} \,, 
\end{equation}
which corresponds to a semiclassical regime. The physics of the $N$-particle system is governed by the $N$-particle Hamiltonian 
\begin{equation}
	H_N = \sum_{i = 1}^N h_i + \frac{1}{N} \sum_{1 \leq i < j \leq N} v(x_i - x_j)
	\label{eq:Hamiltonian1stqu}
\end{equation}
acting on $L^2_{\mathrm{sym}}(\mathbb{R}^{3N})$, the closed linear subspace of $L^2 (\mathbb{R}^{3N} )$ consisting of those functions $\Psi(x_1,...,x_N)$ that are invariant under any permutation of  coordinates $x_1, \dots, x_N$. 
We employ the usual notation that $h_i$ is the operator acting as $h$ on the $i$-th particle, and the identity on all the others.

The motivation for the prefactor $N^{-1}$ in \eqref{eq:Hamiltonian1stqu}, as well as the choice \eqref{hbar} of $\hbar$,  will be explained below after introducing the relevant temperature scale. 
Concerning the regularity of the interaction potential $v$, we make the following assumption.
\begin{assumption}
	\label{as:regularitypotential}
	We assume that $v \in L^1(\mathbb{R}^3) \bigcap W^{2,\infty}(\mathbb{R}^3)$ is a nonnegative function with $v(x) = v(-x)$ for all $x \in \mathbb{R}^3$, such that
		\begin{equation}
	\sup_{x \in \mathbb{R}^3} \left\Vert D^2 v(x) \right\Vert < \frac{\omega^2}{2}
	\label{eq:condintpot}
	\end{equation}
	holds. Here $D^2v$ denotes the Hessian of $v$ and $\Vert \cdot \Vert$ is the operator norm on $3 \times 3$ matrices. Additionally, we assume $\hat{v} \geq 0$, that is, $v$ is of positive type.
\end{assumption}

We shall comment on the significance of the assumption \eqref{eq:condintpot} in Remark~\ref{rem6} after the statement of our main results. 

\subsubsection*{Gibbs free energy functional, free energy and Gibbs states}

Let us denote by 
\begin{equation}
	\mathcal{S}^{\mathrm{c}}_N = \left\{ \Gamma \in \mathcal{B}\left( L^2_{\mathrm{sym}}\left(\mathbb{R}^{3N} \right) \right) \ \Bigg| \ 0 \leq \Gamma \leq 1, \ \tr  \Gamma  = 1, \ \tr\left[\sum_{i=1}^N h_i \Gamma \right] < + \infty \right\}
	\label{eq:canonicalstates}
\end{equation} 
the set of bosonic $N$-particle states. 
In the above definition and in the following, we interpret $\tr[H \Gamma]$ for positive operators $H$ and $\Gamma$ as $\tr[H^{1/2} \Gamma H^{1/2}]$. This expression is always well-defined if one allows the value $+\infty$. In particular, finiteness of $\tr[H \Gamma]$ does not require the operator $H \Gamma$ to be trace-class, only that $H^{1/2} \Gamma H^{1/2}$ is. 

For states $\Gamma \in \mathcal{S}^{\mathrm{c}}_N$ the Gibbs free energy functional at inverse temperature $\beta$ is defined by
\begin{equation}
	\mathcal{F}(\Gamma) = \tr\left[ H_N \Gamma \right] - \frac{1}{\beta} S(\Gamma), \quad \text{ with the von Neumann entropy } \quad S(\Gamma) = - \tr \left[ \Gamma \ln(\Gamma) \right].
	\label{eq:gibbsfreeenergyfunctional}
\end{equation}
When we minimize $\mathcal{F}$ over states in $\mathcal{S}^{\mathrm{c}}_N$, we obtain the canonical free energy of the system
\begin{equation}
	F^{\mathrm{c}}(\beta,N,\omega) = \inf_{\Gamma \in \mathcal{S}^{\mathrm{c}}_N} \mathcal{F}(\Gamma) = -\frac{1}{\beta} \ln \left( \tr \exp\left( - \beta H_N \right) \right).
	\label{eq:freeenergycan}
\end{equation}
The unique minimizer of the above minimization problem is the canonical Gibbs state
\begin{equation}
G_N^{\mathrm{c}} = \frac{e^{-\beta H_N }}{\tr \left[ e^{-\beta H_N} \right]}.
\label{eq:Gibbsstatecan}
\end{equation}

\subsubsection*{Motivation for the scaling limit}

To motivate our scaling limit let us for the moment set $\hbar = 1$. In the absence of interactions, the ideal Bose gas shows BEC for $\beta > \beta_{0,N} = \omega^{-1} (\zeta(3)/ N)^{1/3}$, where $\zeta$ denotes the Riemann zeta function. This  shows that the relevant regime for the inverse temperature is $\beta\omega \sim N^{-1/3}$. In this regime, the expected energy per particle in the ideal gas is proportional to $\beta^{-1} = O( \omega N^{1/3} )$. Particles with this  energy typically occupy a volume of order $( \omega^{-1/2} N^{1/6} )^3$ in the harmonic trap (see also the introduction in \cite{me}). Accordingly, it is natural to consider  interactions acting on that length scale, i.e., an interaction potential of the form $v_N(x) = c_N \omega v(N^{-1/6} \omega^{1/2} x)$ for some coupling constant $c_N>0$. When evaluated in the Gibbs state of the ideal gas, 
the corresponding interaction energy is of the order $\omega c_N N^{2}$. The choice $c_N = N^{-2/3}$ thus makes this comparable to the free energy of the ideal gas in the harmonic trap, which is of the order $\omega N^{4/3}$. To relate this scaling to our semiclassical MF scaling, we rescale length by a factor $N^{-1/6}$, which transforms the $N$-particle Hamiltonian 
\begin{equation}
	\sum_{i=1}^N \left( -\Delta_i + \frac{\omega^2 x^2}{4} \right) + \sum_{1 \leq i < j \leq N} \omega N^{-2/3} v(N^{-1/6} \omega^{1/2} (x_i - x_j))
\end{equation}
into
\begin{equation}
N^{1/3} \left( \sum_{i=1}^N \left( -\hbar^2 \Delta_i + \frac{\omega^2 x^2}{4} \right) + \sum_{1 \leq i < j \leq N} \omega N^{-1} v(\omega^{1/2} (x_i - x_j)) \right)
\end{equation}
with $\hbar = N^{-1/3}$ as in \eqref{hbar}. This motivates our scaling. To simplify the notation, we absorb the $\omega$ factors into the potential $v$, leading to $H_N$ in \eqref{eq:Hamiltonian1stqu}. 
Moreover, we neglect the multiplicative factor $N^{1/3}$  in the definition of $H_N$ in \eqref{eq:Hamiltonian1stqu}, thus we have to choose $\beta \sim N^{1/3} \beta_{0,N} \sim \omega^{-1}$.  In fact, we only assume $\beta \omega \geq C $ for some $C > 0$, allowing in particular for the zero-temperature limit $\beta\omega \to \infty$. In case of fermions the same scaling limit has been considered e.g. in \cite{BenPorSchl2014,BenJaPoSaSchl,FouLeSol2018,LewinMadsenTriay}.

\subsubsection*{One-particle reduced density matrix and Bose--Einstein condensation}

For a state $\Gamma \in \mathcal{S}^{\mathrm{c}}_N$, we define its one-particle density matrix (1-pdm) via its integral kernel by
\begin{equation}
\gamma_{\Gamma}(x,y) = \tr \left[ a^*_y a_x \Gamma \right].
\label{eq:1pdma}
\end{equation}
Here $a_x^*$ and $a_x$ denote the usual creation and annihilation operators (actually operator-valued distributions) of a particle at a point $x \in \mathbb{R}^3$, which obey the canonical commutation relations $[a_x,a_y^*] = \delta(x-y)$. By $\delta(x)$ we denote the Dirac delta distribution. The operator $\gamma_{\Gamma}$ is a positive operator on the one-particle Hilbert space $L^2(\mathbb{R}^3)$ with trace equal to $N$. Equivalently, 
\begin{equation}
\gamma_{\Gamma}(x,y) = N \int_{\mathbb{R}^{3(N-1)}} \Gamma(x,x_2,...,x_N;y,x_2,...,x_N) \de (x_2,...,x_N),
\label{eq:1pdmb}
\end{equation}
with the integral kernel $\Gamma(x_1,...,x_N;y_1,...,y_N)$ of the state $\Gamma$. By $\varrho_{\Gamma}(x) = \gamma_{\Gamma}(x,x)$ we denote the density of a state $\Gamma$.

A sequence of states $\Gamma_N \in \mathcal{S}^{\mathrm{c}}_N$ indexed by the particle number shows BEC if and only if
\begin{equation}
	\liminf_{N\to \infty} \frac{ \sup_{\Vert \psi \Vert = 1} \langle \psi, \gamma_{\Gamma_N} \psi \rangle }{N} > 0
	\label{eq:BEC}
\end{equation} 
holds for its 1-pdms $\gamma_{\Gamma_N}$. In case the sequence $\Gamma_N$ shows BEC, we call the eigenvector corresponding to the largest eigenvalue of its 1-pdm $\gamma_{\Gamma_N}$ the condensate wave function.

\subsubsection*{The Husimi function}

To be able to state our main result we need to introduce the Husimi function related to a state, which, roughly speaking, contains the information how the particles are distributed in the classical phase space. Let $\ell \in \mathcal{C}_{\mathrm{c}}^{\infty}(\mathbb{R}^3)$ be a nonnegative radial function with $L^2(\mathbb{R}^3)$-norm equal to $1$ and define for $p,q \in \mathbb{R}^3$
\begin{equation}
\ell_{p,q}^{\hbar}(x) = \hbar^{-3/4} \ell \left( \frac{x-q}{\hbar^{1/2}} \right) e^{i p x / \hbar}.
\label{eq:coherentstate1i}
\end{equation}
For a given  1-pdm $\gamma$ we define its Husimi function by
\begin{equation}
m_{\gamma}(p,q) = \left\langle \ell_{p,q}^{\hbar}, \gamma \ell_{p,q}^{\hbar} \right\rangle \geq 0.
\label{eq:Husimifunctioni}
\end{equation}
Since the coherent states $| \ell_{p,q}^{\hbar} \rangle \langle \ell_{p,q}^{\hbar} |$ yield a resolution of the identity of the form
\begin{equation}
\left( \frac{1}{2 \pi \hbar} \right)^3 \int_{\mathbb{R}^6} \left| \ell_{p,q}^{\hbar} \right\rangle \left \langle \ell_{p,q}^{\hbar} \right| \de(p,q) = \mathds{1}_{L^2(\mathbb{R}^3)},
\label{eq:resid}
\end{equation} 
see e.g. \cite[Theorem~12.8]{LiLo10}, we have
\begin{equation}
	\frac{1}{(2 \pi \hbar)^3} \int_{\mathbb{R}^6} m_{\gamma}(p,q) \de(p,q) = \trs [\gamma],
	\label{eq:89}
\end{equation}
where $\trs[\cdot]$ denotes the trace on the one-particle Hilbert space $L^2(\mathbb{R}^3)$. That is, the Husimi function is a distribution on the classical phase space whose phase space integral equals the particle number $N$.
\subsection{The effective models}
\label{sec:effectivemodels}

\subsubsection*{The Hartree free energy functional}

Our goal is to show that the free energy \eqref{eq:freeenergycan} of the full quantum mechanical model \eqref{eq:Hamiltonian1stqu} can be approximated by a simpler model, which we introduce now. 
We define the set of 1-pdms
\begin{equation}
	\mathcal{D}_N^{\mathrm{H}} = \left\{ \gamma \in \mathcal{B}\left(L^2 \left(\mathbb{R}^3 \right) \right) \ \Big| \ \gamma \geq 0, \ \trs[\gamma] = N, \ \trs[h \gamma] < + \infty \right\}.
	\label{eq:formdomainH}
\end{equation}
For an operator $\gamma \in \mathcal{D}_N^{\mathrm{H}}$ the Hartree free energy functional is defined by
\begin{equation}
	\mathcal{F}^{\mathrm{H}}(\gamma) = \trs \left[ \left( h + \frac{1}{2 N} v \ast \varrho_{\gamma} \right) \gamma \right] - \frac{1}{\beta} s(\gamma),
	\label{eq:Hartrees1pfreeenergyfkt} 
\end{equation}
where $\varrho_{\gamma}(x) = \gamma(x,x)$, and $*$ denotes convolution. By $s(\gamma)$ we denote the bosonic entropy
\begin{equation}
	s(\gamma) = -\trs\left[ f(\gamma) \right], \quad \text{ with } \quad f(x) = x \ln(x) - (1+x) \ln(1+x).
	\label{eq:bosonicentropy}
\end{equation}
The functional $\mathcal{F}^{\mathrm{H}}$ is a natural extension of the Hartree energy functional to positive temperatures. The Hartree free energy is given by
\begin{equation}
	F^{\mathrm{H}}(\beta,N,\omega) = \inf_{\gamma \in \mathcal{D}_N^{\mathrm{H}}} \mathcal{F}^{\mathrm{H}}(\gamma).
	\label{eq1pHartreefreeenergy}
\end{equation}
It is not difficult to see that under our assumptions on $v$ it has a unique minimizer $\gamma^{\mathrm{H}}$ that  solves the Euler--Lagrange equation
\begin{equation}
	\gamma = \frac{1}{ e^{\beta \left( h + N^{-1} v \ast \varrho_{\gamma} - \mu \right)} - 1},
	\label{eq:ELHartree1part}
\end{equation}
where the chemical potential $\mu$ is chosen such that $\trs \gamma^{\mathrm{H}} = N$. (See Lemma~\ref{lem:minimizersHartree} below.)

\subsubsection*{The semiclassical free energy functional}

We shall now introduce a novel semiclassical free energy functional that is independent of $N$ and turns out to capture the correct leading order behavior of the free energy and the 1-pdm of the corresponding Gibbs state in the semiclassical MF limit considered. It is more complicated than its fermionic equivalent in \cite{LewinMadsenTriay} due to the possible occurrence of BEC. 

Let $\mathcal{D}^{\mathrm{sc}}$ denote the set of pairs $(\gamma,g)$ with an integrable function $\gamma(p,x) \geq 0$ on the phase space $\mathbb{R}^3 \times \mathbb{R}^3$ and a number $g \in [0,1]$, satisfying 
\begin{equation}
	\int_{\mathbb{R}^6} \left( p^2 + \frac{\omega^2 x^2}{4} \right) \gamma(p,x) \de(p,x) < +\infty
\end{equation}
as well as  the normalization condition
\begin{equation}
	\frac{1}{(2 \pi)^3} \int_{\mathbb{R}^6} \gamma(p,x) \de(p,x) + g = 1 \,.
	\label{eq:normalizationcondition}
\end{equation}
We think of $g$ as a condensate fraction and say that a pair $(\gamma,g)$ shows BEC if $g > 0$. The phase space density $\gamma$ describes a thermal cloud. 

For $(\gamma,g) \in \mathcal{D}^{\mathrm{sc}}$ we define the  semiclassical free energy functional by 
\begin{equation}
\mathcal{F}^{\mathrm{sc}}(\gamma,g) = \frac{1}{(2 \pi)^3} \int_{\mathbb{R}^6} \left( p^2 + \frac{\omega^2 x^2}{4} \right) \gamma(p,x) \de(p,x) - \frac{1}{\beta} S^{\mathrm{sc}}(\gamma) + \frac{1}{2} \int_{\mathbb{R}^6} v(x-y) \varrho(x) \varrho(y) \de(x,y).
\label{eq:semic-lassicalfunctional3i}
\end{equation}
Here the density $\varrho$ is given by
\begin{equation}
\varrho(x) = \frac{1}{(2\pi)^3} \int_{\mathbb{R}^3} \gamma(p,x) \de p  + g \delta(x),
\label{eq:semiclassicalfunctional2i}
\end{equation} 
where $\delta$ denotes the Dirac delta measure with unit mass at the point $x=0$. Moreover,
\begin{equation}
S^{\mathrm{sc}}(\gamma) = \frac{-1}{(2 \pi)^3} \int_{\mathbb{R}^6} f\left( \gamma(p,x) \right) \de(p,x)
\label{eq:semiclassicalfunctional4i}
\end{equation}
is the bosonic entropy of the phase space density $\gamma$ with $f$ defined in \eqref{eq:bosonicentropy}. 
The semiclassical free energy is defined by
\begin{equation}
F^{\mathrm{sc}}(\beta,\omega) = \inf_{(\gamma,g) \in \mathcal{D}^{\mathrm{sc}}} \mathcal{F}^{\mathrm{sc}}(\gamma,g) \,,
\end{equation}
and the unique minimizer of this minimization problem (see Lemma~\ref{lem:semiclassicalfreeenergy}) is denoted by $(\gamma^{\mathrm{sc}},g^{\mathrm{sc}})$. It solves the self-consistent equation
\begin{equation}
\gamma(p,x) = \frac{1}{\exp\left( \beta \left( p^2 + \frac{\omega^2 x^2}{4} + v \ast \varrho(x) - \mu \right) \right)-1}
\label{eq:semiclassicalfunctional5i}
\end{equation}
for some chemical potential $\mu \leq v \ast \varrho(0)$.

\subsubsection*{The semiclassical ideal Bose gas}

In case of $v = 0$ the semiclassical free energy functional $\mathcal{F}^{\mathrm{sc}}$ in \eqref{eq:semic-lassicalfunctional3i} can be minimized explicitly, and the minimizing phase space density $\gamma_0$ is given by \eqref{eq:semiclassicalfunctional5i} with $v=0$. If 
\begin{equation}\label{eq:idealgascriticaltemp}
\frac{1}{(2 \pi)^3} \int_{\mathbb{R}^6} \frac{1}{\exp\left( \beta \left( p^2 + \frac{\omega^2 x^2}{4}  \right) \right)-1} \de(p,x)  = \frac{\zeta(3)}{(\beta \omega)^3} < 1
\end{equation}
the minimizing condensate fraction $g_0$ is positive and assures the normalization condition in \eqref{eq:normalizationcondition}. The chemical potential related to the minimizing phase space density $\gamma_0$ will be denoted by $\mu_0$ and $\varrho_0(x) = (1/(2\pi))^3 \int_{\mathbb{R}^3} \gamma_0(p,x) \de p + g_0 \delta(x)$. Because of the monotonicity of the map $\mu \mapsto \int_{\mathbb{R}^6} \gamma_0(p,x) \de(p,x)$, the model shows a BEC phase transition with inverse critical temperature $\beta_0$ determined by equality in \eqref{eq:idealgascriticaltemp}, i.e., 
\begin{equation}
	\beta_0 = \frac{\zeta(3)^{1/3}}{\omega}\,.
	\label{eq:idealgascriticaltemp2}
\end{equation}
We have $g_0=0$ and $\mu_0 < 0$ if $\beta < \beta_0$, $0 < g_0 < 1$ and $\mu_0 = 0$ if $\beta > \beta _0$ as well as $g_0=0$ and $\mu_0 = 0$ if $\beta = \beta_0$. Note that $\beta_0$ times $N^{-1/3}$ equals the critical temperature $\beta_{0,N}$ for BEC in the ideal Bose gas in the harmonic trap $\omega^2 x^2/4$, as discussed in  the paragraph about the motivation of our scaling limit above.

\subsection{Notation}

For functions $a$ and $b$ depending on $N$ or on other parameters, we use the notation $a \lesssim b$ to say that there exists a constant $C > 0$ independent of the parameters such that $a \leq Cb$. If $a \lesssim b$ and $ b \lesssim a$ we write $a \sim b$.

\subsection{Main results}
\label{sec:mainresult}
Our first main result concerns the relation between the full quantum mechanical model \eqref{eq:gibbsfreeenergyfunctional} and the Hartree free energy functional \eqref{eq:Hartrees1pfreeenergyfkt}. 
\begin{theorem}
	\label{thm:freeenergyscaling2}
	Let Assumption~\ref{as:regularitypotential} hold. In the limit $N \to \infty$ with $\beta \omega \gtrsim 1$ we have
	\begin{align}
	\left| F^{\mathrm{c}}(\beta,N,\omega) - F^{\mathrm{H}}(\beta,N,\omega) \right| \lesssim \omega N^{1/3}.
	\label{eq:mainresults21}
	\end{align}
	Moreover, for any any sequence of states $\Gamma_N \in \mathcal{S}^{\mathrm{c}}_N$ with 1-pdm $\gamma_N$ and
	\begin{equation}
	\left| \mathcal{F}(\Gamma_N) - F^{\mathrm{H}}(\beta,N,\omega) \right| \leq \omega \delta
	\label{eq:mainresults21b}
	\end{equation}
	for some $\delta > 0$, we have
	\begin{equation}
	\left\Vert \gamma_{N} - \gamma^{\mathrm{H}} \right\Vert_1 \lesssim N^{5/6} (1+ \delta)^{1/4}.
	\label{eq:mainresults21c}
	\end{equation} 
	Here $\gamma^{\mathrm{H}}$ denotes the unique minimizer of $\mathcal{F}^{\mathrm{H}}$ in \eqref{eq:Hartrees1pfreeenergyfkt}.
\end{theorem}

Recall that for $\beta\omega \sim 1$, $F^{\mathrm{c}}(\beta,N,\omega) \sim \omega N$ to leading order. The bound \eqref{eq:mainresults21} states that not only the leading order is correctly captured by the Hartree free energy $F^{\mathrm{H}}(\beta,N,\omega)$, but also lower order terms larger than $O(\omega N^{1/3})$. Moreover, as long as $\delta = o( N^{2/3})$, the 1-pdm of a state satisfying \eqref{eq:mainresults21b} agrees with the Hartree minimizer $\gamma^{\mathrm{H}}$ to leading order. This is in particular true for the true Gibbs state, for which $\delta \lesssim N^{1/3}$.  

Our second main  result concerns the asymptotic behavior of the Hartree free energy functional in the semiclassical MF limit and its relation to the semiclassical free energy functional \eqref{eq:semic-lassicalfunctional3i}.
\begin{theorem}
	\label{thm:limitHartreetheory}
	Let Assumption~\ref{as:regularitypotential} hold. In the limit $N \to \infty$ with $\beta \omega \gtrsim 1$ we have
	\begin{equation}
	\left| F^{\mathrm{H}}(\beta,N,\omega) - N F^{\mathrm{sc}}\left(\beta,\omega \right) \right| \lesssim \omega N^{2/3}.
	\label{eq:propsemiclassical1a}
	\end{equation}
	Moreover, with $\gamma^{\mathrm{H}}$ the unique minimizer of $\mathcal{F}^{\mathrm{H}}$, let $P^{\mathrm{H}}$ be the projection onto the eigenspace of the largest eigenvalue of $\gamma^{\mathrm{H}}$ and define $Q^{\mathrm{H}} = 1-P^{\mathrm{H}}$. Let $(\gamma^{\mathrm{sc}},g^{\mathrm{sc}})$ be the minimizing pair of $\mathcal{F}^{\mathrm{sc}}$ in \eqref{eq:semic-lassicalfunctional3i}. In the limit $N \to \infty$ with $\beta \omega \gtrsim 1$ we  have 
	\begin{subequations}\label{eq:propsemiclassical2a}
	\begin{align}\label{1.36a}	
	\left| N^{-1} \trs\left[P^{\mathrm{H}} \gamma^{\mathrm{H}} \right]  - g^{\mathrm{sc}} \right| &\lesssim N^{-1/9 + \sigma} \quad \text {as well as} \\ \label{1.36b}
	\int_{\mathbb{R}^6} \left| m_{Q^{\mathrm{H}} \gamma^{\mathrm{H}}}(p,x) - \gamma^{\mathrm{sc}}(p,x) \right| \de(p,x) &\lesssim N^{-1/9 + \sigma} 
	\end{align}
	\end{subequations}
	for any $\sigma > 0$.
\end{theorem}

Recall the definition of the  Husimi function $m_\gamma$ in \eqref{eq:Husimifunctioni}. Theorem~\ref{thm:limitHartreetheory} shows that $\gamma^{\mathrm{sc}}$ agrees to leading order with the Husimi function of the Hartree minimizer $\gamma^{\mathrm{H}}$ once the condensate has been removed. Moreover, the condensate fraction in Hartree theory is to leading order given by $g^{\mathrm{sc}}$.

The next two statements are a direct consequence of the bounds derived in order to prove Theorems~\ref{thm:freeenergyscaling2} and~\ref{thm:limitHartreetheory}, and we therefore state them as Corollaries. The first concerns a weak form of BEC under a weak energy condition on  approximate minimizers of the Gibbs free energy functional. The question whether the full quantum mechanical model shows BEC in this sense is equivalent to the question whether the semiclassical free energy functional in \eqref{eq:semic-lassicalfunctional3i} shows BEC in the sense that $g^{\mathrm{sc}}>0$. 

\begin{corollary}
	\label{cor:weakBEC}
	Let Assumption~\ref{as:regularitypotential} hold and assume that $\Gamma_N$ with 1-pdm $\gamma_N$ is an approximate minimizer of the Gibbs free energy functional in the sense that \eqref{eq:mainresults21b} holds with $\delta = o(N)$. Let $B_{r} \subset \mathbb{R}^3$ be the ball with radius $r>0$ centered at the origin and consider the limit $N \to \infty$ with $\beta \omega > 0$ fixed. Then
	\begin{equation}
	\lim_{N \to \infty} \int_{ \left( B_{\omega^{1/2} R} \times B_{\omega^{-1/2} R} \right)^{\mathrm{c}} } \left| m_{\gamma_N}(p,x) - \gamma^{\mathrm{sc}}(p,x) \right| \de(p,x) = 0
	\label{eq:83}
	\end{equation}
	for any $R>0$ and, in particular,
	\begin{equation}
	\lim_{N \to \infty} \left( \frac{1}{2 \pi} \right)^3 \int_{ \mathbb{R}^6 } m_{\gamma_N}(p,x) f(p,x) \de(p,x) = \left( \frac{1}{2 \pi} \right)^3 \int_{ \mathbb{R}^6 } \gamma^{\mathrm{sc}}(p,x) f(p,x) \de(p,x) + g^{\mathrm{sc}} f(0,0) \, \label{eq:83b}
	\end{equation}
	for any bounded continuous function $f : \mathbb{R}^6 \to \mathbb{C}$.
\end{corollary}

Under a stronger energy assumption on  approximate minimizers of the Gibbs free energy functional, the above relation can be strengthened to imply BEC in the sense of \eqref{eq:BEC} for the full quantum mechanical model if the semiclassical free energy functional in \eqref{eq:semic-lassicalfunctional3i} shows BEC. 

\begin{corollary}
	\label{cor:strongBEC}
	Let Assumption~\ref{as:regularitypotential} hold. Assume that $\Gamma_N$ with 1-pdm $\gamma_N$ is an approximate minimizer of the Gibbs free energy functional in the sense of \eqref{eq:mainresults21b}, denote by $P$ the projection onto its largest eigenvalue and define $Q = 1 - P$. In the limit $N \to \infty$ with $\beta \omega \gtrsim 1$ we have 
	\begin{subequations}
	\begin{align}
	\left| \frac{\trs[P \gamma_N]}{N} - g^{\mathrm{sc}} \right| &\lesssim N^{-1/9 + \sigma} + N^{-1/6} \delta^{1/4} \quad \text{ as well as } \label{eq:90} \\ 
	\int_{\mathbb{R}^6} \left| m_{Q \gamma_N}(p,x) - \gamma^{\mathrm{sc}}(p,x) \right| \de(p,x) &\lesssim N^{-1/9 + \sigma} + N^{-1/6} \delta^{1/4} \label{eq:90a}
	\end{align}
	\end{subequations}
    for any  $\sigma > 0$. 
\end{corollary}

In particular, for $\delta = o(N^{2/3})$, the condensate fraction of an approximate minimizer equals the semiclassical value $g^{\mathrm{sc}}$ to leading order. This applies in particular to the actual Gibbs state, for which $\delta \lesssim N^{1/3}$ according to Theorem~\ref{thm:freeenergyscaling2}. We note that the condition $\delta = o(N^{2/3})$ is sharp even for an ideal Bose gas, since the energy gap above the ground state of $h$ in \eqref{eq:harmonicoscillator} equals $\hbar\omega = \omega N^{-1/3}$, and hence moving a fraction of the condensed particles to the first excited state leads to energy increase of the order  $\omega N^{2/3}$.  

Our final statement concerns the BEC  transition temperature in the semiclassical free energy functional and hence, because of Corollary~\ref{cor:weakBEC}~and~\ref{cor:strongBEC}, also the one for the full quantum mechanical model. We show that for weak coupling  there is a unique critical temperature which is strictly lower than the critical temperature for the ideal Bose gas.

\begin{proposition}
	\label{prop:crittempscmf}
	Let Assumption~\ref{as:regularitypotential} hold and assume that the interaction potential is given by $\lambda v(x)$ with $0 < \lambda \leq 1$. Denote by $g^{\mathrm{sc}}$ the condensate fraction of the unique minimizer of the semiclassical free energy functional in \eqref{eq:semic-lassicalfunctional3i} and let $\mu^{\mathrm{sc}}$ be the chemical potential in $\gamma^{\mathrm{sc}}$ in \eqref{eq:semiclassicalfunctional5i}. For small enough $\lambda$,  the following holds:
	\begin{enumerate}[label={\alph*)}]
\item There exists an inverse critical temperature $\beta_{\mathrm{c}}$ such that $g^{\mathrm{sc}} > 0$ and $\mu^{\mathrm{sc}} = 0$ for $\beta > \beta_{\mathrm{c}}$, and  $g^{\mathrm{sc}} = 0$ as well as $\mu^{\mathrm{sc}} < 0$  for $\beta < \beta_{\mathrm{c}}$. At $\beta = \beta_{\mathrm{c}}$ we   have $g^{\mathrm{sc}} = 0 = \mu^{\mathrm{sc}}$.  
\item Let 
\begin{equation}
	\Xi = \frac{\beta_0}{24 \pi^3}  \int_{\mathbb{R}^6} \gamma_{0}^2(p,x) \exp\left( \beta_0 \left( p^2 + \frac{\omega^2 x^2}{4} \right) \right) \left( v\ast \varrho_{0}(0) - v \ast \varrho_{0}(x) \right) \de(p,x)  > 0,
	\label{eq:semiclassicalfunctional29}
	\end{equation}
	where $\gamma_0$ denotes the minimizer of \eqref{eq:semic-lassicalfunctional3i} for $v\equiv 0$ at the ideal gas critical inverse temperature $\beta_0$ in \eqref{eq:idealgascriticaltemp2}, with $\rho_0$ its density. The
 inverse critical temperature satisfies 
	\begin{equation}
	{\beta}_{\mathrm{c}}(\lambda) = \beta_0\left( 1 + \lambda \Xi + O\left(\lambda^2\right)\right)   \label{eq:crittempshift}
	\end{equation}
	as $\lambda\to 0$. 
	\end{enumerate}
\end{proposition}  

In particular, since $\Xi > 0$, the inverse critical temperature increases for small $\lambda$ due to the repulsive interactions. The change $\Xi$ is known as the mean-field shift (see  \cite{PitaevskiiStringari}), and is due to the decrease in particle density at the center of the trap.

\subsubsection*{Remarks}
\begin{enumerate}
	\item The constants in the bounds in Theorem~\ref{thm:freeenergyscaling2}, Theorem~\ref{thm:limitHartreetheory} and Corollary~\ref{cor:strongBEC} are uniform in $\beta \omega \gtrsim 1$, and therefore allow for a zero temperature limit. In particular, \eqref{eq:mainresults21c} shows that the 1-pdm of an approximate ground state of $H_N$ approaches the projection onto the minimizer of the semiclassical Hartree energy functional 
	\begin{equation}
		\mathcal{E}^{\mathrm{H}}(\Phi) = \left\langle \Phi, \left( - \hbar^2 \Delta + \frac{\omega^2 x^2}{4} \right) \Phi \right\rangle + \frac 12 \int_{\mathbb{R}^6} | \Phi(x) |^2 v(x-y) | \Phi(y) |^2 \de(x,y)
		\label{eq:semiclassicalHartree}
	\end{equation}
	among all functions with $\Vert \Phi \Vert = 1$. In particular, there is complete BEC in the ground state. 	Due to the factor $\hbar^2 = N^{-2/3}$ in front of the Laplacian in \eqref{eq:semiclassicalHartree}, 
	the density of the minimizer of $\mathcal{E}^{\mathrm{H}}$ is supported  on the  length scale $\omega^{-1/2} \hbar^{1/2}$ and converges to a delta function as $\hbar \to 0$.
	\item As the other statements, the proof of Corollary~\ref{cor:weakBEC} is based on quantitative estimates. That is, 	we show an explicit rate for the convergence of the terms in \eqref{eq:83} for fixed $R>0$ as $N \to \infty$. The condition that $\beta \omega >0 $ is fixed could be replaced by  the requirement that it approaches a positive limiting value as $N \to \infty$. 
	\item \label{rem:gc} The results in Theorem~\ref{thm:freeenergyscaling2}, Theorem~\ref{thm:limitHartreetheory}, Corollary~\ref{cor:weakBEC} and Corollary~\ref{cor:strongBEC} remain true if the canonical free energy $F^{\mathrm{c}}(\beta,N,\omega)$ is replaced by its grand-canonical analogue (defined in Section~\ref{sec:sq} below), 
	and if  approximate minimizers in $\mathcal{S}_N^{\mathrm{c}}$ are replaced by  approximate minimizer among grand-canonical states. This, in particular, shows that the canonical and the grand-canonical  free energies 	agree within our accuracy, and the same is true for the 1-pdms of approximate minimizers. 
	\item \label{rem:ex} The result for the free energy in Theorem~\ref{thm:freeenergyscaling2} is optimal in the sense that the order of magnitude $\omega N^{1/3}$ of the remainder equals the order of magnitude of the exchange term, which would be present in Hartree-Fock theory but is absent  in Hartree theory. 
	\item For $\beta \omega \sim 1$ all terms in the Gibbs free energy functional contribute to the free energy at the order $\omega N$. As Theorem~\ref{thm:freeenergyscaling2} and Theorem~\ref{thm:limitHartreetheory} show, the leading order behavior of the  free energy can be described by the semiclassical free energy functional. Because the condensate density converges to a delta distribution on the relevant length scale, the leading order contribution of the interaction is felt by the condensate as an effective chemical potential. At the order $\hbar \omega N = \omega N^{2/3}$ we expect a contribution from the semiclassical Hartree energy functional in \eqref{eq:semiclassicalHartree} describing the local energy of the condensate. This contribution is captured by the Hartree free energy but not by the semiclassical free energy. 
	\item \label{rem6} The assumption \eqref{eq:condintpot} on the interaction potential in the above statements is crucial since it guarantees that the condensate is located around the center of the harmonic trap. In the Hartree free energy functional the condensate sees an effective potential from interactions with the thermal cloud. If its curvature exceeds the one of the harmonic trap the total potential for the condensate becomes concave. In this case we expect the condensate to localize in the vicinity of some sphere around the origin. If this is true the semiclassical free energy functional $\mathcal{F}^{\mathrm{sc}}$ \eqref{eq:semic-lassicalfunctional3i} is certainly not the correct effective theory because it assumes that the condensate is sitting in the center of the trap. Hence Theorem~\ref{thm:limitHartreetheory} and the Corollaries cannot be expected to hold in this case. Moreover, in this situation the resulting effective potential in the Hartree operator $h^{\mathrm{H}} = h + N^{-1} v \ast \varrho_{\gamma^{\mathrm{H}}}(x)$ will have a Mexican hat like shape  
	and its low lying excitation spectrum will consist of waves moving in the valley of the Mexican hat. The effective potential lives on the length scale $\omega^{-1/2}$, and hence the spectral gap above the ground state of $h^{\mathrm{H}}$ will be of the order $\omega \hbar^2 = \omega N^{-2/3}$. The first excited state will accordingly have an expected occupation of the order $N^{2/3}$. This should be compared to the case when \eqref{eq:condintpot} holds. Here the spectral gap is of the order $\hbar \omega$ (the effective potential resembles a harmonic oscillator close to the origin) and the first excited state has an expected occupation of the order $N^{1/3}$. The higher occupation number of the first excited state would imply that the exchange term is of the order $\omega N^{2/3}$ instead of $\omega N^{1/3}$, compare with Remark~\ref{rem:ex}. Accordingly, also Theorem~\ref{thm:freeenergyscaling2} would not be true in the form it is stated if \eqref{eq:condintpot} does not hold. 
	\item The statement in Corollary~\ref{cor:strongBEC} remains true if the spectral projections $P$ and $Q$ related to $\gamma_N$ are replaced by $P^{\mathrm{H}}$ and $Q^{\mathrm{H}}$, where $P^{\mathrm{H}}$ denotes the spectral projection onto the subspace related to the largest eigenvalue of $\gamma^{\mathrm{H}}$ and $Q^{\mathrm{H}} = 1 - P^{\mathrm{H}}$. In particular, the condensate wave function of $\gamma_N$ and that of $\gamma^{\mathrm{H}}$ are equal within our accuracy. Using this relation, one can show that the bounds in Corollary~\ref{cor:strongBEC} imply \eqref{eq:83b} with explicit rates if continuously differentiable functions $f : \mathbb{R}^6 \to \mathbb{C}$ with a bounded derivative are considered.
	\item The techniques used to prove Theorem~\ref{thm:freeenergyscaling2}, Theorem~\ref{thm:limitHartreetheory}, Corollary~\ref{cor:weakBEC} and Corollary~\ref{cor:strongBEC} carry over with moderate adjustments to the case of trapping potentials behaving as $|x|^s$ with some $s > 0$ for large $|x|$. The main point is that such potentials still lead to an asymptotic power law behavior of the eigenvalues of the related Schrödinger operator and cause a separation of length scales between the condensate and the thermal cloud. It should be noted that this difference is more pronounced for $s < 2$ and less pronounced if $s > 2$. If hard walls are considered (formally $s = + \infty$) the condensate and the thermal cloud live on the same length scale and the critical temperature of the interacting model equals that of the ideal gas to leading order. 
	\item Theorem~\ref{thm:freeenergyscaling2} shows that the 1-pdm of any approximate minimizer $\Gamma_N$ of the Gibbs free energy functional $\mathcal{F}$ in the sense of \eqref{eq:mainresults21b} is close to the minimizer $\gamma^{\mathrm{H}}$ of the Hartree free energy functional in trace norm. It is an interesting open question to prove a similar result for the reduced $k$-particle density matrix of $\Gamma_N$ when $k \geq 2$.  
	\item We conclude our discussion with a brief comment on the similarities and the differences between our results and the results in \cite{NarnThi1981}. In this reference the authors consider the same set-up as we do here but with the explicit choice $v(x) = 1/|x|$ for the interaction potential (note that it does not satisfy Assumption~\ref{as:regularitypotential}). They are interested in the leading order asymptotics of the grand-canonical pressure and show that it convergences to the maximum of a semiclassical pressure functional. Although they do not provide rates for this convergence, their methods are explicit and rates could be extracted. In a second step they consider the canonical free energy and show that its Legendre transform converges to the same limit as the grand-canonical pressure, which proves the equivalence of ensembles. Statements about the relevant Gibbs state or about approximate maximizers of the quantum mechanical pressure functional are not provided. Their choice of interaction potential causes the condensate in the semiclassical pressure functional to be located on the surface of a ball with strictly positive radius around the origin. This should be compared to remark~6. The relation between their semiclassical pressure functional and our semiclassical free energy functional can be understood as follows: When we reformulate their semiclassical pressure functional in a natural way as a free energy functional, then the resulting functional is related to our semiclassical free energy functional via Lemma~\ref{lem:secondvariationalcharacterization} below. The main difference between the results in \cite{NarnThi1981} and our work concerns the fact that we are interested in characterizing the BEC phase transition and in showing how the critical temperature depends on the interaction potential. To achieve these goals we require much more elaborate techniques than the ones needed to compute the leading order asymptotics of the quantum mechanical free energy or pressure.
\end{enumerate}
\subsection{Accuracy of Hartree theory in the mean-field limit}
\label{sec:MFresults}
The semiclassical mean-field limit is a natural parameter regime for the trapped Bose gas because all  terms of the Gibbs free energy functional, that is, the energy related to $h$, the interaction energy and $1/\beta$ times the entropy, are of the same order in $N$. There is, however, another interesting parameter regime for the trapped Bose gas, whose relevance stems from the fact that as the temperature goes to zero, one recovers the Hartree energy functional (without a semiclassical parameter) in case of a mean-field scaling, and the Gross-Pitaevskii (GP) energy functional if an appropriate dilute limit is considered. Both models have been investigated in detail in the literature, and we refer to the introduction for more details and for references. 
In this subsection, we shall investigate this second parameter regime. The results discussed here are independent from the ones in the previous subsection, but share some similarity in their proofs. For simplicity, we shall use the same notation for the relevant objects as above, even though these are really  different from before. 

The relevant one-particle Hamiltonian describing the above parameter regime is given by \eqref{eq:Hamiltonian1stqu} with the choice $\hbar = 1$, that is,
\begin{equation}
	h = -\Delta + \frac{\omega^2 x^2}{4}
	\label{eq:h2}
\end{equation}
and the $N$-particle Hamiltonian reads
\begin{equation}
	H_N = \sum_{i=1}^N h_i + \sum_{1 \leq i < j \leq N} v_N(x_i-x_j) \,,
	\label{eq:H_N2}
\end{equation}
where now 
\begin{equation}
v_N(x) = N^{-1+3 \kappa} v( N^{\kappa} x) 
\label{eq:Interactionpotential1}
\end{equation}
for some $0 \leq \kappa \leq 1$. For $\kappa = 0$ this is the mean-field Hamiltonian for a trapped Bose gas and with increasing values of $\kappa$ the interaction potential $v_N$ becomes stronger and shorter ranged. For $\kappa = 1$ we obtain the dilute GP scaling. In this work we will be concerned with rather small values of $\kappa$ ($\kappa \leq 1/6$), which is why we refer to the scaling in \eqref{eq:Interactionpotential1} as the MF scaling. As $N \to \infty$ and for $\kappa = 0$, the ground state energy of $H_N$ divided by $N$ converges to the minimum of the Hartree energy functional $\mathcal{E}^{\mathrm{H}}$ in \eqref{eq:semiclassicalHartree} with $\hbar = 1$, and the 1-pdm of any approximate ground state of $H_N$ converges in trace norm to the minimizer of $\mathcal{E}^{\mathrm{H}}$, see e.g. \cite{LewinNamRougerie20141}. In case of $\kappa = 1$ the relevant limiting theory is the GP energy functional and comparable statements than in the MF scaling hold, see e.g. \cite{Themathematicsofthebosegas}.
 
 We shall again consider inverse temperatures of the order of the inverse critical temperature  for BEC in the ideal Bose gas, $\beta_{0,N} = \omega^{-1} (\zeta(3)/ N)^{1/3}$, which now means that $\beta \omega \sim N^{-1/3}$. 
 For $\kappa=1$, this is the regime considered in \cite{me}. Choosing $\kappa$ smaller will allow us to extend these results and compute the free energy, as well as the 1-pdm of (approximate) Gibbs states, with greater accuracy. 
 
 The Gibbs free energy functional of a state $\Gamma \in \mathcal{S}_N^{\mathrm{c}}$ and the canonical free energy for the Hamiltonian $H_N$ in \eqref{eq:H_N2} are denoted by
 \begin{equation}
 	\mathcal{F}(\Gamma) = \tr\left[ H_N \Gamma \right] - \frac{1}{\beta} S(\Gamma) \quad \text{ and } \quad F^{\mathrm{c}}(\beta,N\omega) = \inf_{\Gamma \in \mathcal{S}_N^{\mathrm{c}} } \mathcal{F}(\Gamma),
 	\label{eq:freeenergy2}
 \end{equation}
 respectively. The relevant Hartree free energy functional and the Hartree free energy are now 
 \begin{equation}
 	\mathcal{F}^{\mathrm{H}}(\gamma) = \trs\left[ \left(h + \frac{1}{2} v_N \ast \varrho_{\gamma} \right) \gamma \right] - \frac{1}{\beta} s(\gamma) \quad \text{ and } \quad F^{\mathrm{H}}(\beta,N,\omega) = \inf_{\gamma \in \mathcal{D}_N^{\mathrm{H}}} \mathcal{F}^{\mathrm{H}}(\gamma) 
 	\label{Hartreefreeenergy2}
 \end{equation}
with $h$ defined in \eqref{eq:h2}, $v_N$ in \eqref{eq:Interactionpotential1}, $\mathcal{D}_N^{\mathrm{H}}$ in \eqref{eq:formdomainH} and  the bosonic entropy $s(\gamma)$  in \eqref{eq:bosonicentropy}. The main result in this section is the following Theorem.

\begin{theorem}
	\label{thm:freeenergyscaling1}
	Let $v : \mathbb{R}^3 \to \mathbb{R}_+ \cup \{ 0 \}$ be a function in $L^1(\mathbb{R}^3)$ with $v(-x) = v(x)$, $\hat{v} \in L^1(\mathbb{R}^3)$ and $\hat{v} \geq 0$. Let $v_N$ be given by \eqref{eq:Interactionpotential1} with $0 \leq \kappa \leq 1/6$. In the limit $N \to \infty$ with $\beta \omega \gtrsim N^{-1/3}$ we have
	\begin{align}
	\left| F^{\mathrm{c}}(\beta,N,\omega) - F^{\mathrm{H}}(\beta,N,\omega) \right| \lesssim \omega N^{1/3} \left( N^{\kappa} + \ln N \right).
	\label{eq:mainresults1}
	\end{align}
	Moreover, for any sequence of states $\Gamma_N \in \mathcal{S}^{\mathrm{c}}_N$ with 1-pdm $\gamma_N$ and
	\begin{equation}
	\left| \mathcal{F}(\Gamma_N) - F^{\mathrm{H}}(\beta,N,\omega) \right| \lesssim\omega \delta \label{eq:mainresults1b}
	\end{equation}
	for some $\delta > 0$, we have
	\begin{equation}
	\left\Vert \gamma_{N} - \gamma^{\mathrm{H}} \right\Vert_1 \lesssim N^{5/6} \left( N^{\kappa} + \ln N\right)^{1/4}  + N^{3/4} \delta^{1/4}.
	\label{eq:mainresults1c}
	\end{equation} 
	Here $\gamma^{\mathrm{H}}$ denotes the unique minimizer of $\mathcal{F}^{\mathrm{H}}$ in \eqref{Hartreefreeenergy2}. 
\end{theorem}

\subsubsection*{Remarks}
\begin{enumerate}
	\item For $\beta \sim \beta_{0,N}$ we have $F^{\mathrm{H}}(\beta,N,\omega) \sim \omega (\beta \omega)^{-4} \sim \omega N^{4/3}$. The accuracy of \eqref{eq:mainresults1} allows us to describe the interaction between all particles in the system. While the interaction among particles inside the condensate is of order $\omega N$, all interactions among particles in the thermal cloud are only of the order $\omega N^{1/2}$. This is because the density of the thermal cloud lives on the length scale $\omega^{-1/2} N^{1/6}$ (while the length scale of the condensate wave function is $\omega^{-1/2}$), compare with the discussion in Section~\ref{sec:model}. For $\kappa = 1/6$ the exchange term and the direct interaction energy in the thermal cloud are of the same order, which is why we choose to state Theorem~\ref{thm:limitHartreetheory} only for the range $0 \leq \kappa \leq 1/6$. Theorem~\ref{thm:limitHartreetheory} should be contrasted with the result in \cite{me} in the GP limit, where the interaction can be seen only in the condensate. 
	\item The result for the free energy in Theorem~\ref{thm:freeenergyscaling1} is optimal in the following sense: (a) The size of the remainder in \eqref{eq:mainresults1} equals for $0 < \kappa \leq 1/6$ the size of the exchange term, which is not included in the Hartree free energy. To avoid this contribution one would need to consider Hartree-Fock theory instead.
	(b) In the Hartree free energy functional the condensate is effectively described by a quasi-free state. Due to the large fluctuations of the number of particles in the condensate (they are of order $N$ for a quasi-free state), the entropy of the condensate is of order $\ln N$. This leads to a contribution to the free energy of the order $\beta^{-1} \ln N$, which is the order of magnitude of the remainder if $\kappa = 0$. To avoid this contribution one needs to describe the condensate by a coherent state, see e.g. \cite{NRS18}.
	\item Similarly to Remark~\ref{rem:gc} in the previous section, the result in Theorem~\ref{thm:freeenergyscaling1} remains true if  $F^{\mathrm{c}}(\beta,N,\omega)$ is replaced by its grand-canonical analogue, and the same holds for  approximate minimizers of the Gibbs free energy functional.  Theorem~\ref{thm:freeenergyscaling1} therefore shows that, within our accuracy, the  free energies of the canonical and of the grand-canonical ensemble are equal, and likewise for  the 1-pdms of approximate Gibbs states.  The entropy related to the fluctuation of the number of particles in the grand-canonical ensemble is of order $\ln N$ and leads to a contribution to the free energy of the order $\beta^{-1} \ln N$. The canonical ensemble does not have such a contribution, and we therefore expect that the free energies in the two ensembles indeed differ by a term of this order.
	\item Remark~8 in the previous section applies similarly to Theorem~\ref{thm:freeenergyscaling1}.
\end{enumerate}

\subsection{Proof strategy and organization of the article}
\label{sec:proofstrategy}
For the convenience of the reader we give here a short summary of the organization of the paper and the proof of the statements in the previous  two subsections. 

In Section~\ref{sec:PropertiesHartreefreeenergyfunctional} we study the Hartree free energy functional $\mathcal{F}^{\mathrm{H}}$ in \eqref{eq:Hartrees1pfreeenergyfkt}. We also 
introduce a canonical version of this functional and bound the difference in the corresponding free energies. 
This bound will later allow us to show that the canonical and the grand-canonical free energies of the full quantum model agree within the desired accuracy. 

The semiclassical free energy functional 
$\mathcal{F}^{\mathrm{sc}}$ in \eqref{eq:semic-lassicalfunctional3i} 
is studied in Section~\ref{sec:semiclassicalfreeenergyfunctional}. We prove the existence of a unique minimizer and derive the corresponding Euler--Lagrange equation. With these preparations at hand we shall give the proof of Proposition~\ref{prop:crittempscmf}.

In Section~\ref{sec:semiclassicalfreeenergyfunctionalb} we study the Hartree free energy functional in the semiclassical MF limit and prove Theorem~\ref{thm:limitHartreetheory}. To relate the Husimi function of the Hartree minimizer to the minimizing phase space density of the semiclassical free energy functional, we prove a lower bound for a phase space version of the bosonic relative entropy to quantify its coercivity. This bound is analogous to the one used for the same problem in the case of density matrices in \cite{me,me2}. 

In Section~\ref{sec:freeenergybounds} we derive upper and lower bounds for the  free energy $F^{\mathrm{c}}$ in \eqref{eq:freeenergycan} showing that it can be approximated with good precision by the Hartree free energy. For the lower bound this can be done with a standard inequality for interaction potentials of positive type, while for the  upper bound we need to estimate the size of the exchange terms. This proves the first statement in Theorem~\ref{thm:freeenergyscaling2}. 

The second statement in Theorem~\ref{thm:freeenergyscaling2} concerning estimates for the 1-pdm of approximate minimizers of the Gibbs free energy functional is proved in Section~\ref{sec:boundson1pdm}. The analysis is based on the free energy bounds in Section~\ref{sec:freeenergybounds} and on an inequality for the bosonic relative entropy  proved in \cite[Lemma~4.1]{me,me2}. The proofs of Corollary~\eqref{cor:weakBEC} and Corollary~\eqref{cor:strongBEC} are a consequence of these bounds and of Theorem~\ref{thm:limitHartreetheory}, and are also given in this section. 

In Section~\ref{sec:lastsection} we shall explain the necessary modifications of the analysis in Sections~\ref{sec:freeenergybounds} and~\ref{sec:boundson1pdm} in order to prove Theorem~\ref{thm:freeenergyscaling1}. The main difference lies in the analysis of the spectral gap of the Hartree operator, which in the case of the semiclassical mean-field limit is guaranteed by the assumption \eqref{eq:condintpot} on the Hessian of the interaction potential, but needs a separate proof here because of the different scaling in the mean-field limit.

\section{The Hartree free energy functional}
\label{sec:PropertiesHartreefreeenergyfunctional}
At several points in the paper it will be convenient to use the second quantized formalism and we start by introducing the relevant notation. Afterwards, we define a canonical version of the Hartree free energy functional, which plays an important role in the proof of an upper bound for the  free energy $F^{\mathrm{c}}$ in Section~\ref{sec:proofupperboundc}. We prove several statements for the two versions of the Hartree free energy functional that are used during the proof of the main results,  e.g. the existence of a unique minimizer. In the last part of this section we bound the difference of the canonical and the grand-canonical Hartree free energies. This bound will allow us to show in Section~\ref{sec:freeenergybounds} that the  canonical and the interacting grand-canonical free energies of the full quantum model agree within the desired accuracy.
\subsection{Second quantized formalism}\label{sec:sq}
Let $\mathscr{F}$ denote the bosonic Fock space over the one-particle Hilbert space $L^2(\mathbb{R}^3)$. The second quantization of a one-particle operator $A$, which is an operator on $\mathscr{F}$, is denoted by $\de \Upsilon(A) = 0 \bigoplus_{M=1}^{\infty} \sum_{i=1}^M A_i$, where $A_i$ is the operator acting as $A$ on the $i$-th particle, and as the identity on the others. Similarly, the second quantized version of the interaction potential is given by
\begin{equation}
\mathcal{V}_N = 0 \oplus 0 \bigoplus_{M=2}^{\infty} \sum_{1 \leq i < j \leq M} v_N(x_i - x_j).
\label{eq:secondquantizedinteraction}
\end{equation}
Here $v_N$ equals either by $N^{-1} v$ in \eqref{eq:Hamiltonian1stqu} or $v_N$ in \eqref{eq:Interactionpotential1}. We emphasize that the interaction potential is given by $v_N$ in all $M$-particle sectors of the Fock space, and not by $v_M$. The second quantized equivalent of the $N$-particle Hamiltonians \eqref{eq:Hamiltonian1stqu} and \eqref{eq:H_N2} acting on the bosonic Fock space thus reads
\begin{equation}
\mathcal{H} = \de \Upsilon(h) + \mathcal{V}_N
\label{eq:Hamiltonian2ndqu}
\end{equation}
with $h$ in \eqref{eq:harmonicoscillator} or in \eqref{eq:h2}. By 
\begin{equation}
\mathcal{S}^{\mathrm{gc}}_N = \left\{ \Gamma \in \mathcal{B}\left( \mathscr{F} \right) \ \Big| \ 0 \leq \Gamma \leq 1, \ \tr [ \Gamma ]  = 1, \tr[ \mathcal{N} \Gamma ] = N, \ \tr\left[ \de \Upsilon(h) \Gamma \right] < + \infty \right\}
\label{eq:grandcanonicalstates}
\end{equation}
we denote the set of bosonic states on $\mathscr{F}$ with an expected number of $N$ particles. Here $\mathcal{N} = \bigoplus_{M=0}^{\infty} M$ denotes the particle number operator. By the definitions of $\mathcal{S}_N^{\mathrm{c}}$ and $\mathcal{S}^{\mathrm{gc}}_N$ in \eqref{eq:canonicalstates} and \eqref{eq:grandcanonicalstates}, and the fact that $L^2_{\mathrm{sym}}(\mathbb{R}^{3N} )$ equals the $N$-particle sector of $\mathscr{F}$, we have $\mathcal{S}^{\mathrm{c}}_N \subset \mathcal{S}^{\mathrm{gc}}_N$ for $N \in \mathbb{N}$. For states $\Gamma \in \mathcal{S}^{\mathrm{gc}}_N$ the Gibbs free energy functional is defined by 
\begin{equation}
\mathcal{F}(\Gamma) = \tr\left[ \mathcal{H} \Gamma \right] - \frac{1}{\beta} S(\Gamma),
\end{equation}
where the trace in the first term and the one in the definition of the entropy are now over $\mathscr{F}$. The grand-canonical free energy is given by
\begin{equation}
F^{\mathrm{gc}}(\beta,N,\omega) = \inf_{\Gamma \in \mathcal{S}^{\mathrm{gc}}_N} \mathcal{F}(\Gamma) = -\frac{1}{\beta} \ln \left( \tr \exp\left( - \beta \left( \mathcal{H} - \mu \mathcal{N} \right) \right) \right) + \mu N,
\label{eq:freeenergygrancan}
\end{equation}
where the chemical potential $\mu$ is chosen such that the grand-canonical Gibbs state 
\begin{equation}
G_N^{\mathrm{gc}} = \frac{e^{-\beta \left( \mathcal{H} - \mu \mathcal{N} \right)}}{\tr \left[ e^{-\beta \left( \mathcal{H} - \mu \mathcal{N} \right)} \right]}
\label{eq:Gibbsstategrancan}
\end{equation}
has an expected number of $N$ particles. It is the unique minimizer of $\mathcal{F}$ when the minimization is performed over states in $\mathcal{S}^{\mathrm{gc}}_N$. 

\subsection{A canonical version of the Hartree free energy functional}
\label{sec:Hartreefreeenergyfunctionals}
For a state $\Gamma \in \mathcal{S}^{\mathrm{c}}_N$ we define the canonical Hartree free energy functional by
\begin{equation}
\mathcal{F}^{\mathrm{H},\mathrm{c}}(\Gamma) = \tr \left[ \de \Upsilon \left( h + \frac{1}{2} v_N \ast \varrho_{\Gamma} \right) \Gamma \right] - \frac{1}{\beta} S(\Gamma).
\label{eq:Hartreefunctional}
\end{equation}
Here $h$ and $v_N$ are either given by \eqref{eq:harmonicoscillator} and $N^{-1} v$ or by \eqref{eq:h2} and \eqref{eq:Interactionpotential1}. The corresponding canonical Hartree free energy 
is given by
\begin{equation}
F^{\mathrm{H},\mathrm{c}}(\beta,N,\omega) = \inf_{\Gamma \in \mathcal{S}^{\mathrm{c}}_N} \mathcal{F}^{\mathrm{H},\mathrm{c}}(\Gamma). 
\label{eq:Hartreeenergyc}
\end{equation}
In the next subsection we will see that our assumptions on $v$ imply the existence of a unique minimizer for the above minimization problem. The minimizer satisfies the Euler--Lagrange equation
\begin{equation}
G^{\mathrm{H},\mathrm{c}} = \frac{e^{-\beta \sum_{i=1}^N \left( h_i + v_N \ast \varrho_{G^{\mathrm{H},\mathrm{c}} }(x_i) \right) }}{\tr \left[ e^{-\beta \sum_{i=1}^N \left( h_i + v_N \ast \varrho_{G^{\mathrm{H},\mathrm{c}} }(x_i) \right) } \right]},
\label{eq:HartreeGibbsstatec}
\end{equation}
which is a self-consistent equation for the state $G^{\mathrm{H},\mathrm{c}}$ since its density $\varrho_{G^{\mathrm{H},\mathrm{c}} }(x)$ appears on the right-hand side. The canonical Hartree free energy can be written in terms of the unique minimizer $G^{\mathrm{H},\mathrm{c}}$ in  \eqref{eq:HartreeGibbsstatec} as
\begin{align}
F^{\mathrm{H},\mathrm{c}}(\beta,N,\omega) = - \frac{1}{\beta} \ln\left( \tr \exp\left( -\beta \sum_{i=1}^N \left( h_i + v_N \ast \varrho_{G^{\mathrm{H},\mathrm{c}}}(x_i) \right) \right) \right) - D_N\left(\varrho_{G^{\mathrm{H},\mathrm{c}}} ,\varrho_{G^{\mathrm{H},\mathrm{c}}} \right).  \label{eq:Hartreeenergyc2}
\end{align}
Here and in the following we denote
\begin{equation}
D_N(f,g) = \frac{1}{2} \int_{\mathbb{R}^6} v_N(x-y) f(x) g(y) \de (x,y),
\label{eq:abbreviationinteraction}
\end{equation}
and by a slight abuse of notation we use the symbol $D_N(\mu,\nu)$ also for the natural extension of the above definition to finite Borel measures $\mu$ and $\nu$ on $\mathbb{R}^3$.  In case of $N=1$ we simply write $D$ instead of $D_1$. 
The 1-pdm of $G^{\mathrm{H},\mathrm{c}}$ will be denoted by $\gamma^{\mathrm{H},\mathrm{c}}$.

Finally let us mention that minimization of $\mathcal{F}^{\mathrm{H},\mathrm{c}}$ over states in $\mathcal{S}_N^{\mathrm{gc}}$ in \eqref{eq:grandcanonicalstates} yields the Hartree free energy $F^{\mathrm{H}}(\beta,N,\omega)$ in \eqref{eq1pHartreefreeenergy}. To see this, we note that $S(\Gamma) \leq s(\gamma_{\Gamma})$ holds for any state $\Gamma \in \mathcal{S}_N^{\mathrm{gc}}$, see \cite[Chapter~2.5.14.5]{Thirring_4}, with 
 the von Neumann entropy $S$ in \eqref{eq:gibbsfreeenergyfunctional} (the trace is taken over  Fock space) and the bosonic entropy $s$ in \eqref{eq:bosonicentropy}. This proves 
\begin{equation}
	\inf_{\Gamma \in \mathcal{S}_N^{\mathrm{gc}}} \mathcal{F}^{\mathrm{H},\mathrm{c}}(\Gamma) \geq F^{\mathrm{H}}(\beta,N,\omega).
	\label{eq:154}
\end{equation}
The reverse inequality follows if we use the unique quasi-free state with 1-pdm  $\gamma^{\mathrm{H}}$ as a trial state. This state is also the unique minimizer of $\mathcal{F}^{\mathrm{H},\mathrm{c}}(\Gamma)$ when minimized over the set $\mathcal{S}_N^{\mathrm{gc}}$. 

\subsection{Existence of a unique minimizer and Euler--Lagrange equation}
\label{sec:Hartreeexmin}
The following three statements concern the existence of a unique minimizer of the Hartree free energy functional in the canonical and in the grand-canonical setting, and the justification of the corresponding Euler--Lagrange equations. 
Since $N$ is a fixed parameter here, the precise form of $v_N$ is not important. We shall only need that $v_N, \hat{v}_N \in L^1(\mathbb{R}^3)$ and that $v_N, \hat{v}_N \geq 0$ holds, which is guaranteed by our assumptions. 
 The first statement concerns the Hartree free energy functional $\mathcal{F}^{\mathrm{H}}$.
\begin{lemma}
	\label{lem:minimizersHartree}
	The Hartree free energy functional $\mathcal{F}^{\mathrm{H}}$ in \eqref{eq:Hartrees1pfreeenergyfkt} admits a unique minimizer in the set $\mathcal{D}^{\mathrm{H}}_N$ defined in \eqref{eq:formdomainH} and the minimizer solves the Euler--Lagrange equation in \eqref{eq:semiclassicalfunctional5i}. The Hartree free energy can be expressed in terms of the minimizer $\gamma^{\mathrm{H}}$ as
	\begin{equation}
		F^{\mathrm{H}}(\beta,N,\omega) = \frac{1}{\beta} \trs\left[ \ln\left( 1 - e^{-\beta \left( h + v_N \ast \varrho^{\mathrm{H}} - \mu \right)} \right) \right] + \mu N - D_N\left( \varrho^{\mathrm{H}},\varrho^{\mathrm{H}} \right),
		\label{eq:32}
	\end{equation}
	where $\varrho^{\mathrm{H}}(x) = \gamma^{\mathrm{H}}(x,x)$ and with $D_N$ defined in \eqref{eq:abbreviationinteraction}.
\end{lemma}
\begin{proof}
	A proof of the first two statements can be found in \cite[Lemma~3.2]{LewinNamRougerie2018}. Eq.~\eqref{eq:32} is a direct consequence of the definition of $\mathcal{F}^{\mathrm{H}}$ and of the Euler--Lagrange equation in \eqref{eq:semiclassicalfunctional5i}.
\end{proof}
The next two statements concern the canonical version of the Hartree free energy functional. 
\begin{lemma}
	\label{prop:existenceminimizersHartreefunctional}
	The canonical Hartree free energy functional $\mathcal{F}^{\mathrm{H},\mathrm{c}}$ in \eqref{eq:Hartreefunctional} admits a unique minimizer in the set $\mathcal{S}_N^{\mathrm{c}}$ defined in \eqref{eq:canonicalstates}. 
\end{lemma}
\begin{proof}
	Since $v_N\geq 0$, $\mathcal{F}^{\mathrm{H},\mathrm{c}}$ is easily seen to be bounded from below. 
	Let $\{ \Gamma_n \}_{n=1}^{\infty}$ be a minimizing sequence for $\mathcal{F}^{\mathrm{H},\mathrm{c}}$ in $\mathcal{S}^{\mathrm{c}}_N$. Since $\Gamma_n$ is a sequence of states there exists a state $\Gamma$ on the truncated Fock space 
	\begin{equation}
		\mathscr{F}^{\leq N} = \mathbb{C} \oplus \bigoplus_{m=1}^{N} \mathcal{H}_m, \quad \text{ where } \quad \mathcal{H}_m = L^2_{\mathrm{sym}}\left( \mathbb{R}^{3m} \right)
	\end{equation}
	such that $\Gamma_n$ converges to $\Gamma$ in the geometric topology,   
	see \cite[Lemmas~3 \& 4]{Lewin2010}. This means that every $k$-particle reduced density matrix (k-pdm) of $\Gamma_n$ with $k \geq 0$ converges in the weak operator topology to the k-pdm of $\Gamma$ \cite[Definition~1]{Lewin2010}. 
	
	Let $\gamma_n$ and $\gamma$ be the 1-pdm of $\Gamma_n$ and $\Gamma$, respectively. As far as density matrices are considered, convergence in the weak operator topology implies convergence in the weak-$*$ topology of the trace class. In combination with the fact that $\trs[ h \gamma_{\Gamma_n}]$ is uniformly bounded in $n$ and that 
	 $h$ has a compact resolvent, this shows $\trs[\gamma_n] \to \trs[\gamma]$. Using \cite[Lemma~4]{Lewin2010}, we conclude that $\Gamma_n \to \Gamma$ in trace-norm and that $\Gamma \in \mathcal{S}^{\mathrm{c}}_N$.
		
	Next we will show that $\mathcal{F}^{\mathrm{H},\mathrm{c}}$ is lower semicontinuous along the sequence $\{ \Gamma_n \}_{n=1}^{\infty}$. We can write
	\begin{equation}
	\tr \left[ \sum_{i=1}^N h_i \Gamma_n \right] - \frac{1}{\beta} S(\Gamma_n) = \frac{1}{\beta} S(\Gamma_n,G) - \frac{1}{\beta} \ln\left( \tr \exp\left( -\beta \sum_{i=1}^N h_i \right) \right),
	\label{eq:proofsHartreefunctional7}
	\end{equation}
	with the relative entropy $S(\Gamma,G) = \tr\left[ \Gamma \left( \ln \Gamma - \ln G \right) \right]$ and the canonical Gibbs state of the ideal gas $G=\exp(-\beta( \sum_{i=1}^N h_i ))/\tr[\exp(-\beta( \sum_{i=1}^N h_i ))]$. The relative entropy is known to be lower semicontinuous w.r.t. the trace-class topology, see e.g. \cite[2.2.20]{Thirring_4} or \cite[Corollary~5.12]{OhyaPetz1993}, and hence the same is true for the left-hand side of \eqref{eq:proofsHartreefunctional7}. It remains to consider the continuity properties of the  interaction term $\trs[v_N \ast \varrho_{\Gamma_n} \gamma_n] = 2 D_N(\varrho_{\Gamma_n},\varrho_{\Gamma_n})$. Using that $\varrho_{\Gamma_n} \to \varrho_{\Gamma}$ in $L^{1}(\mathbb{R}^3)$, which follows from $\gamma_n \to \gamma$ in trace norm, see e.g. \cite[Lemma~5.1]{GriesHan2012}, and \cite[Theorem~2.7]{LiLo10}, we see that $\varrho_{\Gamma_n} \to \varrho_{\Gamma}$ pointwise almost everywhere, at least for a suitable subsequence. Fatou's Lemma and the assumption $v_N \geq 0$ therefore imply that 
	\begin{equation}
	\liminf_{n \to \infty} \int_{\mathbb{R}^6} v_N(x-y) \varrho_{\Gamma_n}(x) \varrho_{\Gamma_n}(y) \de(x,y) \geq \int_{\mathbb{R}^6} v_N(x-y) \varrho_{\Gamma}(x) \varrho_{\Gamma}(y) \de(x,y)
	\label{eq:proofsHartreefunctional8}
	\end{equation}
	along this subsequence. 
	In combination, these considerations  show that $\liminf_{n \to \infty} \mathcal{F}^{\mathrm{H},\mathrm{c}}(\Gamma_n) \geq \mathcal{F}^{\mathrm{H},\mathrm{c}}(\Gamma)$ and we conclude the existence of a minimizer $\Gamma \in \mathcal{S}^{\mathrm{c}}_N$. 
	
	Uniqueness of minimizers follows from the strict convexity of $\mathcal{F}^{\mathrm{H},\mathrm{c}}$, since $\hat v_N\geq 0$ and the map $\Gamma \mapsto S(\Gamma)$ is strictly concave. 
\end{proof}

The third statement establishes the Euler--Lagrange equation for $G_N^{\mathrm{H},\mathrm{c}}$ as well as the formula in \eqref{eq:Hartreeenergyc2} for the minimal free energy. Its proof follows standard arguments and is left to the reader.
\begin{lemma}
	\label{lem:eulerlagrange1}
	The unique minimizer of the canonical Hartree free energy functional $\mathcal{F}^{\mathrm{H},\mathrm{c}}$ in $\mathcal{S}_N^{\mathrm{c}}$ solves the self-consistent equation \eqref{eq:HartreeGibbsstatec}. Moreover, the minimal free energy is given by \eqref{eq:Hartreeenergyc2}.
\end{lemma} 

\subsection{A bound on the difference of the canonical and the grand-canonical Hartree free energies}
\label{sec:closenesseffectivefreeenergies}
In this Section we characterize the infimum of $\mathcal{F}^{\mathrm{H}}$ by an alternative variational principle, which we use afterwards to derive a bound for the difference of the canonical and grand-canonical Hartree free energies. As in the previous Section we assume $v_N, \hat{v}_N \in L^1(\mathbb{R}^3)$ and that $v_N, \hat{v}_N \geq 0$.
\begin{lemma}
	\label{lem:secondvariationalcharacterization}
	Define the set 
	\begin{equation}
	\mathcal{D}_N^{\mathrm{gc}} = \left\{ \left(\eta,\mu\right) \in  L^1 \left(\mathbb{R}^3\right) \times  \mathbb{R}  \ \bigg|  \ \trs \left[ \frac{1}{e^{\beta \left( h + v_N \ast \eta - \mu \right)} - 1} \right] = N \right\}. \label{eq:sec4secondlemmav21}
	\end{equation}
	We have
	\begin{equation}
	F^{\mathrm{H}}(\beta,N,\omega) = \sup_{(\eta,\mu) \in \mathcal{D}^{\mathrm{gc}}_N} \left\{ \frac{1}{\beta} \trs\left[ \ln\left( 1 - e^{-\beta \left( h + v_N \ast \eta - \mu \right)} \right) \right] + \mu N - D_N\left(\eta,\eta \right) \right\}. \label{eq:sec4secondlemmav22}
	\end{equation}
\end{lemma}
\begin{proof}
	For $\Gamma \in \mathcal{S}_N^{\mathrm{gc}}$ and $\eta \in L^1(\mathbb{R}^3)$, define the functional 
	\begin{equation}
	G(\Gamma,\eta) = \tr\left[ \de \Upsilon \left(h + v_N \ast \eta \right) \Gamma \right] - \frac{1}{\beta} S(\Gamma) - D_N(\eta,\eta).
	\label{eq:sec4secondlemmav23}
	\end{equation}
	Using $\hat{v} \geq 0$ we check that 
	\begin{equation}
		2 D_N (\eta, \varrho_{\Gamma} ) - D_N( \eta, \eta) \leq D_N( \varrho_{\Gamma}, \varrho_{\Gamma}),
	\end{equation}
	which implies
	\begin{equation}
	\sup_{\sigma \in L^1(\mathbb{R}^3)} G(\Gamma,\sigma) = \mathcal{F}^{\mathrm{H},\mathrm{c}}(\Gamma).
	\label{eq:sec4secondlemmav24}
	\end{equation}
	The left-hand side of \eqref{eq:sec4secondlemmav24} is bounded from below by $G(\Gamma,\eta)$ for any $\eta \in L^1(\mathbb{R}^3)$. We minimize this expression over $\Gamma \in \mathcal{S}_{N}^{\mathrm{gc}}$ and arrive at
	\begin{equation}
	\mathcal{F}^{\mathrm{H},\mathrm{c}}(\Gamma) \geq \frac{1}{\beta} \trs \left[ \ln\left( 1 -  e^{ -\beta \left( h + v_N \ast \eta - \mu \right)  } \right) \right] + \mu N - D_N\left(\eta,\eta\right)
	\label{eq:sec4secondlemmav25}
	\end{equation}
	for any  $\mu \in \mathbb{R}$  such that $h + v_N \ast \eta - \mu$ is strictly positive. When we take the infimum over $\Gamma \in \mathcal{S}_N^{\mathrm{gc}}$ as well the supremum over $(\eta, \mu ) \in \mathcal{D}_N^{\mathrm{gc}}$ on both sides of \eqref{eq:sec4secondlemmav25}, we find
	\begin{equation}
	F^{\mathrm{H}}(\beta,N,\omega) \geq \sup_{(\eta,\mu) \in \mathcal{D}_N^{\mathrm{gc}}} \left\{ \frac{1}{\beta} \trs \left[ \ln\left( 1 -  e^{ -\beta \left( h + v_N \ast \eta - \mu \right)  } \right) \right] + \mu N - D_N\left(\eta,\eta\right) \right\}.
	\label{eq:sec4secondlemmav26}
	\end{equation}
	Here we also used that \eqref{eq:154} holds as an equality. 
	It remains to prove the reverse inequality.
	
	Let $\mu^{\mathrm{H}}$ be the chemical potential for the Hartree minimizer $\gamma^{\mathrm{H}}$ in \eqref{eq:ELHartree1part}, and denote $\varrho^{\mathrm{H}}(x) = \gamma^{\mathrm{H}}(x,x)$. Since $(\varrho^{\mathrm{H}}, \mu^{\mathrm{H}}) \in \mathcal{D}_N^{\mathrm{gc}}$ we have
	\begin{align}
	&\sup_{(\eta,\mu) \in \mathcal{D}_N^{\mathrm{gc}}} \left\{ \frac{1}{\beta} \trs \left[ \ln\left( 1 -  e^{ -\beta \left( h + v_N  \ast \eta - \mu \right)  } \right) \right] + \mu N - D_N\left(\eta,\eta\right) \right\} \label{eq:sec4secondlemmav27} \\
	&\hspace{6cm} \geq \frac{1}{\beta} \trs \left[ \ln\left( 1 -  e^{ -\beta \left( h + v_N \ast \varrho^{\mathrm{H}} - \mu^{\mathrm{H}} \right)  } \right) \right] + \mu^{\mathrm{H}} N - D_N\left( \varrho^{\mathrm{H}} ,\varrho^{\mathrm{H}} \right). \nonumber
	\end{align}
	In combination with Lemma~\ref{lem:minimizersHartree} and \eqref{eq:sec4secondlemmav26}, this proves the claim.
\end{proof}
In the next Lemma we estimate the difference of $F^{\mathrm{H}}(\beta,N,\omega)$ and $F^{\mathrm{H},\mathrm{c}}(\beta,N,\omega)$.
\begin{lemma}
	\label{lem:freeenergyboundcgc}
	We have the bound
	\begin{equation}
		F^{\mathrm{H}}(\beta,N,\omega) \leq F^{\mathrm{H},\mathrm{c}}(\beta,N,\omega) \leq F^{\mathrm{H}}(\beta,N,\omega) + \frac{1}{\beta} \left( 1 + \ln\left( 1 + N \right) \right).
		\label{eq:closenessHFE} 
	\end{equation}
\end{lemma}
\begin{proof}
	The lower bound on $F^{\mathrm{H},\mathrm{c}}(\beta,N,\omega)$ in \eqref{eq:closenessHFE} follows from $\mathcal{S}_N^{\mathrm{c}} \subset \mathcal{S}_N^{\mathrm{gc}}$ and the fact that \eqref{eq:154} holds as an equality. 
From \cite[Corollary~A.1]{me} we know that 
\begin{align}
&-\frac{1}{\beta} \ln \left( \tr \exp\left( -\beta \sum_{i=1}^N  h_i + v_N * \eta (x_i) \right) \right) - D_N(\eta, \eta) \label{eq:mainresults6} \\
&\hspace{3cm} \leq \frac{1}{\beta} \tr \left[ \ln\left( 1 - e^{ - \beta \left(   h +  v_N*\eta - \mu \right) } \right) \right] + \mu N - D_N(\eta, \eta) + \frac{1}{\beta} \left( 1 + \ln\left( 1 + N \right) \right) \nonumber \\
&\hspace{3cm} \leq \sup_{ (\sigma, \widetilde{ \mu}) \in \mathcal{D}_N^{\mathrm{gc}} } \left\{ \frac{1}{\beta} \tr \left[ \ln\left( 1 - e^{ - \beta \left(   h +  v_N * \sigma  - \widetilde{\mu} \right) } \right) \right] + \widetilde{\mu} N - D_N(\sigma, \sigma) \right\} + \frac{1}{\beta} \left( 1 + \ln\left( 1 + N \right) \right) \nonumber 
\end{align}
holds for any $\eta \in L^1(\mathbb{R}^3)$. By choosing $\eta = \varrho_{G^{\mathrm{H},\mathrm{c}}}$ with $G^{\mathrm{H},\mathrm{c}}$ in \eqref{eq:HartreeGibbsstatec} on the left-hand side of \eqref{eq:mainresults6}, 
the desired upper bound follows from Lemma~\ref{lem:secondvariationalcharacterization}.
\end{proof}
\section{The semiclassical free energy functional and its critical temperature}
\label{sec:semiclassicalfreeenergyfunctional}
This third section is devoted to the study of the semiclassical free energy functional $\mathcal{F}^{\mathrm{sc}}$ in \eqref{eq:semic-lassicalfunctional3i}. In the first part we prove the existence of a unique minimizer and establish the corresponding Euler--Lagrange equation. In the second part we use these preparations to prove Proposition~\ref{prop:crittempscmf}. 
\subsection{Properties of the semiclassical free energy functional}
\label{sec:scfefandcriticaltemp}

\begin{lemma}
	\label{lem:semiclassicalfreeenergy}
	Let Assumption~\ref{as:regularitypotential} hold. Then the semiclassical free energy functional $\mathcal{F}^{\mathrm{sc}}$ admits a unique minimizer $(\gamma,g)$ in the set $\mathcal{D}^{\mathrm{sc}}$ in  \eqref{eq:normalizationcondition}. Moreover, the minimizer solves the Euler--Lagrange equation
	\begin{equation}
	\gamma(p,x) = \frac{1}{\exp\left( \beta \left( p^2 + \frac{\omega^2 x^2}{4} + v \ast \varrho(x) - v \ast \varrho(0) - \mu \right) \right)-1}
	\label{eq:semiclassicalfunctional5}
	\end{equation}
	pointwise a.e., where the density $\varrho$ is defined in \eqref{eq:semiclassicalfunctional2i} and the chemical potential satisfies $\mu \leq 0$ as well as  
	$\mu g = 0$. The semiclassical free energy can be written in terms of the minimizer $(\gamma,g)$ as
	\begin{align}
	F^{\mathrm{sc}}(\beta,\omega) =& \frac{1}{\beta} \frac{1}{(2 \pi)^3} \int_{\mathbb{R}^6} \ln \left( 1-\exp\left( - \beta \left( p^2 + \frac{\omega^2 x^2}{4} + v \ast \varrho(x) - v \ast \varrho(0) - \mu \right) \right) \right) \de(p,x) 
	\label{eq:semiclassicalfunctional5b} \\
	& + v \ast \varrho(0) + \mu - D(\varrho,\varrho). \nonumber
	\end{align}
\end{lemma}

Note that the condition $\mu g =0$ implies that $\mu = 0$ whenever $g>0$, and $g=0$ whenever $\mu < 0$. 

\begin{proof}
	Since $v \geq 0$ we have $D(\varrho,\varrho) \geq 0$. In combination with the Gibbs variational principle, this implies the lower bound
	\begin{equation}
	\mathcal{F}^{\mathrm{sc}}(\gamma,g) \geq \frac{1}{2 (2\pi)^3} \int_{\mathbb{R}^6} \left( p^2 + \frac{\omega^2 x^2}{4} \right) \gamma(p,x) \de(p,x) + \frac{1}{ \beta (2\pi)^3} \int_{\mathbb{R}^6} \ln\left( 1 - e^{-\beta/2 \left( p^2 + \frac{\omega^2 x^2}{4} \right)} \right) \de(p,x). 
	\label{eq:semiclassicalfunctional7}
	\end{equation}
	The first term on the right-hand side is nonnegative and the second term is finite. In particular, $\mathcal{F}^{\mathrm{sc}}$ is bounded from below. 
	
	The proof of the existence of a unique minimizer will be carried out in three steps. Our approach is motivated by a similar proof strategy in 
	\cite{NRS18}. In the first step we show the existence of a unique minimizer $(\gamma_C,g_C)$ in the set 
	\begin{equation}
	\widetilde{\mathcal{D}}_{C} = \left\{ (\gamma,g) \in \mathcal{D}^{\mathrm{sc}} \ \big| \
	\gamma(p,x) \leq C \text{ for all } (p,x) \in \mathbb{R}^6 \right\}
	\label{eq:semiclassicalfunctional8}
	\end{equation}
	with $C>0$, which can be achieved with standard techniques. Afterwards, we show  that $(\gamma_C,g_C)$ is a minimizing sequence for the original problem as $C \to \infty$,  and in the final step we conclude the existence of a unique minimizer for the original problem. This procedure is necessary in order to exclude the possibility that the phase space density of the minimizing sequence converges to a measure that is not absolutely continuous w.r.t. the Lebesgue measure. The Euler--Lagrange equation of the unrestricted problem can then be established with standard arguments. Eq.~\eqref{eq:semiclassicalfunctional5b} for $F^{\mathrm{sc}}(\beta,\omega)$ follows if one inserts the solution of the Euler--Lagrange equation into $\mathcal{F}^{\mathrm{sc}}$.
	
	\textit{Step 1 (Existence of a unique minimizer for the restricted problem): } Let $(\gamma_n,g_n) \in \widetilde{\mathcal{D}}_{C}$ be a minimizing sequence for the restricted problem. We know that $\Vert \gamma_n \Vert_{L^1(\mathbb{R}^6)} \leq (2 \pi)^3$ and $0 \leq \gamma_n(p,x) \leq C$, and hence $\gamma_n$ is uniformly bounded in $L^p(\mathbb{R}^6)$ for $1 \leq p \leq \infty$. Pick  $1 < q < \infty$. By going to a subsequence (which we do not display in the notation for simplicity) we conclude that there exists a $\gamma \in L^q(\mathbb{R}^6)$ such that $\gamma_n \rightharpoonup \gamma$ weakly in $L^q(\mathbb{R}^6)$. Using Mazur's theorem, see e.g. \cite[Theorem~2.13]{LiLo10}, we can, by going to a convex combination of the original sequence, assume that $\gamma_n \rightarrow \gamma$ strongly in $L^q(\mathbb{R}^6)$. This new sequence is still minimizing because $\mathcal{F}^{\mathrm{sc}}$ is convex. Using \cite[Theorem~2.7]{LiLo10} we can also assume that $\gamma_n \rightarrow \gamma$ pointwise almost everywhere. For the sequence of condensate fractions we have $0 \leq g_n \leq 1$, and hence there exists a subsequence and a $g \in [0,1]$ such that $g_n \rightarrow g$. Next we show that $(\gamma,g) \in \widetilde{\mathcal{D}}_{C}$, 	which can be deduced from the fact that the model has a confining potential. More precisely,  since 
	\begin{equation}
		 \int_{\mathbb{R}^6} \left( p^2 + \frac{ \omega^2 x^2}{4} \right) \gamma_n(p,x) \de(p,x) \leq \text{const. }
	\end{equation}
	by \eqref{eq:semiclassicalfunctional7}, we have
	\begin{equation}
		\left( \frac{1}{2\pi} \right)^3 \int_{ \left\{ p^2 + \frac{ \omega^2 x^2}{4} > R^2 \right\}}  \gamma_n(p,x) \de(p,x) \leq \frac{ \text{const.}}{R^2},
	\end{equation}
	from which we conclude that
	\begin{equation}
		\left( \frac{1}{2\pi} \right)^3 \int_{ \left\{ p^2 + \frac{ \omega^2 x^2}{4} \leq R^2 \right\} } \gamma_n(p,x) \de(p,x) \geq \left( \frac{1}{2\pi} \right)^3 \int_{\mathbb{R}^6} \gamma_n(p,x) \de(p,x) - \frac{ \text{const.}}{R^2} = 1 - g_n - \frac{ \text{const.}}{R^2}\,.
		\label{eq:1}
	\end{equation}
	We take the limit $n \to \infty$ on both sides, use the convergence $\gamma_n \to \gamma$ in $L^1_{\mathrm{loc}}(\mathbb{R}^6)$ and $g_n \to g$, and afterwards take the limit $R \to \infty$. This proves
	\begin{equation}
		\left( \frac{1}{2\pi} \right)^3 \int_{\mathbb{R}^6} \gamma(p,x) \de(p,x) \geq 1 - g.
	\end{equation}
	The reverse inequality follows from Fatou's Lemma, and we conclude that $(\gamma,g) \in \widetilde{\mathcal{D}}_{C}$.
	
	To see that $\mathcal{F}^{\mathrm{sc}}$ is lower semicontinuous along the minimizing sequence, we first observe that
	\begin{equation}
	\left( p^2 + \frac{\omega^2 x^2}{4} \right) \gamma_n(p,x) + \frac{1}{\beta} f(\gamma_n(p,x)) \geq \frac{1}{\beta} \ln\left( 1 - e^{-\beta \left( p^2 + \frac{\omega^2 x^2}{4} \right)} \right)
	\label{eq:semiclassicalfunctional9}
	\end{equation}
	 with $f$ defined in \eqref{eq:bosonicentropy}. In combination, \eqref{eq:semiclassicalfunctional9}, Fatou's Lemma and the pointwise convergence of $\gamma_n$ imply
	\begin{align}
	&\liminf_{n \to \infty} \frac{1}{(2 \pi )^3} \int_{\mathbb{R}^6} \left[  \left( p^2 + \frac{\omega^2 x^2}{4} \right) \gamma_n(p,x) + \frac{1}{\beta} f(\gamma_n(p,x)) \right] \de(p,x) \label{eq:semiclassicalfunctional10} \\
	& \hspace{5cm} \geq \frac{1}{(2 \pi )^3} \int_{\mathbb{R}^6} \left[  \left( p^2 + \frac{\omega^2 x^2}{4} \right) \gamma(p,x) + \frac{1}{\beta} f(\gamma(p,x)) \right] \de(p,x). \nonumber
	\end{align}
	Similarly, using Fatou's Lemma, the pointwise convergence of $\gamma_n$ and the convergence of $g_n$, we also see that 
	\begin{equation}
	\liminf_{n \to \infty} D(\varrho_n,\varrho_n) \geq D(\varrho,\varrho).
	\label{eq:semiclassicalfunctional11}
	\end{equation}
	This proves the lower semicontinuity along the minimizing sequence, and we conclude the existence of a minimizing pair $(\gamma_C,g_C) \in \widetilde{\mathcal{D}}_{C}$. The minimizer is unique because $\mathcal{F}^{\mathrm{sc}}$ is strictly convex. This follows from $\hat v\geq 0$ and the fact that the semiclassical entropy is strictly concave.

	By standard arguments one can conclude  that the minimizing pair $(\gamma_C,g_C)$ solves the equation
	\begin{equation}
	\gamma_C(p,x) = \min\left\{ C \  , \  \frac{1}{\exp\left(\beta \left( p^2 + \frac{\omega^2 x^2}{4} + v \ast \varrho_C(x) - v \ast \varrho_C(0) - \mu_C \right)\right) -1} \right\}
	\label{eq:semiclassicalfunctional15}
	\end{equation}
	for some $\mu_C \in \mathbb{R}$, where $\varrho_C$ denotes the density of the minimizer $(\gamma_C,g_C)$ (defined in \eqref{eq:semiclassicalfunctional2i}). We claim that $\mu_C \leq 0$. To see this, assume on the contrary that $\mu_C > 0$ and define, for $\epsilon>0$, the trial states $(\tilde\gamma_\epsilon,\tilde g_\epsilon)$ by
	\begin{equation}
	\tilde\gamma_\epsilon(p,x) = \min\left\{ C \  , \  \frac{1}{\exp\left(\beta \left( p^2 + \frac{\omega^2 x^2}{4} + v \ast \varrho_C(x) - v \ast \varrho_C(0) - \mu_C + \epsilon \right)\right) -1} \right\}
	\label{eq:semiclassicalfunctional15a}
	\end{equation}
	with corresponding $\tilde g_\epsilon > g_C$ such that the normalization condition \eqref{eq:normalizationcondition} holds. A simple calculation shows that
	\begin{equation}
	\lim_{\epsilon\to 0}  \epsilon^{-1} \left( \mathcal{F}^{\mathrm{sc}}(\tilde \gamma_\epsilon, \tilde g_\epsilon) - \mathcal{F}^{\mathrm{sc}}(\gamma_C,g_C) \right) =  \frac{-\mu_C}{ ( 2 \pi )^3}  \int_{ \{ \gamma_C < C \}} \gamma_C(p,x) \left( \gamma_C(p,x) + 1\right) \de(p,x)
	\end{equation}
	and contradicts the assumption that $(\gamma_C,g_C)$ is a minimizer. Hence $\mu_C\leq 0$. We note that if $g_C>0$, the above calculation actually shows that $\mu_C=0$, since in this case $\epsilon$ can also be taken negative. In particular, $\mu_C g_C = 0$.

	\textit{Step 2 ($(\gamma_C,g_C)$ is a minimizing sequence for the unrestricted problem): } We have
	\begin{equation}
	\liminf_{C \to \infty} \mathcal{F}^{\mathrm{sc}}(\gamma_C,g_C) \geq \inf_{(\gamma,g) \in \mathcal{D}^{\mathrm{sc}}}  \mathcal{F}^{\mathrm{sc}}(\gamma,g).
	\label{eq:semiclassicalfunctional19}
	\end{equation}
	To establish the reverse inequality for the $\limsup$, we pick  $(\gamma,g) \in \mathcal{D}^{\mathrm{sc}}$ and define 
	\begin{equation}
	\widetilde{\gamma}_C(p,x) = \gamma(p,x) \mathds{1}_{\{ \gamma < C \}}(p,x).
	\label{eq:104}
	\end{equation}
	The condensate fraction $\widetilde{ g}_C \geq g$ is chosen such that the normalization condition \eqref{eq:normalizationcondition} holds for the pair $(\widetilde{\gamma}_C,\widetilde{ g}_C)$. In the following, we derive a lower bound for the free energy of $(\gamma,g)$ in terms of the one of $(\widetilde{\gamma}_C,\widetilde{ g}_C)$. First, we note that
	\begin{align}
	\int_{\mathbb{R}^6} \left[  \left( p^2 + \frac{\omega^2 x^2}{4} \right) \gamma(p,x) + \frac{1}{\beta} f(\gamma(p,x)) \right] \de(p,x) \geq& \int_{ \mathbb{R}^6 } \left[  \left( p^2 + \frac{\omega^2 x^2}{4} \right) \widetilde{\gamma}_C(p,x) + \frac{1}{\beta} f(\widetilde{\gamma}_C(p,x))\right] \de(p,x) \nonumber  \\
	&+\frac{1}{\beta} \int_{\{ \gamma \geq C \}} \ln\left( 1 - e^{-\beta \left( p^2 + \frac{\omega^2 x^2}{4} \right) } \right) \de(p,x). \label{eq:semiclassicalfunctional20} 
	\end{align}
	To bound the second term on the right-hand side, we use
	\begin{equation}
	\left| \{ \gamma \geq C \} \right| \leq \frac{(2 \pi)^3}{C},
	\label{eq:semiclassicalfunctional21}
	\end{equation}
	which follows from $\frac{1}{(2 \pi )^3} \int_{\mathbb{R}^6} \gamma(p,x) \de(p,x) \leq 1$. Here $| A |$ denotes the Lebesgue measure of the set $A$. More precisely, we have
	\begin{align}
	&\frac{1}{\beta } \frac{1}{(2 \pi )^3} \int_{ \{ \gamma \geq C \} } \ln\left( 1- e^{-\beta \left( p^2+ \frac{ \omega^2 x^2}{4} \right)} \right) \de(p,x) \geq \frac{1}{\beta } \frac{1}{(2 \pi )^3} \int_{ \left\{  \beta \left(p^2 + \frac{\omega^2 x^2}{4} \right) \leq \frac{\beta \omega}{(C/6)^{1/3}} \right\} } \ln\left( 1- e^{-\beta \left( p^2+ \frac{ \omega^2 x^2}{4} \right)} \right) \de(p,x) \nonumber \\
	&\hspace{1cm} = \frac{1}{\beta (\beta \omega)^3} \frac{8}{(2 \pi )^3} \int_{ \left\{ p^2 + x^2  \leq \frac{\beta \omega}{(C/6)^{1/3}} \right\} } \ln\left( 1- e^{-\left( p^2+ x^2\right)} \right) \de(p,x) \gtrsim -\frac{\ln\left( C \right)}{\beta C}. \label{eq:semiclassicalfunctional22}
	\end{align}
	To obtain the first inequality, we replaced the set $\{ \gamma \geq C \}$ by a ball with the same volume centered around zero in $\mathbb{R}^6$. Since $| \ln( 1- \exp(-\beta ( p^2+ \frac{ \omega^2 x^2}{4} )) ) |$ is a radial and monotone decreasing function this can only increase the absolute value of its integral.
	
	It remains to consider the interaction energy.  
	Since $\widetilde{ g }_{C} \geq g$, we have the lower bound
	\begin{align}
	D(\varrho,\varrho) &\geq D(\widetilde{\varrho}_C,\widetilde{\varrho}_C) + (g-\widetilde{g}_C) \int_{\mathbb{R}^3} v(x) \varrho_{\widetilde{\gamma}_C}(x) \de x + \frac{v(0)}{2} \left( g^2 - \widetilde{g}_C^2 \right) \label{eq:semiclassicalfunctional23} \\
	&\geq D(\widetilde{\varrho}_C,\widetilde{\varrho}_C) - 2 \Vert v \Vert_{L^{\infty}(\mathbb{R}^3)} |g-\widetilde{g}_C|. \nonumber  
	\end{align}
	Here $\varrho_{\widetilde{\gamma}_{C}}(x) = \frac{1}{(2 \pi)^3} \int_{\mathbb{R}^3} \widetilde{\gamma}(p,x) \de p$ and $\widetilde{ \varrho }_C(x) = \varrho_{\widetilde{\gamma}_{C}}(x) + \widetilde{g}_{C} \delta(x)$. In combination, \eqref{eq:semiclassicalfunctional20}, \eqref{eq:semiclassicalfunctional22} and \eqref{eq:semiclassicalfunctional23} imply the bound
	\begin{align}
	\mathcal{F}^{\mathrm{sc}}(\gamma,g) &\geq \mathcal{F}^{\mathrm{sc}}(\widetilde{\gamma}_C,\widetilde{g}_C)  -\frac{\text{const.}}{\beta C} \ln\left( C \right) - 2 \Vert v \Vert_{L^{\infty}(\mathbb{R}^6)} |g-\widetilde{g}_C| 
	\label{eq:semiclassicalfunctional24} \\
	&\geq \mathcal{F}^{\mathrm{sc}}(\gamma_C,g_C) -\frac{\text{const.}}{\beta C} \ln\left( C \right) - 2 \Vert v \Vert_{L^{\infty}(\mathbb{R}^6)} |g-\widetilde{g}_C|. \nonumber
	\end{align}
	with the minimizing pair $(\gamma_C,g_C)$ of the restricted problem. Next, we take the $\limsup_{C \to \infty}$ on both sides of \eqref{eq:semiclassicalfunctional24}, and use the definition of $\widetilde{\gamma}_C$ in \eqref{eq:104} as well as the normalization condition in \eqref{eq:normalizationcondition} to check that $g - \widetilde{g}_C \to 0$ as $C \to \infty$. Finally, we minimize  over $(\gamma,g) \in \mathcal{D}^{\mathrm{sc}}$ and find
	\begin{equation}
	\inf_{(\gamma,g) \in \mathcal{D}^{\mathrm{sc}}}  \mathcal{F}^{\mathrm{sc}}(\gamma,g) \geq \limsup_{C \to \infty} \mathcal{F}^{\mathrm{sc}}(\gamma_C,g_C).
	\label{eq:semiclassicalfunctional25}
	\end{equation}
	In combination with \eqref{eq:semiclassicalfunctional19} this proves that $(\gamma_C,g_C)$ is a minimizing sequence for the unrestricted problem.
	
	\textit{Step 3 (Existence of a unique minimizer for the unrestricted problem): } Let $\varrho : \mathbb{R}^3 \to \mathbb{R}$ be a nonnegative function with $\int_{\mathbb{R}^3} \varrho(x) \de x \leq 1$ and $\varrho(-x) = \varrho(x)$ or, more generally, a non-negative measure with these properties. The assumption \eqref{eq:condintpot} on $v$ implies the bound 
	\begin{equation}
		\frac{\omega^2 x^2}{4} + v \ast \varrho(x) - v \ast \varrho(0) \geq c \omega^2 x^2
		\label{eq:34}
	\end{equation}	
	for some $c>0$. To see this, we use a second order Taylor approximation to write the effective interaction potential as
	\begin{equation}
		v \ast \varrho(x) = v \ast \varrho(0)  + x^T \cdot \left( \frac{1}{2} \int_{\mathbb{R}^3} \int_0^1 D^2 v\left( y + s(x-y) \right) \varrho(y) \de s \de y \right) \cdot x,
		\label{eq:60}
	\end{equation}
	and use \eqref{eq:condintpot} to obtain a lower bound for the second term on the right-hand side. This proves \eqref{eq:34}.
	
	From $v(-x) = v(x)$ and the fact that the pair $(\gamma_C,g_C)$ is the unique minimizer of the restricted problem, we know that $\varrho_C(-x) = \varrho_C(x)$. Since $\mu_C\leq 0$, Eq.~\eqref{eq:34} thus shows that $\gamma_C$ in \eqref{eq:semiclassicalfunctional15} is bounded by
	\begin{equation}
	\gamma_C(p,x)  \leq  \frac{1}{\exp\left(\beta \left( p^2 +  c {\omega^2 x^2} \right)\right) -1} \,.
	\end{equation}
	In particular, the sequence $\gamma_C$ is uniformly in $L^p(\mathbb{R}^6)$ for $p<3$. 
	With this information we apply Mazur's theorem in the same way as in the case of the restricted problem and prove the existence of a unique minimizer for the unrestricted problem. Uniqueness follows from the strict concavity of the semiclassical entropy, which implies that $\mathcal{F}^{\mathrm{sc}}$ is strictly convex.
	
	The Euler--Lagrange equation for the unrestricted problem can be obtained in the same way as the one for the restricted problem, and in the same way one concludes that $\mu  g = 0$. Eq.~\eqref{eq:semiclassicalfunctional5b} for $F^{\mathrm{sc}}(\beta,\omega)$ 
	follows if we insert the solution of the Euler--Lagrange equation into the formula for $\mathcal{F}^{\mathrm{sc}}$. We omit the details.
	This completes the proof of Lemma~\ref{lem:semiclassicalfreeenergy}.
\end{proof}
	
\subsection{Proof~of~Proposition~\ref{prop:crittempscmf}}
	
	Let $(\gamma^{\mathrm{sc}},g^{\mathrm{sc}})$ be the unique minimizing pair of $\mathcal{F}^{\mathrm{sc}}$, whose existence is guaranteed by Lemma~\ref{lem:semiclassicalfreeenergy},  let $\mu^{\mathrm{sc}}$ be the corresponding chemical potential in  \eqref{eq:semiclassicalfunctional5},  and  $\varrho^{\mathrm{sc}}(x) = \varrho_{\gamma^{\mathrm{sc}}}(x) + g^{\mathrm{sc}} \delta(x)$. By $\gamma_0$ and $\varrho_0$ we denote the corresponding quantities of the ideal gas at inverse critical temperature $\beta_0$ in \eqref{eq:idealgascriticaltemp2}.
	
We start by observing that $\gamma^{\mathrm{sc}}$ and $g^{\mathrm{sc}}$ depend continuously on $\beta$. This follows from uniqueness of the minimizers. Following the proof of Lemma~\ref{lem:semiclassicalfreeenergy}, one readily checks that for a sequence $\{\beta_n\}$ with $\lim_{n\to \infty} \beta_n = \beta>0$, the corresponding minimizers form a minimizing sequence for the problem at inverse temperature $\beta$ and actually converge to the corresponding minimizer. 

From the  continuity of both $\mu^{\mathrm{sc}}$ and $g^{\mathrm{sc}}$ in $\beta$ we conclude that the boundary of the temperature region where $g^{\mathrm{sc}}=0$ is characterized by the conditions that $\mu^{\mathrm{sc}} = g^{\mathrm{sc}} =0$. In the following, we show that this characterizes a unique critical temperature, at least for weak coupling. The latter restriction comes from the use of a fixed point argument.

The Euler--Lagrange equation \eqref{eq:semiclassicalfunctional5}  at $g^{\mathrm{sc}}=\mu^{\mathrm{sc}} =0$ implies upon integration over $p$ the equation 
\begin{equation}\label{rso}
\rho^{\mathrm{sc}}(x) = \beta^{-3/2} \eta \left( \beta \left(\tfrac{\omega^2}{4} x^2 + \lambda v*\rho^{\mathrm{sc}}(x) - \lambda v*\rho^{\mathrm{sc}}(0) \right)\right)
\end{equation}
where
\begin{equation}
\eta( t) =  \frac{1}{(2\pi)^3} \int_{\mathbb{R}^3} \frac 1{ \exp\left( p^2 + t \right) -1}  \de p
\end{equation}
for $t\geq 0$. Note that $\eta$ is a bounded, rapidly decreasing and convex function. On the set $X= \{ \rho\in L^1(\mathbb{R}^3) ,\ \rho\geq 0, \ \int \rho = 1\}$ define the map $T$ that maps 
\begin{equation}
X \ni \rho \mapsto\beta^{-3/2} \eta \left( \beta \left(\tfrac{\omega^2}{4} x^2 + \lambda v*\rho(x) - \lambda v*\rho(0) \right)\right) \,,
\end{equation}
where $\beta>0$ is chosen (depending on $\rho$)  as the unique value  such that the integral of the right-hand side of \eqref{rso} equals $1$. Note that by strict monotonicity in $\beta$, there exists a unique such $\beta$ for every $\rho \in X$. 

\begin{lemma} \label{lem:ban}
For $\lambda$ small enough, the map $T:X\to X$ is a contraction. In particular, in this case there exists a unique solution 
 $(\beta,\rho^{\mathrm{sc}})\in \mathbb{R}_+ \times X$ of Eq.~\eqref{rso}.
 \end{lemma} 

\begin{proof}
We start by giving uniform bounds on the possible values that $\beta$ can take. The assumption~\eqref{eq:condintpot} on $v$ implies,  as in \eqref{eq:34}, that 
		\begin{equation}
		 \frac{\omega^2 x^2}{4} \left(  1 + c \lambda\right) \geq \frac{\omega^2 x^2}{4} + \lambda v \ast \varrho(x) - \lambda v \ast \varrho(0) \geq  \frac{\omega^2 x^2}{4} \left(  1 - c \lambda\right)
		\label{eq:34n}
	\end{equation}	
	for a suitable $0<c<1$. By the monotonicity of $\eta$  we thus have, for any $\rho \in X$, 
	\begin{equation}
	1 = \int T\rho \leq  \int_{\mathbb{R}^3} \beta^{-3/2} \eta \left( \beta \omega^2 x^2 \tfrac{1}{4} \left(  1 - c \lambda\right) \right) \de x  = \left( \frac{\beta_0 }\beta\right)^3 \left(  1 - c \lambda\right)^{-3/2}
	\end{equation}
	and a similar argument can be used for a bound in the opposite direction. In particular, 
	\begin{equation}\label{eq:34m}
	\left(  1 + c \lambda\right)^{-1/2} \beta_0 \leq \beta \leq \beta_0 \left(  1 - c \lambda\right)^{-1/2}
	\end{equation}
	holds for any $\rho \in X$. 
	
	To show that $T$ is a contraction, we write, for any $\rho_1,\rho_2 \in X$,
	\begin{equation}
	\int \left| T\rho_1- T\rho_2 \right|  = \int_{\mathbb{R}^3} \left| \beta_1^{-3/2} \eta \left( \beta_1 \left(\tfrac{\omega^2}{4} x^2 + \lambda V_1(x) \right)\right)  - \beta_2^{-3/2} \eta \left( \beta_2 \left(\tfrac{\omega^2}{4} x^2 + \lambda V_2(x) \right)\right)  \right| \de x \,,
	\end{equation}
	where $\beta_1$ and $\beta_2$ are the values $\beta$ takes in $T\rho_1$ and $T\rho_2$, respectively, and we introduced the notation $V_j(x) =  v*\rho_j(x) -  v*\rho_j(0)$ for $j\in \{1,2\}$ for simplicity.  In particular,
	\begin{equation}
	\int \left| T\rho_1- T\rho_2 \right|\leq \left| 1- (\beta_1/ \beta_2)^{3/2}\right|  + \beta_2^{-3/2}  \int_{\mathbb{R}^3} \left|  \eta \left( \beta_1 \left(\tfrac{\omega^2}{4} x^2 + \lambda V_1(x) \right)\right)  - \eta \left( \beta_2 \left(\tfrac{\omega^2}{4} x^2 + \lambda V_2(x) \right)\right)  \right| \de x\,. \label{ltorhs}
	\end{equation}
	To bound the difference of $\beta_1$ and $\beta_2$, let us assume without loss of generality that $\beta_1>\beta_2$. Then, by the monotonicity of $\eta$,
	\begin{align}\nonumber
	1 & =  \beta_2^{-3/2}  \int_{\mathbb{R}^3} \eta \left( \beta_2 \left(\tfrac{\omega^2}{4} x^2 + \lambda V_2(x) \right)\right)   \de x \\ \nonumber & \geq \beta_2^{-3/2}  \int_{\mathbb{R}^3} \eta \left( \beta_1 \left(\tfrac{\omega^2}{4} x^2 + \lambda V_2(x) \right)\right) \de x \\ &=  \left(\beta_1/\beta_2\right)^{3/2} -  \beta_2^{-3/2}  \int_{\mathbb{R}^3} \left(  \eta \left( \beta_1 \left(\tfrac{\omega^2}{4} x^2 + \lambda V_1(x) \right)\right)- \eta \left( \beta_1 \left(\tfrac{\omega^2}{4} x^2 + \lambda V_2(x) \right)\right)\right) \de x\,.
	\end{align}
	Using a first-order Taylor expansion of the integrand, the monotonicity of $\eta'$ as well as \eqref{eq:34n}, we can bound the last term as
	\begin{align}\nonumber
	 & \int_{\mathbb{R}^3} \left(  \eta \left( \beta_1 \left(\tfrac{\omega^2}{4} x^2 + \lambda V_1(x) \right)\right)- \eta \left( \beta_1 \left(\tfrac{\omega^2}{4} x^2 + \lambda V_2(x) \right)\right)\right) \de x  \\ &
	  \leq  \beta_1 \ \lambda \int_{\mathbb{R}^3} \left|  \eta' \left( \beta_1 \left(\tfrac{\omega^2}{4} x^2(1-c\lambda) \right)\right)\right|   \left| V_1(x)- V_2(x) \right|   \de x \,.
	  \end{align}
	Using again Assumption~\eqref{eq:condintpot} on $v$ we can bound $|V_1(x) - V_2(x)| \lesssim \omega^2 x^2 \|\rho_1-\rho_2\|_1$. The resulting integral is then finite and uniformly bounded because of 
 \eqref{eq:34m}. We thus conclude that 
	\begin{equation}
	\beta_2^{-3/2}  \int_{\mathbb{R}^3} \left(  \eta \left( \beta_1 \left(\tfrac{\omega^2}{4} x^2 + \lambda V_1(x) \right)\right)- \eta \left( \beta_1 \left(\tfrac{\omega^2}{4} x^2 + \lambda V_2(x) \right)\right)\right) \de x  \lesssim \lambda \|\rho_1 - \rho_2\|_1\,.
	\end{equation} 
	In particular,
	\begin{equation}
	|\beta_1-\beta_2 | \lesssim \lambda \|\rho_1 -\rho_2\|_1 \,.
	\end{equation}
	By using a similar argument, one readily checks that the last term on the right-hand side of \eqref{ltorhs} can be bounded in terms of $|\beta_1-\beta_2|$ and $\lambda \|\rho_1-\rho_2\|_1$. Altogether, this proves that
	\begin{equation}
	\int \left| T\rho_1- T\rho_2 \right| \lesssim \lambda \|\rho_1 -\rho_2\|_1 
	\end{equation}
	and hence shows the desired contraction property for $\lambda$ small enough. The uniqueness of solutions of \eqref{rso} then follows from the Banach fixed point theorem. 
	\end{proof}
	
This completes the proof of part a) of Proposition~\ref{prop:crittempscmf}.
In order to prove part b), we first derive a bound  on the  difference of $\varrho^{\mathrm{sc}}$ at inverse temperature $\beta_{\mathrm{c}}$ and $\varrho_0$ at inverse temperature $\beta_0$. Using \eqref{rso} we can write
\begin{equation}
\varrho^{\mathrm{sc}}(x) - \varrho_{0}(x) = \beta_{\mathrm{c}}^{-3/2} \eta \left( \beta_{\mathrm{c}} \left(\tfrac{\omega^2}{4} x^2 + \lambda v*\varrho^{\mathrm{sc}}(x) - \lambda v*\varrho^{\mathrm{sc}}(0) \right)\right) - \beta_0^{-3/2} \eta \left( \beta_0 \left(\tfrac{\omega^2}{4} x^2  \right)\right) \,.
\end{equation}
From \eqref{eq:34m} we already know that $|\beta_{\mathrm{c}} - \beta_0| \lesssim \lambda$. A  first order Taylor expansion as in the proof of Lemma~\ref{lem:ban} then readily shows that $\| \varrho^{\mathrm{sc}} - \varrho_{0}\|_1 \lesssim \lambda$.

Integration of \eqref{rso} over $x$ shows that the inverse critical temperature $\beta_{\mathrm{c}}$ satisfies the equation
\begin{equation}
\beta_{\mathrm{c}}^{3/2} = \int_{\mathbb{R}^3}  \eta \left( \beta_{\mathrm{c}} \left(\tfrac{\omega^2}{4} x^2 + \lambda v*\varrho^{\mathrm{sc}}(x) - \lambda v*\varrho^{\mathrm{sc}}(0) \right)\right)\de x \,.
\end{equation}
A second order Taylor expansion, the monotonicity of $\eta''$ as well as the bound \eqref{eq:34n} lead to the estimate
\begin{align}\nonumber
& \left| \beta_{\mathrm{c}}^{3/2} - \beta_0^{3}\beta_{\mathrm{c}}^{-3/2} - \lambda \beta_{\mathrm{c}} \int_{\mathbb{R}^3}  \eta' \left( \beta_{\mathrm{c}} \left(\tfrac{\omega^2}{4} x^2  \right)\right) 
\left( v*\varrho^{\mathrm{sc}}(x) -  v*\varrho^{\mathrm{sc}}(0)\right) \de x \right| 
\\  &  \leq \frac 12 \lambda^2 \beta_{\mathrm{c}}^2 \int_{\mathbb{R}^3}  \eta'' \left( \beta_{\mathrm{c}} \left(\tfrac{\omega^2}{4}  x^2 (1-c \lambda) \right)\right)
\left( v*\varrho^{\mathrm{sc}}(x) -  v*\varrho^{\mathrm{sc}}(0)  \right)^2 \de x  \lesssim \lambda^2  \label{lasti}
\end{align}
where the last inequality follows from the uniform bound \eqref{eq:34m} as well as the fact that the assumption \eqref{eq:condintpot} on $v$ implies that $|v*\varrho^{\mathrm{sc}}(x) -  v*\varrho^{\mathrm{sc}}(0)| \lesssim x^2$. In the integrand of the last term on the left-hand side of \eqref{lasti} we can replace $\varrho^{\mathrm{sc}}$ by $\varrho_0$ and bound the difference as
\begin{equation}
\left| v*\varrho^{\mathrm{sc}}(x) -  v*\varrho^{\mathrm{sc}}(0) - v*\varrho_{0}(x) +  v*\varrho_{0}(0)\right| \lesssim x^2 \| \varrho^{\mathrm{sc}} - \varrho_0 \|_1 \lesssim \lambda x^2\,.
\end{equation}
Altogether, this shows that 
\begin{equation}
 \left| \beta_{\mathrm{c}}^{3/2} - \beta_0^{3}\beta_{\mathrm{c}}^{-3/2} - \lambda \beta_{\mathrm{c}} \int_{\mathbb{R}^3}  \eta' \left( \beta_{\mathrm{c}} \left(\tfrac{\omega^2}{4} x^2  \right)\right) 
\left( v*\varrho_{0}(x) -  v*\varrho_{0}(0)\right) \de x \right|  \lesssim \lambda^2\,.
\end{equation}
Similar arguments as above allow to replace $\beta_{\mathrm{c}}$ by $\beta_0$ in the last term on the left-hand side, at the expense of another correction of order $\lambda^2$. Then  \eqref{eq:crittempshift} readily follows. 
	
	The fact that $\Xi > 0$ is a consequence of  $v \ast \varrho_{0}(x)  \leq v\ast \varrho_{0}(0) $ for all $x \in \mathbb{R}^3$. This, in turn, follows from the fact that   $\varrho_{0}$ has a nonnegative Fourier transform:
	\begin{align}
	\hat{\varrho}_{0}(p) &= \frac{1}{(2\pi)^{9/2}} \int_{\mathbb{R}^6} e^{-ipx} \frac{1}{ \exp\left( \beta_0 \left( q^2 + \frac{ \omega^2 x^2}{4} \right) \right) -1 } \de(q,x) \label{eq:15} \\
	&= \frac{1}{(2\pi)^{9/2}} \sum_{\alpha = 1}^{\infty} \left( \int_{\mathbb{R}^3} \exp\left( -\frac{ \beta_0 \alpha \omega^2 x^2}{4} \right) e^{-ipx}  \de x \right) \left( \int_{\mathbb{R}^3} \exp\left( -\beta_0 \alpha q^2 \right) \de q \right). \nonumber 
	\end{align}
	Both integrals on the right-hand side are nonnegative, and accordingly $\hat{\varrho}_{0}(p)$ is nonnegative. This completes the proof of Proposition~\ref{prop:crittempscmf}.

\section{Semiclassical mean-field limit of the Hartree free energy functional}
\label{sec:semiclassicalfreeenergyfunctionalb}
In this section we give the proof of Theorem~\ref{thm:limitHartreetheory}, which will be carried out in three steps. In the first two steps we prove  upper and  lower bounds on the Hartree free energy that, when combined, imply \eqref{eq:propsemiclassical1a}. In the third step we prove bounds for the Husimi function and the condensate fraction of the Hartree minimizer $\gamma^{\mathrm{H}}$ that imply \eqref{eq:propsemiclassical2a}. 
\subsection{Free energy bounds}
\label{sec:Hartreefreeenergybounds}
\subsubsection*{Proof of the upper bound for $F^{\mathrm{H}}(\beta,N,\omega)$}
Our argument is based on a trial state that is motivated in part by a similar argument at zero temperature in \cite{Li1981} and in part by an analysis at positive temperature using coherent states in \cite{NarnThi1981}. We start with the definition of the trial state.

Let $(\gamma^{\mathrm{sc}},g^{\mathrm{sc}})$ be the minimizer of the semiclassical free energy functional, whose existence is guaranteed by Lemma~\ref{lem:semiclassicalfreeenergy}, and let $\ell$ be a nonnegative radial function with $\int_{\mathbb{R}^3} \ell^2(x) \de x = 1$ and $\langle \ell, (-\Delta+\omega^2 x^2) \ell \rangle < + \infty$. Our trial state is given by
\begin{equation}
	\gamma_{\hbar} = g^{\mathrm{sc}} N \left| \ell_{0,0}^{\hbar} \right \rangle \left \langle \ell_{0,0}^{\hbar} \right| + \left( \frac{1}{2 \pi \hbar} \right)^3 \int_{\mathbb{R}^6} \left| \ell_{p,q}^{\hbar} \right \rangle \left \langle \ell_{p,q}^{\hbar} \right| \gamma^{\mathrm{sc}}(p,q) \de(p,q)
	\label{eq:7}
\end{equation}
with $\ell_{p,q}^{\hbar}$ defined in \eqref{eq:coherentstate1i}. Note that $\hbar^{-3} = N$ implies
\begin{equation}
	\trs[ \gamma_{\hbar} ] = N \left( g^{\mathrm{sc}} + \left( \frac{1}{2 \pi} \right)^3 \int_{\mathbb{R}^6} \gamma^{\mathrm{sc}}(p,q) \de(p,q) \right) = N.
	\label{eq:8}
\end{equation}
We need to compute $\mathcal{F}^{\mathrm{H}}(\gamma_{\hbar})$ with $\mathcal{F}^{\mathrm{H}}$ defined in \eqref{eq:Hartrees1pfreeenergyfkt} and we start with the entropy. 

The function $f$ in the definition of the bosonic entropy in \eqref{eq:bosonicentropy} is monotone decreasing. This implies that the map $A \mapsto \trs[f(A)]$ is monotone decreasing in the sense that $A \geq B$ implies $\trs[f(A)] \leq \trs[f(B)]$ for bounded operators $A, B \geq 0$. 
Accordingly, we have
\begin{equation}
	\trs[f(\gamma_{\hbar})] \leq \trs\left[ f\left( \left( \frac{1}{2 \pi \hbar} \right)^3 \int_{\mathbb{R}^6} \left| \ell_{p,q}^{\hbar} \right \rangle \left \langle \ell_{p,q}^{\hbar} \right| \gamma^{\mathrm{sc}}(p,q) \de(p,q) \right) \right].
	\label{eq:42}
\end{equation}
An upper bound on the right-hand side is provided by the Berezin--Lieb inequality, which we state in the following Lemma. For a proof see \cite{Berezin,Lieb1973,KlauSkag2012}.
\begin{lemma}
	\label{lem:Berezin-Lieb}
	Let $\zeta : \mathbb{R}_+ \cup \{ 0 \} \to \mathbb{R}$ be a convex function and define 
	\begin{equation}
		A = \left( \frac{1}{2 \pi \hbar} \right)^3 \int_{\mathbb{R}^6} \left| \ell_{p,q}^{\hbar} \right \rangle \left \langle \ell_{p,q}^{\hbar} \right| a(p,q) \de(p,q) 
		\label{eq:43}
	\end{equation}
	with $a(p,q) \geq 0$ chosen such that $\zeta(a) \in L^1(\mathbb{R}^6)$. Then 
	\begin{equation}
		\trs[\zeta(A)] \leq \left( \frac{1}{2 \pi \hbar} \right)^3 \int_{\mathbb{R}^6} \zeta\left( a(p,q) \right) \de(p,q)
		\label{eq:44}
	\end{equation}
	holds.
\end{lemma}

In combination, Lemma~\ref{lem:Berezin-Lieb}, the convexity of $f$ and \eqref{eq:42} imply
\begin{equation}
	\trs[f(\gamma_{\hbar})] \leq \left( \frac{1}{2 \pi \hbar} \right)^3 \int_{\mathbb{R}^6} f\left( \gamma^{\mathrm{sc}}(p,q) \right) \de(p,q).
	\label{eq:45}
\end{equation}
It remains to compute the energy.
To that end, we write
\begin{equation}
	\trs[h \gamma_{\hbar}] = g N \left\langle \ell^{\hbar}_{0,0}, h \ell^{\hbar}_{0,0} \right \rangle + \left( \frac{1}{2 \pi \hbar} \right)^3 \int_{\mathbb{R}^6} \left\langle \ell^{\hbar}_{p,q}, h \ell^{\hbar}_{p,q} \right \rangle \gamma^{\mathrm{sc}}(p,q) \de(p,q)
	\label{eq:46}
\end{equation}
and note that the first term on the right-hand side is given by
\begin{equation}
	\left\langle \ell^{\hbar}_{0,0}, h \ell^{\hbar}_{0,0} \right \rangle = \hbar \left\langle \ell, (-\Delta + \omega^2 x^2/4) \ell \right \rangle.
	\label{eq:47}
\end{equation}
The inner product in the second term on the right-hand side of \eqref{eq:46} reads
\begin{align}
	\left\langle \ell^{\hbar}_{p,q}, h \ell^{\hbar}_{p,q} \right \rangle &= \hbar^2 \int_{\mathbb{R}^3} \hbar^{-3/2} \left| \hbar^{-1/2} (\nabla \ell)\left( \frac{x-q}{\hbar^{1/2}} \right) + \ell\left( \frac{x-q}{\hbar^{1/2}} \right) \frac{ip}{\hbar} \right|^2 \de x + \int_{\mathbb{R}^3} \hbar^{-3/2} \left| \ell\left( \frac{x-q}{\hbar^{1/2}} \right) \right|^2 \frac{\omega^2 x^2}{4} \de x \label{eq:48} \\
	&= p^2 + \frac{\omega^2 q^2}{4} + \hbar \omega \langle \ell, (-\Delta + \omega^2 x^2/4) \ell \rangle. \nonumber
\end{align}
To arrive at the second line, we used that $\ell$ is a radial function. In combination, \eqref{eq:46}, \eqref{eq:47} and \eqref{eq:48} imply the bound
\begin{equation}
	\trs[h \gamma_{\hbar}] \leq \left( \frac{1}{2 \pi \hbar} \right)^3 \int_{\mathbb{R}^6} \left( p^2 + \frac{\omega^2 q^2}{4} \right) \gamma^{\mathrm{sc}}(p,q) \de(p,q) + N \hbar \langle \ell, (-\Delta + \omega^2 x^2/4) \ell \rangle.
	\label{eq:49}
\end{equation}

In the last step we compute the interaction energy of $\gamma_{\hbar}$.
The density $\varrho_{\gamma_{\hbar}}$ of  $\gamma_{\hbar}$ is given by
\begin{equation}
	\varrho_{\gamma_{\hbar}}(x) = N \hbar^{-3/2}  \int_{\mathbb{R}^3} \left| \ell\left ( \frac{x-q}{\hbar^{1/2}} \right) \right|^2 \varrho^{\mathrm{sc}}(q) \de q,
	\label{eq:50}
\end{equation}
where $\varrho^{\mathrm{sc}}(q) = g^{\mathrm{sc}} \delta(q) + \varrho_{\gamma^{\mathrm{sc}}}(q)$. We need to compute $D_N(\varrho_{\gamma_{\hbar}},\varrho_{\gamma_{\hbar}})$. A second order Taylor expansion yields the bound
\begin{align}
	\hbar^{-3} \int_{\mathbb{R}^6} v(x-y)  \left| \ell\left ( \frac{x-q_1}{\hbar^{1/2}} \right) \right|^2 \left| \ell\left ( \frac{y-q_2}{\hbar^{1/2}} \right) \right|^2 \de (x,y) &= \int_{\mathbb{R}^6} v\left(q_1 - q_2 + \hbar^{1/2}(x-y) \right) \left| \ell\left( x \right) \right|^2 \left| \ell\left ( y \right) \right|^2 \de (x,y) \nonumber \\
	&\leq v(q_1-q_2) + \hbar \sup_{x \in \mathbb{R}^3} \Vert D^2 v(x) \Vert \langle \ell, x^2 \ell \rangle, \label{eq:51}
\end{align}
where $D^2 v$ denotes the Hessian of $v$. Hence
\begin{equation}
	\frac{1}{N} D(\varrho_{\gamma_{\hbar}},\varrho_{\gamma_{\hbar}}) \leq N D(\varrho^{\mathrm{sc}},\varrho^{\mathrm{sc}}) + N \hbar \sup_{x \in \mathbb{R}^3} \Vert D^2 v(x) \Vert \langle \ell, x^2 \ell \rangle \,.
	\label{eq:52}
\end{equation}
In combination, \eqref{eq:45}, \eqref{eq:49} and \eqref{eq:52} imply the final estimate
\begin{equation}
	\hbar^{3} \mathcal{F}^{\mathrm{H}}(\gamma_{\hbar}) \leq F^{\mathrm{sc}}(\beta,\omega) + \hbar \left( \langle \ell, (-\Delta + \omega^2 x^2/4) \ell \rangle + \sup_{x \in \mathbb{R}^3} \Vert D^2 v(x) \Vert \langle \ell, x^2 \ell \rangle \right). 
	\label{eq:53}
\end{equation}

\subsubsection*{Proof of the lower bound for $F^{\mathrm{H}}(\beta,N,\omega)$}
The strategy for the lower bound is to estimate the free energy $\mathcal{F}^{\mathrm{H}}(\gamma)$ for a given 1-pdm $\gamma$ with $\trs[\gamma] = N$ from below in terms of the semiclassical free energy of its Husimi function. Let $\ell : \mathbb{R}^3 \to \mathbb{R}$ be a nonnegative radial function with $\int_{\mathbb{R}^3} \ell^2(x) \de x = 1$ and $\langle \ell, (-\Delta + \omega^2 x^2) \ell \rangle < +\infty$, and recall the definition \eqref{eq:coherentstate1i}--\eqref{eq:Husimifunctioni} of the Husimi function of $\gamma$. 

We start the analysis by noting that for any Borel measure $\eta \in \mathcal{M}_+(\mathbb{R}^3)$ the positivity of $\hat v$ implies that 
\begin{equation}
\mathcal{F}^{\mathrm{H}}(\gamma) \geq \tr\left[ \left( - \hbar^2 \Delta + \frac{\omega^2 x^2}{4} + \frac{1}{N} v \ast \eta(x) \right) \gamma \right] - 
\frac{1}{\beta} s(\gamma) - \frac{1}{N} D(\eta,\eta).
\label{eq:propsemiclassical4a}
\end{equation}
 The first term on the right-hand side of \eqref{eq:propsemiclassical4a} equals
\begin{align}
&\tr\left[ \left( - \hbar^2 \Delta + \frac{\omega^2 x^2}{4} + \frac{1}{N} v \ast \eta(x)  \right) \gamma \right]  \label{eq:propsemiclassical4} \\
&\hspace{0.2cm} = \frac{1}{(2 \pi \hbar)^3} \int_{\mathbb{R}^6} \left( p^2 + \frac{\omega^2 x^2}{4} + \frac{1}{N} v \ast \eta(x) \right) m_{\gamma}(p,x) \de(p,x)  \nonumber \\
&\hspace{0.6cm}- N \hbar \int_{\mathbb{R}^3} \left| \nabla \ell(x) \right|^2 \de x + \int_{\mathbb{R}^3} \varrho_{\gamma} (x) \left( \frac{\omega^2 x^2}{4} + \frac{1}{N} v \ast \eta(x) - \left( \frac{\omega^2 (\cdot)^2}{4} + \frac{1}{N} v \ast \eta \right) \ast  | \ell_{0,0}^{\hbar} |^2(x) \right) \de x, \nonumber
\end{align}
see \cite[Eq.~(5.20)]{Li1981} or \cite[Corollary~2.5]{FouLeSol2018}.
 The fact that $\mathcal{\ell}$ is a radial function implies that
\begin{equation}
\int_{\mathbb{R}^3} \varrho_{\gamma} (x) \left( \frac{\omega^2 x^2}{4} - \frac{\omega^2 (\cdot)^2}{4} \ast  | \ell_{0,0}^{\hbar} |^2(x) \right) \de x = -N \hbar \int_{\mathbb{R}^3} \left| \ell(x) \right|^2 \frac{\omega^2 x^2}{4} \de x.
\label{eq:propsemiclassical5}
\end{equation}
Similarly, we find that
\begin{equation}
\frac{1}{N} \int_{\mathbb{R}^3} \varrho_{\gamma} (x) \left( v \ast \eta(x)  -  \left( v \ast \eta \right) \ast  | \ell_{0,0}^{\hbar} |^2(x) \right) \de x \geq -\hbar \eta\left( \mathbb{R}^3 \right) \left( \sup_{x \in \mathbb{R}^3} \left\Vert D^2 v(x) \right\Vert \right) \int_{\mathbb{R}^3}  | \ell(y) |^2 y^2 \de y \,.
\label{eq:propsemiclassical6}
\end{equation}
We  choose $\eta$ such that  $\eta( \mathbb{R}^3 ) = N$. Hence \eqref{eq:propsemiclassical4} is bounded from below by the terms in the second line minus a correction of the order $N \hbar \omega$. 

The completeness relation for coherent states in \eqref{eq:resid}, Jensen's inequality and the convexity of $f$ implies the Berezin--Lieb inequality
\begin{equation}
-s(\gamma) = \trs\left[ f(\gamma) \right] \geq \frac{1}{(2 \pi \hbar)^3} \int_{\mathbb{R}^6} f\left( m_{\gamma}(p,x) \right) \de(p,x),
\label{eq:propsemiclassical8}
\end{equation}
see also \cite{Berezin,Lieb1973,KlauSkag2012}. In combination with \eqref{eq:propsemiclassical4a}--\eqref{eq:propsemiclassical6} we thus have the lower bound
\begin{align}
\mathcal{F}^{\mathrm{H}}(\gamma) & \geq \frac{1}{(2 \pi \hbar)^3} \int_{\mathbb{R}^6} \left( p^2 + \frac{\omega^2 x^2}{4} + \frac{1}{N} v \ast \eta(x) - \frac{1}{N} v \ast \eta(0) - \mu \right) m_{\gamma}(p,x) \de(p,x) - \frac{1}{ \beta \hbar^3} S^{\mathrm{sc}} \left( m_{\gamma} \right) \label{eq:propsemiclassical9} \\
& \quad  - \frac{D(\eta,\eta)}{N} + N \left( \frac{1}{N} v \ast \eta(0) + \mu \right) - \text{const.}\, N\hbar \omega  \nonumber
\end{align}
where we added and subtracted a constant $ N \left( \frac{1}{N} v \ast \eta(0) + \mu \right)$ for convenience. 

We denote by $\varrho^{\mathrm{sc}}$ and $\mu^{\mathrm{sc}}$ the density and the chemical potential of the minimizing pair $(\gamma^{\mathrm{sc}},g^{\mathrm{sc}})$, and choose $\eta = N \varrho^{\mathrm{sc}}$ as well as $\mu = \mu^{\mathrm{sc}}$. With this choice  the first two terms on the right-hand side of \eqref{eq:propsemiclassical9} are given by
\begin{equation}
\frac{1}{\beta \hbar^3} \mathcal{S}^{\mathrm{sc}}\left( m_{\gamma}, \gamma^{\mathrm{sc}} \right) + \frac{1}{\beta (2 \pi \hbar)^3} \int_{\mathbb{R}^6} \ln\left( 1 - \exp\left( -\beta \left( p^2 + \frac{\omega^2 x^2}{4} + \frac{1}{N} v \ast \varrho^{\mathrm{sc}}(x) - \frac{1}{N} v \ast \varrho^{\mathrm{sc}}(0) - \mu^{\mathrm{sc}} \right) \right) \right) \de (p,x) \label{eq:propsemiclassical11} 
\end{equation}
where the semiclassical relative entropy $\mathcal{S}^{\mathrm{sc}}$ for two nonnegative integrable functions $a$ and $b$ on $\mathbb{R}^6$ equals
\begin{equation}
\mathcal{S}^{\mathrm{sc}}(a,b) = \frac{1}{(2 \pi)^3} \int_{\mathbb{R}^6} \left[ f(a(p,x)) - f(b(p,x)) - f'(b(p,x)) \left( a(p,x) - b(p,x) \right) \right] \de (p,x).
\label{eq:propsemiclassical10}
\end{equation}
Because of the convexity of $f$ the integrand is nonnegative, hence $\mathcal{S}^{\mathrm{sc}}$ is always well-defined (if we allow it to take the value $+\infty$).
When we combine \eqref{eq:propsemiclassical9}, \eqref{eq:propsemiclassical11} and \eqref{eq:semiclassicalfunctional5b} in Lemma~\ref{lem:semiclassicalfreeenergy}, we arrive at the lower bound
\begin{equation}\label{eq:propsemiclassical12}
\hbar^3 \mathcal{F}^{\mathrm{H}}(\gamma) \geq F^{\mathrm{sc}}(\beta,\omega) + \frac{1}{\beta} \mathcal{S}^{\mathrm{sc}}\left( m_{\gamma}, \gamma^{\mathrm{sc}} \right) - {\text{const.}}\, \hbar \omega  \,. 
\end{equation}
The semiclassical relative entropy can be dropped for a lower bound. In combination with \eqref{eq:53} we thus obtain \eqref{eq:propsemiclassical1a}. In Section~\ref{sec:weakBEC} we will use \eqref{eq:propsemiclassical12} with the term $\mathcal{S}^{\mathrm{sc}}( m_{\gamma}, \gamma^{\mathrm{sc}} )$ in order to prove Corollary~\ref{cor:weakBEC}.
\subsection{Asymptotics of the minimizer of $\mathcal{F}^{\mathrm{H}}$}
In this Section we prove the claimed asymptotics for the minimizer $\gamma^{\mathrm{H}}$ of $\mathcal{F}^{\mathrm{H}}$ in Theorem~\ref{thm:limitHartreetheory}. For this purpose we shall need a refined lower bound for $\mathcal{F}^{\mathrm{H}}$. This is necessary because it is the Husimi function of $Q^{\mathrm{H}} \gamma^{\mathrm{H}}$ that converges to $\gamma^{\mathrm{sc}}$ and not the one of $\gamma^{\mathrm{H}}$. Here $Q^{\mathrm{H}}$ denotes the projection onto the orthogonal complement of the eigenspace corresponding to the largest eigenvalue of $\gamma^{\mathrm{H}}$. This makes it necessary to obtain a bound with $m_{\gamma}$ replaced by $m_{Q\gamma}$ in the semiclassical relative entropy in \eqref{eq:propsemiclassical12}, where $Q$ is defined as in the case of $\gamma^{\mathrm{H}}$ but w.r.t. $\gamma$. We start by deriving this bound. 

\subsubsection*{Refined lower bound for $\mathcal{F}^{\mathrm{H}}$}

Let $\gamma \in \mathcal{D}_N^{\mathrm{H}}$, denote by $P$ the projection onto the eigenspace of its largest eigenvalue and define $Q = 1 - P$. Since $- \hbar^2 \Delta + \frac{\omega^2 x^2}{4} +  v \ast \varrho^{\mathrm{sc}}(x) -  v \ast  \varrho^{\mathrm{sc}}(0) - \mu^{\mathrm{sc}} \geq 0$, which follows from \eqref{eq:34}, we have
\begin{equation}
\tr\left[ \left(- \hbar^2 \Delta + \frac{\omega^2 x^2}{4} +  v \ast \varrho^{\mathrm{sc}}(x) -  v \ast  \varrho^{\mathrm{sc}}(0) - \mu^{\mathrm{sc}}  \right) \gamma \right]  \geq \tr\left[ \left(- \hbar^2 \Delta + \frac{\omega^2 x^2}{4} +  v \ast \varrho^{\mathrm{sc}}(x) -  v \ast  \varrho^{\mathrm{sc}}(0) - \mu^{\mathrm{sc}}  \right) Q\gamma \right] 
\end{equation}
and we can then proceed as in \eqref{eq:propsemiclassical4}--\eqref{eq:propsemiclassical6} above. 
For the next step we need the following Lemma.
\begin{lemma}
	\label{lem:convexity}
	Let $A$ be a self-adjoint operator on a Hilbert space $\mathcal{H}$ whose spectrum consists of eigenvalues, and let $Q$ be an orthogonal projection. By $\sigma(A)$ we denote the spectrum of $A$ and by $\text{Conv}\left( \sigma(A) \right)$ its convex hull. Let $f : \text{Conv}\left( \sigma(A) \right) \to \mathbb{R}$ be a convex function and assume that $Q f(A) Q$ and $f(QAQ)$ are trace-class. Then we have
	\begin{equation}
	\trs[Q f(A) Q] \geq \trs[f(QAQ)].
	\end{equation} 
	\end{lemma}
	
	\begin{proof}
		Denote by $\{ a_{i} \}_{i=1}^{\infty}$ and $\{ v_{i} \}_{i=1}^{\infty}$ the eigenvalues and eigenvectors of $A$. Moreover, let $\{ w_{j} \}_{j=1}^{\infty}$ be the eigenvectors of $QAQ$ in the range of $Q$ and let $\{ \widetilde{a}_{j} \}_{j=1}^{\infty}$ be the corresponding eigenvalues. An application of Jensen's inequality yields
		\begin{align}
		\trs[Q f(A) Q] &= \sum_{j=1}^{\infty} \langle w_{j}, f(A) w_{j} \rangle = \sum_{j=1}^{\infty} \sum_{i=1}^{\infty} f(a_i) | \langle w_j, v_i \rangle |^2 \label{eq:33} \\
		&\geq \sum_{j=1}^{\infty} f\left( \sum_{i=1}^{\infty} a_i | \langle w_j, v_i \rangle |^2 \right) = \sum_{j=1}^{\infty} f\left( \langle w_j, A w_j \rangle \right) = \trs[f(QAQ)]. \nonumber
		\end{align}
	\end{proof}

The bosonic entropy function $f$ in \eqref{eq:bosonicentropy} is convex and decreasing. Hence Lemma~\ref{lem:convexity} implies that
\begin{equation}
- s(\gamma)  = \trs\left[ P f(\gamma) P \right] + \trs\left[ Q f(\gamma) Q \right] \geq f(N) + \trs\left[ f(Q \gamma Q) \right].
\label{eq:propsemiclassical7}
\end{equation}
We use the above considerations, \eqref{eq:propsemiclassical4}--\eqref{eq:propsemiclassical8} with the obvious replacements and \eqref{eq:propsemiclassical11} to arrive at
\begin{equation}\label{eq:propsemiclassical12b}
\hbar^3 \mathcal{F}^{\mathrm{H}}(\gamma) \geq F^{\mathrm{sc}}(\beta,\omega) + \frac{1}{\beta} \mathcal{S}^{\mathrm{sc}}\left( m_{Q \gamma}, \gamma^{\mathrm{sc}} \right) - \frac{\text{const.} \ln N}{\beta N} - \text{const.}\, \hbar\omega
\end{equation}

\subsubsection*{Asymptotics of $\gamma^{\mathrm{H}}$}

From \eqref{eq:53} and \eqref{eq:propsemiclassical12b} we know that the Husimi function of $\gamma^{\mathrm{H}}$ obeys the bound
\begin{equation}
\mathcal{S}^{\mathrm{sc}}\left( m_{Q^{\mathrm{H}} \gamma^{\mathrm{H}}}, \gamma^{\mathrm{sc}} \right) \lesssim  \beta \hbar \omega.
\label{eq:propsemiclassicalminimizer1}
\end{equation}
In order to obtain the claimed bound on the $L^1(\mathbb{R}^6)$-norm distance of $m_{Q^{\mathrm{H}} \gamma^{\mathrm{H}}}$ and $\gamma^{\mathrm{sc}}$, we need the following Lemma. It is motivated by \cite[Lemma~4.1]{me,me2}.
\begin{lemma}
	\label{lem:semiclassicalcoercivity}
	There exists a constant $C>0$ such that for any two nonnegative functions $a, b \in L^1(\mathbb{R}^6)$ we have 
	\begin{equation}
	\mathcal{S}^{\mathrm{sc}}\left( a, b \right) \geq C \frac{\left( \int_{\mathbb{R}^6} \left| a(p,x) - b(p,x) \right| \de(p,x) \right)^2}{\int_{\mathbb{R}^6} \left( a(p,x) + b(p,x) \right) \left( 1 + b(p,x) \right) \de (p,x)}.
	\label{eq:propsemiclassicalminimizer2}
	\end{equation}
\end{lemma}
\begin{proof}
	From \cite[Eqs.~(4.11) \& (4.12)]{me2} we know that
	\begin{equation}
	f(x) - f(y) - f'(y)(x-y) \geq \text{const. } \frac{\left( x- y \right)^2}{(1+y)(x+y)} \geq \text{const. } \frac{\left( \sqrt{x}- \sqrt{y} \right)^2}{1+y} 
	\label{eq:propsemiclassicalminimizer3}
	\end{equation}
	holds for all $x,y > 0$. Hence,
	\begin{equation}
	\mathcal{S}^{\mathrm{sc}}\left( a, b \right) \geq \text{const. } \int_{\mathbb{R}^6}  \frac{\left( a(p,x)^{1/2}- b(p,x)^{1/2} \right)^2}{1+b(p,x)} \de(p,x).
	\label{eq:propsemiclassicalminimizer4}
	\end{equation}
	With the aid of the Cauchy--Schwarz inequality, we also see that
	\begin{align}
	&\int_{\mathbb{R}^6} \left| a(p,x) - b(p,x) \right| \de(p,x) \leq \left( \int_{\mathbb{R}^6} \frac{\left| a(p,x)^{1/2} - b(p,x)^{1/2} \right|^2}{1+b(p,x)} \de (p,x) \right)^{1/2} \label{eq:propsemiclassicalminimizer5} \\
	&\hspace{6cm} \times \left( \int_{\mathbb{R}^6} \left( a(p,x)^{1/2} + b(p,x)^{1/2} \right)^2 ( 1 + b(p,x) ) \de (p,x) \right)^{1/2}. \nonumber
	\end{align}
	For the term in the second line we use the upper bound $\left( a(p,x)^{1/2} + b(p,x)^{1/2} \right)^2 \leq 2 a(p,x) + 2 b(p,x)$. In combination, \eqref{eq:propsemiclassicalminimizer4} and \eqref{eq:propsemiclassicalminimizer5} imply the claim.
\end{proof}
By combining Lemma~\ref{lem:semiclassicalcoercivity} and the bound \eqref{eq:propsemiclassicalminimizer1}, we find
\begin{equation}
\int_{\mathbb{R}^6} \left| m_{Q^{\mathrm{H}} \gamma^{\mathrm{H}}}(p,x) - \gamma^{\mathrm{sc}}(p,x) \right| \de(p,x)  \label{eq:propsemiclassicalminimizer6}  \lesssim \left(\hbar \beta\omega \right)^{1/2} \left( \int_{\mathbb{R}^6} \left( m_{Q^{\mathrm{H}} \gamma^{\mathrm{H}}}(p,x) + \gamma^{\mathrm{sc}}(p,x) \right) \left( 1 + \gamma^{\mathrm{sc}}(p,x) \right) \de(p,x) \right)^{1/2}. 
\end{equation}
From \eqref{eq:semiclassicalfunctional5} and \eqref{eq:34} we see that 
\begin{equation}
	\Vert \gamma^{\mathrm{sc}} \Vert_{L^q(\mathbb{R}^3)} \leq \left( \int_{\mathbb{R}^6} \left( \frac{1}{ \exp \left( \beta \left( p^2 + c \omega^2 x^2 \right) - 1 \right) } \right)^q \de(p,x) \right)^{1/q} \lesssim \left( \frac{1}{\beta \omega} \right)^{3/q}
	\label{eq:54}
\end{equation}
holds for some $c>0$ if $1 \leq q < 3$. It remains to bound the integrals over $m_{Q^{\mathrm{H}} \gamma^{\mathrm{H}}}$ and over $m_{Q^{\mathrm{H}} \gamma^{\mathrm{H}}}$ times $\gamma^{\mathrm{sc}}$ in \eqref{eq:propsemiclassicalminimizer6}. For the latter integral we use the bound
\begin{equation}
\int_{\mathbb{R}^6} m_{Q^{\mathrm{H}} \gamma^{\mathrm{H}}}(p,x) \gamma^{\mathrm{sc}}(p,x) \de(p,x) \leq \left\Vert m_{Q^{\mathrm{H}} \gamma^{\mathrm{H}}} \right\Vert_{L^{\infty}(\mathbb{R}^6)}^{1-1/p} \left\Vert m_{Q^{\mathrm{H}} \gamma^{\mathrm{H}}} \right\Vert_{L^{1}(\mathbb{R}^6)}^{1/p} \left\Vert \gamma^{\mathrm{sc}} \right\Vert_{L^{q}(\mathbb{R}^6)},
\label{eq:propsemiclassicalminimizer7}
\end{equation}
which follows from H\"older's inequality for $1 = p^{-1} + q^{-1}$. To obtain a bound for the $L^{\infty}(\mathbb{R}^6)$-norm of $m_{Q^{\mathrm{H}} \gamma^{\mathrm{H}}}$ we need the following Lemma, which provides us with an estimate for the spectral gap of the Hartree operator.

\begin{lemma}
	\label{lem:spactralgapSCMF}
	Let $\varrho : \mathbb{R}^3 \to \mathbb{R}$ be a nonnegative function with $\int_{\mathbb{R}^3} \varrho(x) \de x = 1$ and $\varrho(-x) = \varrho(x)$, and denote by $\Delta e$ the spectral gap of the operator $H = -\hbar^2 \Delta + \omega^2 x^2/4 + v \ast \varrho(x)$ above its unique ground state. Under Assumptions~\ref{as:regularitypotential}  on $v$ there exists a constant $c>0$ independent of $\varrho$, $\omega$ and $\hbar$ such that 
		\begin{equation}
			\Delta e \geq c \hbar \omega \,.
		\end{equation}
\end{lemma}

\begin{proof}
	The assumption on the Hessian of $v$ in \eqref{eq:condintpot} guarantees that there exists a constant $c>0$ such that the potential $\omega^2 x^2/4 + v \ast \varrho(x) - c \omega^2 x^2$ is convex. Accordingly, we have $H = - \hbar^2 \Delta + c \omega^2 x^2 + W(x)$, where $W(x)$ is convex. It follows from the work of Brascamp and Lieb \cite{BraLi1976} that $\Delta e$ is bounded from below by the spectral gap of the operator $- \hbar^2 \Delta + c \omega^2 x^2$, which is given by $\hbar \omega \sqrt{c}$. In order to see this, we note that the ground state of $H$ can be written as 
	\begin{equation}
		\psi_0(x) = \left( \frac{\omega \sqrt{c}}{2 \pi} \right)^{3/4} \exp\left( -\frac{ \sqrt{c} \omega x^2}{4} \right) \phi(x)  
		\label{eq:141}
	\end{equation} 
	with a log concave function $\phi$, see \cite[Theorem~6.1]{BraLi1976}. The function on the right-hand side multiplying $\phi(x)$ is the ground state of $H - W(x)$. For the first excited state of $H$ we make the ansatz $\psi(x) = \psi_0(x) \chi(x)$ and note that its energy can be written as
	\begin{equation}
		\langle \psi, H \psi \rangle = e_0 + \int_{\mathbb{R}^3} \left| \nabla \chi(x) \right|^2 \psi_0(x)^2 \de x,
		\label{eq:142}
	\end{equation}
	where $e_0$ denotes the ground state energy of $H$. With \eqref{eq:142} we conclude that the spectral gap of  $H$ can be written as
	\begin{equation}
		\Delta e = \inf_{\Vert \psi_0 \chi \Vert = 1, \ \langle \psi_0^2, \chi \rangle = 0} \int_{\mathbb{R}^3} \left| \nabla \chi(x) \right|^2 \psi_0(x)^2 \de x.
		\label{eq:143}
	\end{equation} 
	From \cite[Theorem~4.1]{BraLi1976} and $\int_{\mathbb{R}^3} \psi^2_0(x) \chi(x) \de x = 0$ we know that
	\begin{equation}
		\int_{\mathbb{R}^3} \psi^2_0(x) \left| \chi(x) \right|^2 \de x \leq - \left\langle \nabla \chi, \left[ D^2 \ln\left( \psi^2_0 \right) \right]^{-1} \nabla \chi \right\rangle,
		\label{eq:144}
	\end{equation}
	where $[ D^2 \ln( \psi^2_0 ) ]^{-1}$ denotes the inverse of the Hessian of $\ln(\psi_0^2)$. Since $\phi$ is log concave, 
	\begin{equation}
		- D^2 \ln\left( \psi^2_0(x) \right) = \left( \hbar \omega \sqrt{c} \right)^{-1} - D^2 \ln\left( \phi^2(x) \right) \geq \left( \hbar \omega \sqrt{c} \right)^{-1}
		\label{eq:145}
	\end{equation}
	holds. In combination, \eqref{eq:144} and \eqref{eq:145} imply
	\begin{equation}
		\int_{\mathbb{R}^3} \left| \nabla \chi(x) \right|^2 \psi_0(x)^2 \de x \geq \hbar \omega \sqrt{c},
		\label{eq:146}
	\end{equation}
	and prove the claim that the spectral gap of $H$ is at least as large as the one of $H-W(x)$. This proves Lemma~\ref{lem:spactralgapSCMF}. 
\end{proof}

To apply Lemma~\ref{lem:spactralgapSCMF} to the Hartree operator $-\hbar^2 \Delta + \omega^2 x^2/4 + N^{-1} v \ast \varrho_{\gamma^{\mathrm{H}}}(x)$, we note that the uniqueness of the minimizer of $\mathcal{F}^{\mathrm{H}}$, see Lemma~\ref{lem:minimizersHartree}, implies that $\varrho_{\gamma^{\mathrm{H}}}(-x) = \varrho_{\gamma^{\mathrm{H}}}(x)$. Using Lemma~\ref{lem:spactralgapSCMF}, we can therefore estimate the $L^{\infty}(\mathbb{R}^3)$-norm of $m_{Q^{\mathrm{H}} \gamma^{\mathrm{H}}}$  as 
\begin{align}
\sup_{(p,q) \in \mathbb{R}^6} \left\langle \ell^{\hbar}_{p,q}, Q^{\mathrm{H}} \gamma^{\mathrm{H}} \ell^{\hbar}_{p,q} \right) &\leq \sup_{\Vert \psi \Vert = 1} \left\langle \psi, Q^{\mathrm{H}} \frac{1}{\exp\left( \beta\left( -\hbar^2 \Delta + \frac{\omega^2 x^2}{4} + \frac{1}{N} v \ast \varrho_{\gamma^{\mathrm{H}}}(x) - \mu^{\mathrm{H}} \right)\right) - 1} \psi \right\rangle \label{eq:propsemiclassicalminimizer8} \\
&\lesssim ( \beta \hbar \omega)^{-1}, \nonumber
\end{align}
where  $\mu^{\mathrm{H}}$ denotes the Hartree chemical potential.  

Next we derive a bound for the $L^1(\mathbb{R}^6)$-norm of $m_{Q^{\mathrm{H}} \gamma^{\mathrm{H}}}$, which satisfies
\begin{equation}
	\left( \frac{1}{2 \pi \hbar} \right)^3 \int_{\mathbb{R}^6} m_{Q^{\mathrm{H}} \gamma^{\mathrm{H}}}(p,x) \de(p,x) = \trs\left[Q^{\mathrm{H}} \frac{1}{\exp\left( \beta \left( -\hbar^2 \Delta + \frac{\omega^2 x^2}{4} + \frac{1}{N} v \ast \varrho_{\gamma^{\mathrm{H}}}(x) - \mu^{\mathrm{H}} \right) \right) - 1 } \right].
	\label{eq:61}
\end{equation}
To proceed, we need the following Lemma.
\begin{lemma}
	\label{lem:semiclassicaltraceestimate}
	Let $\varrho : \mathbb{R}^3 \to \mathbb{R}$ be a nonnegative function with $\int_{\mathbb{R}^3} \varrho(x) \de x = 1$ and $\varrho(-x) = \varrho(x)$, and let  
	\begin{equation}
		H = - \hbar^2 \Delta +   \frac{\omega^2}{4} x^2  + v*\rho(x) - v*\rho(0) \,.
		\label{eq:2b}
	\end{equation}
	By $Q$ we denote the spectral projection of $H$ onto the orthogonal complement of its ground state subspace, and by $\mu$ we denote the chemical potential satisfying the equation
	\begin{equation}
		\trs\left[ \frac{1}{\exp\left( \beta \left( H - \mu \right) \right) - 1 } \right] = N = \hbar^{-3}\,.
	\end{equation}
	Let $\zeta : \mathbb{R}_+ \to \mathbb{R}$ be a nonnegative and monotone decreasing function behaving as $\zeta(x) \sim x^{-\alpha}$ for $x \to 0$ with $0 \leq \alpha < 3$ and as $\zeta(x) \sim x^{-\nu}$ for $x \to \infty$ with $\nu > 3$. Then
	\begin{equation}
		\trs\left[ Q \zeta( \beta(H - \mu)) \right] \lesssim  ( \beta \hbar \omega)^{-3} \,.
		\label{eq:65}
	\end{equation}
\end{lemma}
\begin{proof}
	We first show that the chemical potential $\mu$ obeys the bound
	\begin{equation}
	\mu \leq \sqrt{2} \hbar \omega\,.
	\label{eq:58}
	\end{equation}
	This follows from the fact that $v*\rho(x) - v*\rho(0) \leq \omega^2 x^2/4$ (compare with \eqref{eq:34} and \eqref{eq:34n}), hence the ground state energy of $H$ is bounded from above by $\sqrt{2} \hbar\omega$. 
	
	We denote by $\{e_j\}_{j=0}^\infty$ the eigenvalues of $H$ and by $\{\widetilde{e}_j\}_{j=0}^\infty$ the eigenvalues of the operator $-\hbar^2 \Delta + c \omega^2 x^2$ with $c>0$ chosen such that \eqref{eq:34} holds. Next, we choose $M \in \mathbb{N}$ such that $e_{j} - \mu \geq \widetilde{e}_{j}/2$ for all $j > M$. With \eqref{eq:58} and $e_j \geq \widetilde{e}_j$ for all $j \geq 0$, we see that $M$ can be chosen independently of $\hbar$ and $\varrho$. 
	The monotonicity of $\zeta$ then implies that 
	\begin{equation}
		\trs\left[ Q \zeta( \beta(H - \mu)) \right] \leq \sum_{j=1}^M \zeta(\beta ( e_j - \mu )) + \sum_{j=M+1}^{\infty} \zeta (\beta \ \widetilde{e}_j /2 ).
		\label{eq:64}
	\end{equation}
	An application of Lemma~\ref{lem:spactralgapSCMF} shows that the first term on the right-hand side of \eqref{eq:64} is bounded from above by a constant times $M ( \beta \hbar \omega )^{-\alpha}$ for $\beta \hbar \omega \leq 1$. If $\beta \hbar \omega > 1$ it is bounded from above by a constant times $M ( \beta \hbar \omega )^{-\nu}$.
	
	To estimate the second term on the right-hand side of \eqref{eq:64}, we note that the eigenvalues $\widetilde{e}_j$ are given by $\sqrt{c} \hbar \omega (n+3/2)$ with $n \in \mathbb{N}_0$, and are $(n+1)(n+2)/2$-fold degenerate. When we interpret the relevant sum as a lower Riemann sum for the corresponding integral, we find 
	the bound
	\begin{equation}
	\sum_{j=M+1}^{\infty} \zeta \left( \beta \ \widetilde{e}_j \right) = \frac{1}{2} \sum_{n \geq m}  (n+1)(n+2) \zeta\left(  \beta \hbar \omega \sqrt{c} (n+3/2)/2 \right) \lesssim  \frac{1}{\left( \beta \hbar \omega\right)^{3}}.
	\label{eq:66}
	\end{equation}
	The first identity holds for an appropriately chosen $m \in \mathbb{N}_0$ depending only on $M$. 
This completes the proof of \eqref{eq:65}. 
\end{proof}

Applying Lemma~\ref{lem:semiclassicaltraceestimate} to the Hartree operator yields the bound
\begin{equation}
	\Vert m_{Q^{\mathrm{H}} \gamma^{\mathrm{H}}} \Vert_{L^1(\mathbb{R}^6)} = (2 \pi \hbar)^3 \trs[Q^{\mathrm{H}} \gamma^{\mathrm{H}}] \lesssim  ( \beta  \omega)^{-3}.
	\label{eq:59}
\end{equation}
With \eqref{eq:54}, \eqref{eq:propsemiclassicalminimizer7}, \eqref{eq:propsemiclassicalminimizer8} and \eqref{eq:59} we also conclude that
\begin{equation}
\int_{\mathbb{R}^6} m_{Q^{\mathrm{H}} \gamma^{\mathrm{H}}}(p,x) \gamma^{\mathrm{sc}}(p,x) \de(p,x) \lesssim \hbar^{-1/q} \left( \beta \omega \right)^{-3-1/q}
\label{eq:propsemiclassicalminimizer9}
\end{equation}
for any $1 \leq q < 3$. From \eqref{eq:propsemiclassicalminimizer6}, \eqref{eq:54} and \eqref{eq:propsemiclassicalminimizer9} we obtain the final estimate
\begin{equation}
\int_{\mathbb{R}^6} \left| m_{Q^{\mathrm{H}} \gamma^{\mathrm{H}}}(p,x) - \gamma^{\mathrm{sc}}(p,x) \right| \de(p,x) \lesssim \frac{\hbar^{1/2}}{\beta \omega} + \hbar^{(1-1/q)/2} \left( \frac{1}{\beta \omega} \right)^{1+1/(2q)} \,. \label{eq:propsemiclassicalminimizer10} 
\end{equation}
For $\beta \omega \gtrsim 1$, the choice $q = 3 - \sigma$ with $\sigma > 0$ implies
\begin{equation}
	\int_{\mathbb{R}^6} \left| m_{Q^{\mathrm{H}} \gamma^{\mathrm{H}}}(p,x) - \gamma^{\mathrm{sc}}(p,x) \right| \de(p,x) \lesssim \hbar^{1/2} + \hbar^{1/3-\sigma} 
\end{equation} 
and proves the second part of \eqref{eq:propsemiclassical2a}. To prove the first part, we use $(\trs[P^{\mathrm{H}} \gamma^{\mathrm{H}}] + \trs[Q^{\mathrm{H}} \gamma^{\mathrm{H}} ])/N=1$ with $P^{\mathrm{H}} = 1 - Q^{\mathrm{H}}$, $(1/(2\pi))^3 \int_{\mathbb{R}^6} \gamma^{\mathrm{sc}}(p,x) \de(p,x) + g = 1$ and the resolution of the identity \eqref{eq:resid} to find
\begin{equation}
\left| \frac{\trs[P^{\mathrm{H}} \gamma^{\mathrm{H}}]}{N} - g^{\mathrm{sc}} \right| = \frac{1}{(2 \pi )^3}  \left| \int_{\mathbb{R}^6}  \left( m_{Q^{\mathrm{H}} \gamma^{\mathrm{H}}}(p,x) - \gamma^{\mathrm{sc}}(p,x) \right) \de(p,x)\right| \,.
\label{eq:propsemiclassicalminimizer11}
\end{equation}
Hence \eqref{1.36a} follows immediately from \eqref{1.36b}. This concludes the proof of Theorem~\ref{thm:limitHartreetheory}.
\section{Bounds on the  free energy}
\label{sec:freeenergybounds}
In this section we prove upper and lower bounds on the  free energy \eqref{eq:freeenergycan} of the full quantum model that imply the claimed free energy asymptotics in Theorem~\ref{thm:freeenergyscaling2}. An important ingredient for our bounds is the canonical version of the Hartree free energy functional introduced in Section~\ref{sec:Hartreefreeenergyfunctionals}, and the bound on the difference of the canonical and the grand-canonical Hartree free energies in Lemma~\ref{lem:freeenergyboundcgc}. Our bounds apply both in the canonical and the grand-canonical setting.
\subsection{Lower bound}
\label{sec:prooflowerboundgc}
In the first step, we apply a standard technique to utilize the assumed positivity of $\hat v$  in order to replace the two-body interaction potential by an effective one-body potential.
\begin{lemma}
	\label{lem:lowerboundinteraction}
	For any $\eta \in L^1(\mathbb{R}^3)$ and any $N \in \mathbb{N}_0$ we have	
	\begin{equation}
	\sum_{1 \leq i < j \leq N} v(x_i-x_j) \geq \sum_{i=1}^N \eta \ast v(x_i) - D(\eta,\eta) - \frac{N v(0)}{2}.
	\label{eq:lowerboundgc1}
	\end{equation}
\end{lemma} 
\begin{proof}
	Since $v$ is of positive type, we know that $\int_{\mathbb{R}^6} v(x-y) \de \nu(x) \de \nu(y) \geq 0$ holds for any signed Borel measure $\nu$ on $\mathbb{R}^3$. The statement of the Lemma follows from the choice $\nu(x) = \sum_{i=1}^N \delta_{x_i}(x) - \eta(x) \de x$, where $\delta_{y}(x)$ denotes the Dirac delta measure with unit mass at the point $y$.
\end{proof}
For the second quantized interaction potential this implies the lower bound
\begin{equation}
	\mathcal{V}_N \geq \de \Upsilon\left( N^{-1} \eta \ast v \right) - N^{-1} D(\eta, \eta) - \frac{ \mathcal{N} v(0)}{2 N} \,.
	\label{eq:lowerboundgc2}
\end{equation}
Using \eqref{eq:lowerboundgc2}, we see that the free energy of a given state $\Gamma \in \mathcal{S}_N^{\mathrm{gc}}$ is bounded from below by
\begin{align}
	\tr\left[ \mathcal{H} \Gamma \right] - \frac{1}{\beta} S(\Gamma) &\geq \tr\left[ \de \Upsilon \left( h + N^{-1} v \ast \eta (x) \right) \Gamma \right] - \frac{1}{\beta} S(\Gamma) - N^{-1} D(\eta, \eta) - \frac{ v(0)}{2} \label{eq:lowerboundgc3} \\
	&\geq \frac{1}{\beta} \tr \left[ \ln\left( 1 - e^{ - \beta \left(   h + N^{-1} v \ast \eta - \mu \right) } \right) \right] + \mu N - N^{-1} D(\eta, \eta) - \frac{ v(0)}{2}. \nonumber
\end{align}
The second inequality follows from the Gibbs variational principle and holds for any choice of  the chemical potential $\mu$. In particular, we can choose $\mu$ such that 
\begin{equation}
	\tr \left[ \frac{1}{e^{\beta \left(   h + N^{-1} v \ast \eta - \mu \right)} -1} \right] = N\,.
	\label{eq:lowerboundgc4}
\end{equation}
To obtain the optimal lower bound we take the supremum over all pairs $(\eta,\mu) \in \mathcal{D}_N^{\mathrm{gc}}$, with $\mathcal{D}_N^{\mathrm{gc}}$ defined in \eqref{eq:sec4secondlemmav21}. In combination with Lemma~\ref{lem:secondvariationalcharacterization}, this implies
\begin{equation}
	\tr\left[ \mathcal{H} \Gamma \right] - \frac{1}{\beta} S(\Gamma) \geq F^{\mathrm{H}}(\beta,N,\omega) - \frac{v(0)}{2}.
	\label{eq:lowerboundgc5}
\end{equation}
\subsection{Upper bound}
\label{sec:proofupperboundc}
To prove the upper bound, we choose the minimizer $G^{\mathrm{H},\mathrm{c}}$ in \eqref{eq:HartreeGibbsstatec} of the canonical version of the Hartree free energy functional $\mathcal{F}^{\mathrm{H},\mathrm{c}}$ in \eqref{eq:Hartreefunctional} as a trial state. To simplify the notation, we drop superscripts and denote this state in the following by $G$. We start the analysis by considering the interaction energy of $G$.
\subsubsection*{Interaction energy of $G$}
The interaction energy can be written in terms of the 2-particle density $\varrho_G^{(2)}$ of $G$ as
\begin{equation}
N^{-1} \tr\left[ \sum_{1 \leq i < j \leq N} v(x_i-x_j) \ G \right] = \frac{1}{2 N} \int_{\mathbb{R}^6} v(x - y) \varrho_G^{(2)}(x,y) \de(x,y).
\label{eq:upperboundcan1}
\end{equation}
Here $\varrho_G^{(2)}$ is normalized such that $\int_{\mathbb{R}^6} \varrho_G^{(2)}(x,y) \de(x,y) = N(N-1)$. In order to give a bound on $\varrho_G^{(2)}(x,y)$, we write it as
\begin{equation}
	\varrho_G^{(2)}(x,y) = \sum_{j \geq 0} | \varphi_j(x) |^2 | \varphi_j(y) |^2 \left\langle n_j (n_j - 1)) \right\rangle_{G} + 2 \sum_{\stackrel{i,j \geq 0}{i \neq j}} \left| \frac{1}{2} \sum_{\sigma \in S_2} \varphi_{\sigma(i)}(x) \varphi_{\sigma(j)}(y) \right|^2 \left\langle n_i n_j \right\rangle_{G}, \label{eq:upperboundcan2}
\end{equation}
where  $\{ \varphi_j \}_{j=0}^{\infty}$ denote the eigenfunctions of the Hartree operator $h + N^{-1} v \ast \varrho_G(x)$ (ordered with increasing energy), and $n_j$ counts the number of particles in the state $\varphi_j$, i.e, $n_j = \de \Upsilon( |\varphi_j\rangle\langle \varphi_j|)$.  We note that  $\{ \varphi_j \}_{j=0}^{\infty}$ are also the eigenfunctions of the 1-pdm of $G$. Since $G$ is a non-interacting Gibbs state, we know from \cite[Theorem~(ii)]{Suto} that $\langle n_i n_j \rangle_{G} \leq \langle n_i \rangle \langle  n_j \rangle_{G}$ holds as long as $i \neq j$. We use this inequality in order to bound
\begin{equation}
	2 \left| \frac{1}{2} \sum_{\sigma \in S_2} \varphi_{\sigma(i)}(x) \varphi_{\sigma(j)}(y) \right|^2 \left\langle n_i n_j \right\rangle_{G} \leq  | \varphi_i(x) |^2 | \varphi_j(y) |^2 \langle n_i \rangle_{G} \langle  n_j \rangle_{G} \label{eq:upperboundcan2b}  +  \overline{\varphi_i(x)} \overline{ \varphi_{j}(y) } \varphi_{j}(x) \varphi_{i}(y) \langle n_i \rangle_G \langle  n_j \rangle_{G}. 
\end{equation}
In combination with the estimate $n_j(n_j-1) \leq n_j^2$ for $j \geq 0$, this implies the following bound for the 2-particle density of $G$:
\begin{align}
	\varrho_G^{(2)}(x,y) \leq& \sum_{j \geq 1} | \varphi_j(x) |^2 | \varphi_j(y) |^2 \langle n_j^2 \rangle_{G} + | \varphi_0(x) |^2 | \varphi_0(y) |^2 \left( \langle n_0^2 \rangle_{G} - \langle n_0 \rangle_{G}^2 \right) \label{eq:upperboundcan3} \\
	&+\sum_{i,j \geq 0} | \varphi_i(x) |^2 | \varphi_j(y) |^2 \langle n_i \rangle_{G} \langle  n_j \rangle_{G} + \sum_{\stackrel{i,j \geq 0}{i \neq j}} \overline{\varphi_i(x)} \overline{ \varphi_{j}(y) } \varphi_{j}(x) \varphi_{i}(y) \langle n_i \rangle_G \langle  n_j \rangle_{G}. \nonumber 
\end{align}
Note that we added the positive term $\sum_{j \geq 1} | \varphi_j(x) |^2 | \varphi_j(y) |^2 \langle n_j \rangle_{G}^2 $ on the right-hand side to obtain the first term in the second line. The state $G$ has exactly $N$ particles, which implies that
\begin{equation}
	\left\langle n_0^2 \right\rangle_{G} - \left\langle n_0 \right\rangle_{G}^2 = \left\langle \left( \sum_{j\geq 1} n_j \right)^2 \right\rangle_{G} - \left\langle \sum_{j\geq 1} n_j \right\rangle_{G}^2 \,.
	\label{eq:upperboundcan3b}
\end{equation} 
Using \cite[Theorem~(ii)]{Suto} once more, we therefore have
\begin{equation}
	\langle n_0^2 \rangle_{G} - \langle n_0 \rangle_{G}^2 = \sum_{i,j \geq 1} \left( \langle n_i n_j \rangle_G - \langle n_i \rangle_G \langle n_j \rangle_G \right) \leq \sum_{j \geq 1} \langle n_j^2 \rangle_G.
	\label{eq:upperboundcan4}
\end{equation}
Additionally,  by including terms $i=j \geq 1$ in the sum, the last term on the right-hand side of \eqref{eq:upperboundcan3} is bounded from above by
\begin{equation}
	\sum_{\stackrel{i,j \geq 0}{i \neq j}} \overline{\varphi_i(x)} \overline{ \varphi_{j}(y) } \varphi_{j}(x) \varphi_{i}(y) \langle n_i \rangle_G \langle  n_j \rangle_{G} \leq 2 \text{Re} \gamma_>(x,y) N_0 \overline{\varphi_0(x)} \varphi_0(y) + | \gamma_> (x,y) |^2, \label{eq:upperboundcan5}
\end{equation}
where $\gamma_>(x,y) = \sum_{j \geq 1} \langle n_j \rangle_G \varphi_j(x) \overline{\varphi_j(y)}$ and $N_0 = \langle n_0 \rangle_G$. In combination, \eqref{eq:upperboundcan1}, \eqref{eq:upperboundcan3}, \eqref{eq:upperboundcan3b} and \eqref{eq:upperboundcan5} allow us to bound the interaction energy of $G$ (times $N$) from above by
\begin{align}
	&\frac{1}{2} \int_{\mathbb{R}^6} v(x - y) \varrho_G^{(2)}(x,y) \de(x,y) \leq  D \left( \varrho_G, \varrho_G\right) + \frac{1}{2} \int_{\mathbb{R}^6} v(x-y) | \gamma_> (x,y) |^2 \de(x,y) \label{eq:upperboundcan6} \\
	&\hspace{1cm}+ \text{Re} \int_{\mathbb{R}^6} v(x-y) \gamma_>(x,y) N_0 \overline{\varphi_0(x)} \varphi_0(y) \de(x,y) + \frac{1}{2} \int_{\mathbb{R}^6} v(x-y)  | \varphi_0(x) |^2 | \varphi_0(y) |^2 \de (x,y ) \sum_{j \geq 1} \langle n_j^2 \rangle_G \nonumber \\
	&\hspace{1cm}+ \frac{1}{2} \sum_{j \geq 1} \langle n_j^2 \rangle_{G} \int_{\mathbb{R}^6} v(x-y)  | \varphi_j(x) |^2 | \varphi_j(y) |^2 \de(x,y). \nonumber
\end{align}
It remains to estimate all terms on the right-hand side of this inequality except for the first.

We start with the term in the last line of \eqref{eq:upperboundcan6} and estimate
\begin{equation}
 \int_{\mathbb{R}^6} v(x-y)  | \varphi_j(x) |^2 | \varphi_j(y) |^2 \de(x,y) \leq  \Vert v \Vert_{L^{\infty}(\mathbb{R}^3)}.
	\label{eq:upperboundcan7}
\end{equation}
To estimate the remaining sum over $j$ in this term, and in the fourth term on the right-hand side of \eqref{eq:upperboundcan6}, we need the following Lemma.
\begin{lemma}
	\label{lem:momentbound}
	We have
	\begin{equation}
		\sum_{j \geq 1} \langle n_j^2 \rangle_{G} \lesssim \left( \frac{1}{ \beta \hbar \omega } \right)^{3}.
		\label{eq:lemnsquaredsums0}
	\end{equation}
	\end{lemma}
	
	\begin{proof}
		Denote by $\langle \cdot \rangle_{\textrm{gc}}$ the expectation with respect to the Gibbs state 
		\begin{equation}
			\frac{e^{-\beta \de \Upsilon\left(h + N^{-1} v \ast \varrho_G - \mu\right)}}{\tr[e^{-\beta \de \Upsilon\left(h + N^{-1} v \ast \varrho_{G} - \mu \right)}]},
			\label{eq:grandcanonicalGibbsstate}
		\end{equation}
		with $\mu$ chosen such the expected number of particles in the system equals $N$. From \cite[Remark~A.1]{me} we know that $\langle n_j \rangle_G \lesssim \langle n_j \rangle_{\textrm{gc}}$ for all $j \geq 0$. Denote by $e_j$ the eigenvalues of the operator $h + N^{-1} v \ast \varrho_{G}(x)$. A short computation shows that
		\begin{equation}
		\langle n_j^2 \rangle_{\mathrm{gc}} = \frac{ 1+ \exp\left( \beta( e_j - \mu ) \right)}{ \left( \exp\left( \beta (e_j - \mu) \right) -1 \right)^2}.
		\label{eq:lemnsquaredsums1}
		\end{equation}
		The sum over $j\geq 1$ of $\langle n_j^2 \rangle_{\mathrm{gc}}$ is thus of the form considered in Lemma~\ref{lem:semiclassicaltraceestimate},  
with $\zeta(t) = (1+ e^t)/( e^t-1 )^2$.
An application of the Lemma thus leads to the desired result.
	\end{proof}

The bound in \eqref{eq:upperboundcan7} and an application of Lemma~\ref{lem:momentbound} show that the last two terms on the right-hand side of \eqref{eq:upperboundcan6} are bounded from above as
\begin{equation}
	 \sum_{j \geq 1} \langle n_j^2 \rangle_G \left( \int_{\mathbb{R}^6} v(x-y)  | \varphi_0(x) |^2 | \varphi_0(y) |^2 \de (x,y ) + \int_{\mathbb{R}^6}  v(x-y)  | \varphi_j(x) |^2 | \varphi_j(y) |^2 \de(x,y) \right) \label{eq:upperboundcan8} \lesssim  \Vert v \Vert_{L^{\infty}(\mathbb{R}^3)} (\hbar\beta  \omega)^{-3}.  
\end{equation}
It remains to give a bound for the second and the third term on the right-hand side of \eqref{eq:upperboundcan6}, that is, for the two exchange terms.
We start with the second term on the right-hand side of \eqref{eq:upperboundcan6} and use \cite[Remark~A.1]{me} again to bound it as
\begin{align}
	&\int_{\mathbb{R}^6} v(x-y) | \gamma_> (x,y) |^2 \de(x,y) = \sum_{i,j \geq 1} \langle n_i \rangle_G  \langle n_j \rangle_G \int_{\mathbb{R}^6} v(x-y) \overline{ \varphi_i(x) } \varphi_i(y) \varphi_j(x) \overline{\varphi_j(y)} \de(x,y) \label{eq:upperboundcan9} \\
	&\hspace{2cm} \lesssim \sum_{i,j \geq 1} \langle n_i \rangle_{\textrm{gc}}  \langle n_j \rangle_{\textrm{gc}} \int_{\mathbb{R}^6} v(x-y) \overline{ \varphi_i(x) } \varphi_i(y) \varphi_j(x) \overline{\varphi_j(y)} \de(x,y) = \int_{\mathbb{R}^6} v(x-y) \left| \gamma_{>}^{\mathrm{gc}}(x,y) \right|^2 \de(x,y). \nonumber
\end{align}
Here $\gamma_>^{\mathrm{gc}}(x,y) = \sum_{j \geq 1}  \langle n_j \rangle_{\textrm{gc}} \varphi_j(x) \overline{\varphi_j(y)} $ is the part of the 1-pdm of the grand-canonical Gibbs state \eqref{eq:grandcanonicalGibbsstate} not including its largest eigenvalue, and we used the positivity of $\hat v$ which implies that the coefficients multiplying $\langle n_i \rangle_G  \langle n_j \rangle_G$ in \eqref{eq:upperboundcan9} are non-negative for any $i$ and $j$. 
 The right-hand side of \eqref{eq:upperboundcan9}  can be bounded by
\begin{equation}
\int_{\mathbb{R}^6} v(x-y) \left| \gamma_{>}^{\mathrm{gc}}(x,y) \right|^2 \de(x,y) \leq {\Vert v \Vert_{L^{\infty}(\mathbb{R}^3)}}\trs\left[ \left( \gamma_{>}^{\mathrm{gc}} \right)^2 \right] \leq {\Vert v \Vert_{L^{\infty}(\mathbb{R}^3)}} \left\Vert \gamma_{>}^{\mathrm{gc}} \right\Vert \trs\left[ \gamma_{>}^{\mathrm{gc}} \right] \lesssim \frac{N\Vert v \Vert_{L^{\infty}(\mathbb{R}^6)}}{\beta \hbar \omega}.
\label{eq:55}
\end{equation}
To estimate the operator norm of $\gamma_>^{\mathrm{gc}}$ by a constant times $(\beta \hbar \omega)^{-1}$ in the last step, we used Lemma~\ref{lem:spactralgapSCMF}. 

A similar estimate holds for the third term on the right-hand side of \eqref{eq:upperboundcan6}. We have
\begin{align}
&\text{Re} \int_{\mathbb{R}^6} v(x-y) \gamma_>(x,y) N_0 \overline{\varphi_0(x)} \varphi_0(y) \de(x,y) \lesssim \text{Re} \int_{\mathbb{R}^6} v(x-y) \gamma_>^{\mathrm{gc}}(x,y) N_0 \overline{\varphi_0(x)} \varphi_0(y) \de(x,y). 
\label{eq:upperboundcan10}
\end{align}
Let us denote the operator with integral kernel $ \varphi_0(x) v(x-y) \varphi_0(y)$ by $K$. Then
\begin{align}
 N_0 \text{Re} \int_{\mathbb{R}^6} v(x-y) \overline{\varphi_{0}(x)} \varphi_0(y) \gamma_>^{\mathrm{gc}}(x,y) \de(x,y) &=  N_0 \text{Re} \trs\left[K \gamma_>^{\mathrm{gc}} \right] \leq N\left\Vert \gamma_>^{\mathrm{gc}} \right\Vert \ \Vert K \Vert_1  \label{eq:57} \\
&= N \left\Vert \gamma_>^{\mathrm{gc}} \right\Vert \trs[K] = N v(0)  \left\Vert \gamma_>^{\mathrm{gc}} \right\Vert \lesssim \frac{ N \Vert v \Vert_{L^{\infty}(\mathbb{R}^3)} }{\beta \hbar \omega}. \nonumber
\end{align}
Here we have used that  $K$ is a positive operator, and applied  again Lemma~\ref{lem:spactralgapSCMF} in the last step.

In combination, \eqref{eq:upperboundcan6}, \eqref{eq:upperboundcan7} and \eqref{eq:upperboundcan8}--\eqref{eq:57} imply the bound
\begin{equation}
N^{-1}	\tr\left[ \sum_{1 \leq i < j \leq N} v(x_i-x_j) \ G \right] \leq N^{-1} D\left( \varrho_G, \varrho_G \right) + \frac{ \text{const. } }{\beta \hbar}. \label{eq:upperboundcan12}
\end{equation}
The constant in the above inequality is uniform in $\beta \omega \gtrsim 1$. It remains to consider the kinetic energy and the entropy of the state $G$.
\subsubsection*{Kinetic energy and entropy of the state $G$ and the final upper bound}
\label{sec:kineticenergyentropyc}
Since $G$ is a Gibbs state, we have
\begin{equation}
	\tr\left[ \sum_{i=1}^N h_i G \right] - \frac{1}{\beta} S(G) = - \frac{1}{\beta} \ln \left( \tr \exp\left( - \beta  \sum_{i=1}^N \left( h_i + N^{-1} v \ast \varrho_G(x_i) \right) \right) \right) - 2 N^{-1} D\left( \varrho_G, \varrho_G \right).
	\label{eq:finalestimate1}
\end{equation}
In combination with \eqref{eq:upperboundcan12}, this implies
\begin{equation}
	\tr\left[ H_N G \right] - \frac{1}{\beta} S(G) \leq - \frac{1}{\beta} \ln \left( \tr \exp\left( - \beta  \sum_{i=1}^N \left( h_i + N^{-1} v \ast \varrho_G(x_i) \right) \right) \right) - N^{-1} D\left( \varrho_G, \varrho_G \right) + \frac{ \text{const. } }{\beta \hbar}. \label{eq:finalestimate2}
\end{equation}
From \eqref{eq:Hartreeenergyc2} we see that the first two terms on the right-hand side of \eqref{eq:finalestimate2} equal $F^{\mathrm{H},\mathrm{c}}(\beta,N,\omega)$, and Lemma~\ref{lem:freeenergyboundcgc} shows that this is bounded from above by  $F^{\mathrm{H}}(\beta,N,\omega)$ plus a remainder of the order $\beta^{-1} \ln N$. In combination with \eqref{eq:finalestimate2}, we find the final upper bound
\begin{equation}
	F^{\mathrm{c}}(\beta,N,\omega) \leq F^{\mathrm{H}}(\beta,N,\omega) + \frac{ \text{const. } }{\beta \hbar}.
	\label{eq:69}
\end{equation}
The constant in the above inequality is uniform in $N \geq 1$ and $\beta \omega \gtrsim 1$. Together with \eqref{eq:lowerboundgc5}, this proves \eqref{eq:mainresults21}.
\section{Bounds on the 1-pdm and the Husimi function of approximate Gibbs states}
\label{sec:boundson1pdm}
In the first part of this section we prove \eqref{eq:mainresults21c}, which will conclude the proof of Theorem~\ref{thm:freeenergyscaling2}. Our approach is based on a lower bound for the bosonic relative entropy that quantifies its coercivity and was  proved in \cite[Lemma~4.1]{me2}. Afterwards we show how this result can be combined with Theorem~\ref{thm:limitHartreetheory} to prove  Corollaries~\ref{cor:weakBEC} and~\ref{cor:strongBEC}.
\subsection{Bound on the 1-pdm}
\label{sec:boundgrandcanonicaldensitymatrix}
Let $\Gamma \in \mathcal{S}_N^{\mathrm{gc}}$ with 1-pdm $\gamma$ be an approximate minimizer of the Gibbs free energy functional $\mathcal{F}$ in the sense that \eqref{eq:mainresults21b} holds. We allow for the larger set $\mathcal{S}_N^{\mathrm{gc}}$ instead of $\mathcal{S}_N^{\mathrm{c}}$ in view of Remark~\ref{rem:gc} in Section~\ref{sec:mainresult}; the proof is literally the same in both cases. From \cite[2.5.14.5]{Thirring_4} we know that the entropy satisfies $S(\Gamma) \leq s(\gamma)$, where $s$ denotes the bosonic entropy in \eqref{eq:bosonicentropy}. Let us have a closer look at the first line on the right-hand side of \eqref{eq:lowerboundgc3}. We obtain a lower bound if we replace $S(\Gamma)$ by $s(\gamma)$. We apply the Gibbs variational principle, optimize over the pair $(\eta,\mu) \in \mathcal{D}_N^{\mathrm{gc}}$ with $\mathcal{D}_N^{\mathrm{gc}}$ defined in \eqref{eq:sec4secondlemmav21} and use Lemma~\ref{lem:secondvariationalcharacterization} to conclude that
\begin{align}
	\mathcal{F}(\Gamma) &\geq F^{\mathrm{H}}(\beta,N,\omega) + \frac{1}{\beta} \mathcal{S}\left( \gamma, \gamma^{\mathrm{H}} \right) - \frac{v(0)}{2} 
	\label{eq:asymptoticsgc2} 
\end{align}
holds. Here $\gamma^{\mathrm{H}}$ is the unique minimizer of $\mathcal{F}^{\mathrm{H}}$ in \eqref{eq:Hartrees1pfreeenergyfkt} and $\mathcal{S}(\gamma,\gamma^{\mathrm{H}})$ denotes the bosonic relative entropy of $\gamma$ w.r.t. $\gamma^{\mathrm{H}}$, which is defined by
\begin{equation}
	\mathcal{S}\left( \gamma, \gamma^{\mathrm{H}} \right) = \trs\left[ f(\gamma) - f\left(\gamma^{\mathrm{H}} \right) - f'\left(\gamma^{\mathrm{H}} \right) \left(\gamma - \gamma^{\mathrm{H}} \right) \right]
	\label{eq:asymptoticsgc3}
\end{equation}
for $f$  given in \eqref{eq:bosonicentropy}. In combination, \eqref{eq:mainresults21b} and \eqref{eq:asymptoticsgc2} imply the bound
\begin{equation}
	\mathcal{S}\left(\gamma, \gamma^{\mathrm{H}} \right) \leq \beta\omega \delta +  \frac{\beta v(0)}2  \,.
	\label{eq:asymptoticsgc4}
\end{equation}
In order to quantify the coercivity of the bosonic relative entropy we use the following Lemma, which was  proved in \cite[Lemma~4.1]{me2}. For the sake of completeness we state it here.
\begin{lemma}
	\label{lem:relativeentropy}
	There exists a constant $C>0$ such that for any two nonnegative trace-class operators $a$ and $b$ we have
	\begin{equation}
	\mathcal{S}\left( a, b \right) \geq C \frac{\Vert a - b \Vert_1^2}{\left\Vert 1+b \right\Vert \trs\left[ a + b \right]}.
	\label{eq:asymptoticslemma2}
	\end{equation}
\end{lemma}
The above bound is only helpful in case the operator $b$ does not have a condensate, i.e., its norm is not too large. Accordingly, we first need to get rid of the contribution from the condensate of $\gamma^{\mathrm{H}}$ in $\mathcal{S}(\gamma,\gamma^{\mathrm{H}})$. Denote by $P$ the projection onto the (unique) ground state of $h + N^{-1} v \ast \varrho^{\mathrm{H}}(x)$ and let $Q = 1-P$. We write $\gamma = \sum_{\alpha = 0}^{\infty} \gamma_{\alpha} | \psi_{\alpha} \rangle \langle \psi_{\alpha} |$, $\gamma^{\mathrm{H}} = \sum_{\alpha = 0}^{\infty} \gamma_{\alpha}^{\mathrm{H}} | \varphi_{\alpha} \rangle \langle \varphi_{\alpha} |$ and note that
\begin{equation}
	\trs\left[P \left( f(\gamma) - f\left( \gamma^{\mathrm{H}} \right) - f'\left(\gamma^{\mathrm{H}}\right)\left(\gamma - \gamma^{\mathrm{H}}\right) \right) P \right] = \sum_{\alpha = 0}^{\infty} \left| \langle \varphi_0, \psi_{\alpha} \rangle \right|^2 \left( f(\gamma_{\alpha}) - f\left( \gamma_{0}^{\mathrm{H}} \right) - f'\left( \gamma_{0}^{\mathrm{H}} \right)\left( \gamma_{\alpha} - \gamma_{0}^{\mathrm{H}} \right) \right) \geq 0
\end{equation}
since $f$ is convex. This implies that
\begin{equation}
	\mathcal{S}\left(\gamma, \gamma^{\mathrm{H}} \right) \geq \trs\left[ Q \left( f(Q \gamma Q) - f\left( \gamma^{\mathrm{H}} \right) - f'\left(\gamma^{\mathrm{H}}\right)\left(\gamma - \gamma^{\mathrm{H}}\right) \right) Q \right] + \trs\left[ Q f(\gamma) Q - Q f(Q \gamma Q) Q \right]. 
	\label{eq:asymptoticsgc5}
\end{equation}
From Lemma~\ref{lem:convexity} we know that the last term on the right-hand side is nonnegative and can be dropped for a lower bound. In combination with Lemma~\ref{lem:relativeentropy} this yields  
\begin{equation}
	\mathcal{S}\left(\gamma, \gamma^{\mathrm{H}} \right) \geq \mathcal{S}\left(Q \gamma Q, Q \gamma^{\mathrm{H}} Q\right) \geq C \frac{\left\Vert Q \left( \gamma - \gamma^{\mathrm{H}} \right) Q \right\Vert_1^2}{\left\Vert 1+ Q \gamma^{\mathrm{H}} Q \right\Vert \trs\left[ Q \left( \gamma + \gamma^{\mathrm{H}} \right) Q \right]}.
	\label{eq:asymptoticsgc6}
\end{equation}
Lemma~\ref{lem:spactralgapSCMF} implies that  $\left\Vert Q \gamma^{\mathrm{H}} Q \right\Vert \lesssim (\beta \hbar \omega)^{-1}$. In combination with \eqref{eq:asymptoticsgc4} and \eqref{eq:asymptoticsgc6} we thus obtain
\begin{equation}
	\left\Vert Q \left( \gamma - \gamma^{\mathrm{H}} \right) Q \right\Vert_1 \lesssim (1+\beta \hbar \omega)^{1/2} N^{2/3} (1+ \delta)^{1/2}. 
	\label{eq:asymptoticsgc7}
\end{equation}
It remains to estimate $\Vert P ( \gamma - \gamma^{\mathrm{H}} ) P \Vert_1$ as well as the trace norm of the off-diagonal of $\gamma$ w.r.t. $P$ and $Q$.

We first consider $\Vert P ( \gamma - \gamma^{\mathrm{H}} ) P \Vert_1$ and note that
\begin{equation}
	N = \trs \gamma = \trs \left[ P \gamma P \right] + \trs \left[ Q \gamma Q \right] = \langle \varphi_0, \gamma \varphi_0 \rangle + \trs \left[ Q \gamma^{\mathrm{H}} Q \right] + \trs \left[ Q(\gamma-\gamma^{\mathrm{H}}) Q \right] \,.
	\label{eq:asymptoticsgc8}
\end{equation}
Since $\trs \gamma^{\mathrm{H}} = N$ we conclude with \eqref{eq:asymptoticsgc7} that 
\begin{equation}
\left\Vert P \left( \gamma - \gamma^{\mathrm{H}} \right) P \right\Vert_1 =	\left| \langle \varphi_0, \gamma \varphi_0 \rangle  - \langle \varphi_0, \gamma^{\mathrm{H}} \varphi_0 \rangle \right| \lesssim (1+\beta \hbar \omega)^{1/2} N^{2/3} (1+ \delta)^{1/2}.
	\label{eq:asymptoticsgc9}
\end{equation}
The trace norm of the off-diagonal of $\gamma$ w.r.t. $P$ and $Q$ can be estimated as
\begin{equation}
	\left\Vert P \gamma Q \right\Vert_1 \leq \left\Vert \gamma Q \right\Vert \leq \Vert \gamma \Vert^{1/2} \Vert \gamma^{1/2} Q \Vert \leq N^{1/2} \Vert Q \gamma Q \Vert^{1/2} \leq N^{1/2} \left( \left\Vert Q \left(\gamma - \gamma^{\mathrm{H}} \right) Q \right\Vert_1 + \left\Vert \gamma^{\mathrm{H}} \right\Vert \right)^{1/2}. 
	\label{eq:asymptoticsgc12}
\end{equation}
Using again that $\Vert \gamma^{\mathrm{H}} \Vert \lesssim (\beta \hbar \omega)^{-1}$, this implies
\begin{equation}
	\left\Vert P \gamma Q \right\Vert_1 \lesssim (1+\beta \hbar \omega)^{1/4} N^{5/6} (1+\delta)^{1/4}.
	\label{eq:asymptoticsgc13}
\end{equation}
We combine \eqref{eq:asymptoticsgc13} with \eqref{eq:asymptoticsgc7} and \eqref{eq:asymptoticsgc9} to conclude that
\begin{equation}
	\left\Vert \gamma - \gamma^{\mathrm{H}} \right\Vert_1 \lesssim (1+\beta\hbar \omega)^{1/4} N^{5/6} (1+\delta)^{1/4} + (1+\beta \hbar \omega)^{1/2} N^{2/3} (1+ \delta)^{1/2} \,.
	\label{eq:asymptoticsgc14}
\end{equation}
 This bound proves \eqref{eq:mainresults21c} as long as $ \beta \hbar \omega \lesssim 1$. To obtain the claimed uniformity in the whole region $\beta \hbar \omega \gtrsim N^{-1/3}$, we combine it with a second bound that we apply in the region $\beta \hbar \omega \gtrsim 1$.

This second bound can be obtained by perturbing around the zero temperature ($\beta = + \infty$) case. We use $S(\Gamma) \leq s(\gamma)$ 
and \eqref{eq:lowerboundgc3} with the choices $\eta = \varrho^{\mathrm{H}}$ and $\mu = \mu^{\mathrm{H}}$ and obtain
\begin{align}
	\mathcal{F}(\Gamma) & \geq \tr\left[ \left( h + N^{-1} v \ast \varrho^{\mathrm{H}} - \mu^{\mathrm{H}} \right) \gamma \right] + \mu^{\mathrm{H}} N - \frac{1}{\beta} s(\gamma) - N^{-1} D\left(\varrho^{\mathrm{H}}, \varrho^{\mathrm{H}} \right) - \frac{ v(0)}{2} \label{eq:70} \\
	 & \geq \frac{1}{2} \tr\left[ \left( h + N^{-1} v \ast \varrho^{\mathrm{H}} - \mu^{\mathrm{H}} \right) \gamma \right] + \mu^{\mathrm{H}} N - N^{-1} D\left(\varrho^{\mathrm{H}}, \varrho^{\mathrm{H}} \right) - \frac{ v(0)}{2} \nonumber \\
	& \quad + \frac{1}{\beta} \trs\left[ \ln\left( 1 - \exp\left( -\frac{\beta}{2} \left( h + N^{-1} v \ast \varrho^{\mathrm{H}} - \mu^{\mathrm{H}} \right) \right) \right) \right]. \nonumber
\end{align}
Let $\{ e_j^{\mathrm{H}} \}_{j=0}^{\infty}$ and $\{ e_j \}_{j=0}^{\infty}$ be the eigenvalues of the operators $h + N^{-1} v \ast \varrho^{\mathrm{H}}(x)$ and $h$, respectively. 
Similarly to \eqref{eq:64} we choose $M \in \mathbb{N}$ independently of $\hbar$ and $\varrho^{\mathrm{H}}$, such that $e_j^{\mathrm{H}} - \mu^{\mathrm{H}} \geq e_j/2$ holds for all $j > M$. The last term on the right-hand side of \eqref{eq:70} is for $\beta \hbar \omega \gtrsim 1$ bounded from below by
\begin{equation}
	\frac{1}{\beta} \sum_{j=0}^M \ln\left( 1 - \exp\left( -\frac{\beta}{2} \left( e_j^{\mathrm{H}} - \mu^{\mathrm{H}} \right) \right) \right) + \frac{1}{\beta} \sum_{j>M}^{\infty} \ln\left( 1 - \exp\left( -\frac{\beta}{4} e_j \right) \right) \label{eq:148}  \gtrsim \frac{-M \ln N}{\beta} -  \frac{1}{\beta} \,. 
\end{equation}
To obtain the bound for the first sum we used $ (\exp(\beta(e^{\mathrm{H}}_j - \mu^{\mathrm{H}}))-1)^{-1} \leq N$ for all $j \geq 0$, and to bound the second sum we used the explicit form of the eigenvalues of $h$ as well as $\beta \hbar \omega \gtrsim 1$. 

In combination with Lemma~\ref{lem:minimizersHartree} and the fact that $\Gamma$ is an approximate minimizer of the Gibbs free energy functional in the sense of \eqref{eq:mainresults21b}, this implies that
\begin{equation}
	\frac{1}{2} \tr\left[ \left( h + N^{-1} v \ast \varrho^{\mathrm{H}} - \mu^{\mathrm{H}} \right) \gamma \right] \leq \frac{1}{\beta} \trs\left[ \ln\left( 1 - \exp\left( -\beta \left( h + N^{-1} v \ast \varrho^{\mathrm{H}} - \mu^{\mathrm{H}} \right) \right) \right) \right] \label{eq:71} + \omega \delta +  \text{const. } \frac{\ln N}{\beta} \,.
\end{equation}
The first term on the right-hand side is negative and can be dropped for an upper bound. When we apply the spectral gap estimate in Lemma~\ref{lem:spactralgapSCMF}, we obtain a lower bound for the left-hand side of \eqref{eq:71} and find
\begin{equation}
	\trs\left[ Q \gamma Q \right] \lesssim N^{1/3} \delta + \frac{\ln N}{\beta \hbar \omega} \,.
	\label{eq:72}
\end{equation}
As above, $Q$ denotes the spectral projection onto the orthogonal complement of the ground state subspace of the Hartree operator $h+ N^{-1} v \ast \varrho^{\mathrm{H}}(x)$. 

To conclude the desired bound for the trace norm difference of $\gamma$ and $\gamma^{\mathrm{H}}$, we proceed as in \eqref{eq:asymptoticsgc8}--\eqref{eq:asymptoticsgc12} to estimate
\begin{equation}
	\left\Vert \gamma - \gamma^{\mathrm{H}} \right\Vert_1 \leq 2 \left(   \trs  \left[ Q \gamma Q \right]  + \trs \left[ Q\gamma^{\mathrm{H}}  \right] +  N^{1/2} \Vert Q \gamma Q \Vert^{1/2} \right) \,.
	\label{eq:73}
\end{equation}
Lemma~\ref{lem:semiclassicaltraceestimate} implies that $\trs [ Q\gamma^{\mathrm{H}}  ] \lesssim (\beta\hbar\omega)^{-3}$. In combination with   
\eqref{eq:72} and the fact that the trace norm dominates the operator norm, we thus obtain
\begin{equation}
	\left\Vert \gamma - \gamma^{\mathrm{H}} \right\Vert_1 \lesssim N^{2/3} (1+ \delta)^{1/2} 
	\label{eq:74}
\end{equation}
for $\beta \hbar \omega \gtrsim 1$. 
In combination with \eqref{eq:asymptoticsgc14}, \eqref{eq:74} proves \eqref{eq:mainresults21c}. This concludes the proof of Theorem~\ref{thm:freeenergyscaling2}.
\subsection{Proof of Corollary~\ref{cor:weakBEC}}
\label{sec:weakBEC}
To prove Corollary~\ref{cor:weakBEC}, we directly relate 
the Husimi function of an approximate minimizer of the Gibbs free energy functional in the sense of \eqref{eq:mainresults21b} 
to the corresponding quantity of the semiclassical free energy functional. 
We combine \eqref{eq:lowerboundgc3} with the choices $\eta = N \varrho^{\mathrm{sc}}$, $\mu = v \ast \varrho^{\mathrm{sc}}(0) + \mu^{\mathrm{sc}}$ and \eqref{eq:propsemiclassical12} to see that
\begin{equation}
	\hbar^3 \mathcal{F}(\Gamma) \geq F^{\mathrm{sc}}(\beta,\omega) + \frac{1}{\beta} \mathcal{S}^{\mathrm{sc}}\left( m_{\gamma}, \gamma^{\mathrm{sc}} \right) - O\left( \hbar \omega \right)  \,.
	\label{eq:84}
\end{equation}
In combination with  \eqref{eq:mainresults21b} for $\delta = o(N)$ 
this implies for $R>0$
\begin{equation}
	 \mathcal{S}^{\mathrm{sc}} \left( \mathds{1}_{A_R} m_{\gamma}, \mathds{1}_{A_R} \gamma^{\mathrm{sc}} \right) \leq \mathcal{S}^{\mathrm{sc}}\left( m_{\gamma}, \gamma^{\mathrm{sc}} \right) \leq o(1),
	 \label{eq:85}
\end{equation}
where $\mathds{1}_{A_R}$ denotes multiplication with the characteristic function $\mathds{1}_{A_R}(x)$ of the set $A_R = (B_{ \omega^{1/2} R } \times B_{\omega^{-1/2} R})^{ \mathrm{c}
}  \subset \mathbb{R}^6$ and $B_{r} \subset \mathbb{R}^3$ denotes the ball with radius $r>0$ centered at the origin. An application of Lemma~\ref{lem:semiclassicalcoercivity} thus implies
\begin{equation}
	\int_{ A_R } \left| m_{\gamma}(p,x) - \gamma^{\mathrm{sc}}(p,x) \right| \de(p,x) \leq o(1) \left( \int_{A_R} \left( m_{\gamma}(p,x) + \gamma^{\mathrm{sc}}(p,x) \right) \left( 1 + \gamma^{\mathrm{sc}}(p,x) \right) \de(p,x) \right)^{1/2}.
	\label{eq:86}
\end{equation}
We use \eqref{eq:54} to estimate the $L^1(\mathbb{R}^6)$- and the $L^2(\mathbb{R}^6)$-norm of $\gamma^{\mathrm{sc}}$, as well as $\int_{A_R} m_{\gamma}(p,x) \leq (2\pi)^3$. Moreover, with the Euler--Lagrange equation \eqref{eq:semiclassicalfunctional5} and \eqref{eq:34} we check that
\begin{equation}
	\int_{A_R} m_{\gamma}(p,x) \gamma^{\mathrm{sc}}(p,x) \de(p,x) \leq \left\Vert \mathds{1}_{A_R} \gamma^{\mathrm{sc}} \right\Vert_{L^{\infty}(\mathbb{R}^6)} (2\pi)^3 \lesssim \frac{1}{\beta \omega R^2 }
	\label{eq:87}
\end{equation}
holds. In combination with \eqref{eq:86} this yields \eqref{eq:83}. 
Eq.~\eqref{eq:83b} is a direct consequence of \eqref{eq:83}, \eqref{eq:89} and the normalization condition \eqref{eq:normalizationcondition}. This proves Corollary~\ref{cor:weakBEC}.

\subsection{Proof of Corollary~\ref{cor:strongBEC}}
\label{sec:BECSCMF}
To prove Corollary~\ref{cor:strongBEC}, we relate the condensate fraction and the Husimi function of an approximate minimizer of the Gibbs free energy functional to the corresponding quantities of the Hartree free energy functional in \eqref{eq:Hartrees1pfreeenergyfkt}. In combination with Theorem~\ref{thm:limitHartreetheory}, this will imply the claim.  

Let $\gamma$ be the 1-pdm of an approximate minimizer of the Gibbs free energy functional in the sense of \eqref{eq:mainresults21b}. By $P$ we denote the projection onto the eigenspace of its largest eigenvalue and $Q = 1-P$. Let also $P^{\mathrm{H}}$ be the projection onto the eigenspace of the largest eigenvalue of $\gamma^{\mathrm{H}}$ and define $Q^{\mathrm{H}} = 1 - P^\mathrm{H}$. Theorem~\ref{thm:freeenergyscaling2} and the fact that the trace norm dominates the operator norm imply that
\begin{equation}
N^{5/6}(1+\delta)^{1/4} \gtrsim \left\Vert \gamma - P^{\mathrm{H}} \gamma^{\mathrm{H}} + Q^{\mathrm{H}} \gamma^{\mathrm{H}} \right\Vert \geq \left\Vert \gamma - P^{\mathrm{H}} \gamma^{\mathrm{H}} \right\Vert - \left\Vert Q^{\mathrm{H}} \gamma^{\mathrm{H}} \right\Vert.
\label{eq:18}
\end{equation}
In \eqref{eq:propsemiclassicalminimizer8} we showed that the last term on the right-hand side of \eqref{eq:18} is bounded by a constant times $(\beta \hbar \omega)^{-1} \lesssim N^{1/3}$. Using the min-max principle, see e.g. \cite[Theorem~12.1]{LiLo10}, and Theorem~\ref{thm:freeenergyscaling2} we see that 
\begin{equation}
	\left\Vert Q \gamma \right\Vert \leq \left\Vert Q^{\mathrm{H}} \gamma Q^{\mathrm{H}} \right\Vert \leq \left\Vert Q^{\mathrm{H}} \left( \gamma - \gamma^{\mathrm{H}} \right) Q^{\mathrm{H}} \right\Vert + \left\Vert Q^{\mathrm{H}} \gamma^{\mathrm{H}} \right\Vert \lesssim  N^{5/6}(1+\delta)^{1/4} \,.
	\label{eq:21}
\end{equation}
In combination, \eqref{eq:18} and \eqref{eq:21} yield
\begin{equation}
	\left\Vert P \gamma - P^{\mathrm{H}} \gamma^{\mathrm{H}} \right\Vert \lesssim N^{5/6}(1+\delta)^{1/4}.
	\label{eq:19}
\end{equation}
Since the operators in \eqref{eq:19} have rank one, the inequality is also true if the operator norm is replaced by the trace norm. In combination with Theorem~\ref{thm:limitHartreetheory} this proves \eqref{eq:90}.  

Next, we prove the claimed bound \eqref{eq:90a} for the Husimi function, and start by estimating
\begin{align}
	&\frac{1}{(2 \pi \hbar)^3} \int_{\mathbb{R}^6} \left| m_{Q \gamma}(p,q) - m_{Q^{\mathrm{H}} \gamma^{\mathrm{H}} }(p,q)  \right| \de(p,q) = \frac{1}{(2 \pi \hbar)^3} \int_{\mathbb{R}^6} \left| \left\langle \ell^{\hbar}_{p,q}, \left( Q \gamma - Q^{\mathrm{H}} \gamma^{\mathrm{H}} \right) \ell^{\hbar}_{p,q} \right\rangle \right| \de(p,q) \nonumber \\
	&\hspace{2cm} \leq \frac{1}{(2 \pi \hbar)^3} \int_{\mathbb{R}^6} \left\langle \ell^{\hbar}_{p,q}, \left| \left( Q \gamma - Q^{\mathrm{H}} \gamma^{\mathrm{H}} \right) \right| \ell^{\hbar}_{p,q} \right\rangle \de(p,q) = \left\Vert Q \gamma - Q^{\mathrm{H}} \gamma^{\mathrm{H}} \right\Vert_1\,. \label{eq:17}
\end{align}
From Theorem~\ref{thm:freeenergyscaling2} and \eqref{eq:19} with the operator norm replaced by the trace norm, we know that
\begin{equation}
	\left\Vert Q \gamma - Q^{\mathrm{H}} \gamma^{\mathrm{H}} \right\Vert_1 \lesssim N^{5/6}(1+\delta)^{1/4}.
	\label{eq:20}
\end{equation}
In combination with Theorem~\ref{thm:limitHartreetheory}, Eq.~\eqref{eq:90a} readily follows from the triangle inequality.
This completes the proof of Corollary~\ref{cor:strongBEC}.
\section{Proof of Theorem~\ref{thm:freeenergyscaling1}}
\label{sec:lastsection}
In this section we show how the analysis in Sections~\ref{sec:freeenergybounds} and~\ref{sec:boundson1pdm} needs to be adjusted in order to prove Theorem~\ref{thm:freeenergyscaling1}. One main difference between the proof of Theorem~\ref{thm:freeenergyscaling2} and the one of Theorem~\ref{thm:freeenergyscaling1} concerns the proof of the spectral gap estimate for the Hartree operator. In case of the semiclassical MF scaling such a bound follows from the assumed bound on the Hessian of the interaction potential in \eqref{eq:condintpot}, see Lemma~\ref{lem:spactralgapSCMF}. Here the interaction is of much shorter range and we do not have such a bound at our disposal, and therefore need to rely on other techniques. Our analysis below is based on the fact that the interaction can be seen to leading order only in the condensate. This, in particular, implies that the Hartree minimizer is to leading order given by the 1-pdm of a Gibbs state corresponding to a Hartree operator at $T=0$. Using the diluteness of the density of the thermal cloud of this Gibbs state, we reduce the task of estimating the spectral gap for the Hartree operator to estimating the one of a Hartree operator at $T=0$. In the latter case the existence of a uniform (in $N \geq 1$ and $\beta \omega \gtrsim N^{-1/3}$) spectral gap can easily be proved. The second main difference between the two parameter regimes is that the exchange terms are more difficult to estimate in the scaling limit we consider here. In particular, we need to carefully estimate the integral kernel of the 1-pdm of our trial states in order to be able to obtain bounds with the claimed accuracy. This has to be compared to the relatively simple estimates in \eqref{eq:55} and \eqref{eq:57}.
\subsection{Spectral gap estimate for the Hartree operator in the mean-field limit}
\label{sec:spectralgap}
In this section we consider the Hartree free energy functionals $\mathcal{F}^{\mathrm{H}}$ and $\mathcal{F}^{\mathrm{H},\mathrm{c}}$ in \eqref{eq:freeenergy2} and \eqref{eq:Hartreefunctional}, respectively, with $h$ defined in \eqref{eq:h2} and $v_N$ in \eqref{eq:Interactionpotential1}. We recall that $\gamma^{\mathrm{H},\mathrm{c}}$ denotes the 1-pdm of the unique minimizer $G^{\mathrm{H},\mathrm{c}}$ of $\mathcal{F}^{\mathrm{H},\mathrm{c}}$, whose existence is guaranteed by Lemma~\ref{prop:existenceminimizersHartreefunctional}. By $\varrho^{\mathrm{H},\mathrm{c}}$ we denote its density, that is, $\varrho^{\mathrm{H},\mathrm{c}}(x) = \gamma^{\mathrm{H},\mathrm{c}}(x,x)$.

The goal of this section is to prove the following statement:
\begin{proposition}
	\label{prop:spectralgap}
	Let $v_N$ satisfy the assumptions of Theorem~\ref{thm:freeenergyscaling1} and denote by $\Delta e^{\mathrm{H}}$ the spectral gap of the Hartree operator $h^{\mathrm{H}} = h + v_N \ast \varrho^{\mathrm{H},\mathrm{c}}(x)$ above its unique ground state. Then
	\begin{equation}
	\Delta e^{\mathrm{H}} \gtrsim \omega
	\label{eq:spectralgap3b}
	\end{equation}
	 holds uniformly in $N \geq 1$ and $\beta \omega \gtrsim N^{-1/3}$.
\end{proposition}
\begin{remark}
	The condition $0 \leq \kappa \leq 1/6$ is not needed in the proof of Proposition~\ref{prop:spectralgap}. We only require $\kappa \geq 0$.
\end{remark}
\begin{remark}
	The above Proposition remains valid if $\varrho^{\mathrm{H},\mathrm{c}}$ is replaced by $\varrho^{\mathrm{H}}$ in the definition of the Hartree operator $h^{\mathrm{H}}$. We need the bound in the form stated in Proposition~\ref{prop:spectralgap} because we choose $G^{\mathrm{H},\mathrm{c}}$ as a trial state to prove an upper bound for the  free energy in Section~\ref{sec:freeenergybounds2}.   
\end{remark}
The proof of Proposition~\ref{prop:spectralgap} is carried out in three steps. In the first step we show that the Hartree free energy can be approximated by the free energy of a non-interacting thermal cloud plus the minimum of the Hartree energy functional capturing the energy of the condensate. From this statement we conclude in the second step that $\gamma^{\mathrm{H},\mathrm{c}}$ can be approximated by the 1-pdm of a grand-canonical Gibbs state corresponding to a Hartree operator at $T=0$, to leading order in trace norm. In the third step, we use this statement to relate the spectral gap estimate for the Hartree operator in \eqref{eq:spectralgap3b} to the same question for a Hartree operator at $T=0$. The physical picture behind this reasoning is that because of its higher kinetic energy, the thermal cloud spreads over a much larger volume in the trap than the condensate. Accordingly, it is much more dilute and its interaction energy is much smaller than the one of the condensate.

Before we start with the proof, we introduce the Hartree and NLS energy functionals and some notation concerning the ideal Bose gas in the trap. Let $Q(h)$ denote the form domain of $h$. For functions $\phi \in Q(h)$ we define the two energy functionals by 
\begin{equation}
\mathcal{E}^{\mathrm{H}}(\phi) = \left\langle \phi, h \phi \right\rangle + \frac{\lambda}{2} \int_{\mathbb{R}^3} | \phi(x) |^2 v_N(x-y) | \phi(y) |^2 \de(x,y)
\label{eq:spectralgap2a}
\end{equation}
with $\lambda \geq 0$ and by
\begin{equation}
\mathcal{E}^{\mathrm{NLS}}(\phi) = \left\langle \phi, h \phi \right\rangle + \frac{\alpha}{2} \int_{\mathbb{R}^3} | \phi(x) |^4  \de x
\label{eq:spectralgapn1}
\end{equation}
with $\alpha \geq 0$, respectively. The constants $\lambda$ and $\alpha$ are introduced to be able to capture a condensate fraction different from one. The corresponding energies are given by
\begin{equation}
E^{\mathrm{H}}(\lambda) = \inf_{\Vert \phi \Vert = 1} \mathcal{E}^{\mathrm{H}}(\phi) \quad \text{and} \quad E^{\mathrm{NLS}}(\alpha) = \inf_{\Vert \phi \Vert = 1} \mathcal{E}^{\mathrm{NLS}}(\phi).
\label{eq:spectralgapn2}
\end{equation}
Both energy functionals \eqref{eq:spectralgap2a} and \eqref{eq:spectralgapn1} have a unique minimizer, which we denote by $\phi^{\mathrm{H}}_{\lambda}$ and $\phi^{\mathrm{NLS}}_{\alpha}$, respectively. In case of the Hartree energy functional this follows from the assumption $\hat{v} \geq 0$. Note that the Hartree energy $E^{\mathrm{H}}(\lambda)$ and $\phi^{\mathrm{H}}_{\lambda}$ depend, besides $\lambda$, also on $N$ through $v_N$, but we suppress this dependence in the notation for simplicity. 

By 
\begin{equation}
F_0(\beta,N,\omega) = \frac{1}{\beta} \trs \left[ \ln \left( 1 - e^{-\beta \left( h - \mu_0 \right) }\right) \right] + \mu_0 N
\label{eq:101}
\end{equation}
we denote the free energy of the ideal Bose gas in the grand-canonical ensemble. Here the chemical potential $\mu_0$ is chosen such that the expected number of particles in the systems equals $N$, that is, $\trs \gamma_{N,0} = N$ with
\begin{equation}
\gamma_{N,0} = \frac{1}{e^{\beta(h-\mu_0)} -1}.
\label{eq:gc1-pdm}
\end{equation}
Additionally, $N_0 = (e^{\beta \mu_0} -1)^{-1}$ is the expected number of particles in the condensate.

The first two steps in the proof of Proposition~\ref{prop:spectralgap} are captured in the following Lemma.
\begin{lemma}
	\label{lem:Hartreeenergy}
	Let $v_N$ satisfy the assumptions of Theorem~\ref{thm:freeenergyscaling1} but let $\kappa \geq 0$ be arbitrary. In the limit $N \to \infty$ with $\beta \omega \gtrsim N^{-1/3}$ we have 	
	\begin{equation}
	\left| F^{\mathrm{H}}(\beta,N,\omega) - F_0(\beta,N,\omega) - N_0 E^{\mathrm{H}}(N_0) \right| \lesssim \omega N^{2/3}.
	\label{eq:spectralgap5d}
	\end{equation}
	If $\gamma_N$ is a sequence of approximate minimizers of the Hartree free energy functional in the sense that
	\begin{equation}
	\left| \mathcal{F}^{\mathrm{H}}(\Gamma_N) - F_0(\beta,N,\omega) - N_0 E^{\mathrm{H}}(N_0) \right| \leq \omega \delta
	\label{eq:spectralgap5e}
	\end{equation}
	for some $\delta > 0$ then 
	\begin{equation}
	\left\Vert \gamma_N - \frac{1}{\exp\left( \beta \left( h + N_0 v_N \ast | \phi^{\mathrm{H}}_{N_0} |^2(x) - \mu_N \right) \right) - 1 } \right\Vert_1 \lesssim N^{11/12} + N^{3/4} \delta^{1/4}\,.
	\label{eq:spectralgap5f}
	\end{equation}
	Here $\mu_N = E^{\mathrm{H}}(N_0) + \frac{N_0}{2} \int_{\mathbb{R}^6} | \phi^{\mathrm{H}}_{N_0}(x) |^2 v_N(x-y) | \phi^{\mathrm{H}}_{N_0}(y) |^2 \de(x,y) + \mu_0$ with $\mu_0$ in \eqref{eq:gc1-pdm}. 
\end{lemma}

A  statement similar to Lemma~\ref{lem:Hartreeenergy} was  proved in \cite{me} for the full quantum mechanical free energy in the GP limit (i.e., $\kappa=1$). The proof of Lemma~\ref{lem:Hartreeenergy} uses similar ideas but in a much simpler setting. To not interrupt the main line of the argument, we therefore defer the proof of Lemma~\ref{lem:Hartreeenergy} to Section~\ref{sec:Appendix} and continue with the proof of Proposition~\ref{prop:spectralgap} here. 
\subsubsection*{Proof of Proposition~\ref{prop:spectralgap}}
We only consider the case  $\kappa > 0$ because it is more complicated since the interaction potential $v_N$ converges to a delta distribution as $N\to \infty$. We first show that the operator $h + \alpha | \phi_{\alpha} (x) |^2$ has a spectral gap above its unique ground state that is uniform in $\alpha \in [0, \alpha_0 ]$, where $\alpha_0 = \int_{\mathbb{R}^3} v(x) \de x$, and afterwards reduce the full problem to this case. To simplify the notation, we denoted the NLS minimizer by $\phi_{\alpha}$. For the first step, we use that the map $[0,\infty)\ni\alpha \mapsto \left| \phi_{\alpha} \right|^2$ is continuous in $L^2(\mathbb{R}^3)$, which follows from the uniqueness of minimizers in a standard way. The operator $h + \alpha | \phi_{\alpha} (x) |^2$ has discrete spectrum and a non-zero spectral gap above its ground state, and from the continuity in $\alpha$ we deduce that its eigenvalues are continuous functions of $\alpha$  (see, e.g., \cite[Theorem~VIII.23]{ReSi1980}, or \eqref{eq:spectralgap57}--\eqref{7:31} below).  Hence a uniform  non-zero lower bound  for the spectral gap for $\alpha \in [0,\alpha_0]$ follows from continuity. 

Before we reduce the statement of Proposition~\ref{prop:spectralgap} to the above statement we prove three technical Lemmas. The first shows that the operator $h+ N_0 v_N \ast |\phi^{\mathrm{H}}_{N_0} |^2(x) $ has a uniform spectral gap above its ground state.
\begin{lemma}
	\label{lem:eigenvaluebound}
	Let $\Delta e$ denote the spectral gap of the operator $h+N_0 v_N \ast |\phi^{\mathrm{H}}_{N_0} |^2(x)$ above its unique ground state. Then
	\begin{equation}
		\Delta e \gtrsim \omega
		\label{eq:155}
	\end{equation}
	holds uniformly in $N \geq 1$ and $\beta \omega \gtrsim N^{-1/3}$.
\end{lemma}
\begin{proof}
	Let $\{ e_j \}_{j=0}^{\infty}$ and $\{ \lambda_j \}_{j=0}^{\infty}$ be the increasingly ordered eigenvalues of the operators $h_1 = h+ N_0 v_N \ast |\phi^{\mathrm{H}}_{N_0} |^2(x) $ and $h_2 = h + \alpha | \phi_{\alpha} (x) |^2$ with $\alpha = N_0 N^{-1} \int v(x) \de x$. 
	In the following we will show that
	\begin{equation}
		\lim_{N \to \infty } \sup_{\beta \omega \gtrsim N^{-1/3}} \left| e_j - \lambda_j \right| = 0
		\label{eq:166}
	\end{equation}
	holds for $j=0,1$. In combination with the discussion in the paragraph preceding Lemma \ref{lem:eigenvaluebound} concerning the spectral gap of the operator $h_2$, this will prove the claim.
	
	We only give the proof for the difference of the lowest eigenvalues. The statement for the difference of the second eigenvalues follows from the same arguments and the min-max principle, see e.g. \cite[Theorem~12.1]{LiLo10}. (This way one easily checks that \eqref{eq:166} actually holds for all $j \geq 0$.). Note that $h_2 \phi_{\alpha} = \lambda_0 \phi_{\alpha}$ because $\phi_{\alpha}$ is the minimizer of $\mathcal{E}^{\mathrm{NLS}}$ in \eqref{eq:spectralgapn1}. We have
	\begin{equation}
		e_0 = \inf_{\Vert \psi \Vert = 1} \langle \psi, h_1 \psi \rangle \leq \langle \phi_{\alpha}, h_1 \phi_{\alpha} \rangle = \lambda_0 + \int_{\mathbb{R}^3} \left| \phi_{\alpha}(x) \right|^2 \left(N_0 v_N \ast \left|\phi^{\mathrm{H}}_{N_0} \right|^2(x) - \alpha | \phi_{\alpha} (x) |^2 \right) \de x.
		\label{eq:156}
	\end{equation}
	In the following we use the notation $\varrho_{f}(x) = |f(x)|^2$. Since $\hat{v}_N(p) = N^{-1} \hat{v}(p/N^{\kappa})$ we can bound the absolute value of the last term on the right-hand side of \eqref{eq:156} by
	\begin{equation}
		\frac{N_0}{N} \left| \int_{\mathbb{R}^3} \overline{ \hat{\varrho}_{\phi_{\alpha}}(p) } \hat{v}(p/N^{\kappa}) \left[ \hat{\varrho}_{\phi^{\mathrm{H}}_{N_0}}(p) - \hat{\varrho}_{\phi_{\alpha}}(p) \right] \de p \right| + \frac{N_0}{N} \left| \int_{\mathbb{R}^3} \left| \hat{\varrho}_{\phi_{\alpha}}(p) \right|^2 \left[ \hat{v}(p/N^{\kappa}) - \hat{v}(0) \right] \de p \right|.
		\label{eq:157}
	\end{equation}
	The second term is bounded by
	\begin{equation}
		\left( \sup_{p \in \mathbb{R}^3} \left| \left(1+p^2 \right) \hat{\varrho}_{\phi_{\alpha}}(p) \right| \right)^2 \times \int_{\mathbb{R}^3}  \left( \frac{1}{1+p^2} \right)^2 \left| \hat{v}(p/N^{\kappa}) - \hat{v}(0) \right| \de p.
		\label{eq:158}
	\end{equation}
Since $v\in L^1$ by assumption, $\hat v$ is continuous and the last integral goes to zero as $N\to \infty$ by dominated convergence. To derive a bound on the first term in \eqref{eq:158}, we note that $ (1+p^2 ) \hat{\varrho}_{\phi_{\alpha}}(p)$ is the Fourier transform of
\begin{equation}
(1 - \Delta) \phi_\alpha^2  = \phi_\alpha^2 - 2 \phi_\alpha \Delta \phi_\alpha - 2 |\nabla \phi_\alpha|^2 = (1+ 2\lambda_0) \phi_\alpha^2 - \frac{\omega^2}{2} x^2 \phi_\alpha^2 -2  \alpha \phi_\alpha^4 - 2 |\nabla \phi_\alpha|^2\,,
\end{equation}
hence the supremum is bounded by the $L^1$-norm of this expression. 
All the terms on the right-hand side are clearly uniformly bounded in $L^1$ for $0\leq \alpha\leq \alpha_0$ since they are dominated by the energy.

	 It remains to investigate the first term on the right-hand side of \eqref{eq:157}.
	We bound it as  
	\begin{align}
		\frac{N_0}{N} \left| \int_{\mathbb{R}^3} \hat{\varrho}_{\phi_{\alpha}}(p) \hat{v}(p/N^{\kappa}) \left[ \hat{\varrho}_{\phi^{\mathrm{H}}_{N_0}}(p) - \hat{\varrho}_{\phi_{\alpha}}(p) \right] \de p \right| &\leq \Vert \hat{v} \Vert_{L^{\infty}(\mathbb{R}^3)} \Vert \hat{\varrho}_{\phi_{\alpha}} \Vert_{L^1(\mathbb{R}^3)} \left\Vert \left| \phi^{\mathrm{H}}_{N_0} \right|^2 - \left| \phi_{\alpha} \right|^2 \right\Vert_{L^1(\mathbb{R}^3)} \label{eq:163} \\
		&\leq 4 \Vert \hat{v} \Vert_{L^{\infty}(\mathbb{R}^3)}  \Vert \hat{\varrho}_{\phi_{\alpha}} \Vert_{L^1(\mathbb{R}^3)}  \left\Vert \phi^{\mathrm{H}}_{N_0} - \phi_{\alpha} \right\Vert. \nonumber
	\end{align}
	The analysis above shows that $ \hat{\varrho}_{\phi_{\alpha}}$ is uniformly bounded in $L^1$, hence the claimed uniform convergence of the right-hand side follows if 
	\begin{equation}
		\lim_{N \to \infty} \sup_{\beta \omega \gtrsim N^{-1/3}} \left\Vert \phi^{\mathrm{H}}_{N_0} - \phi_{\alpha} \right\Vert = 0 \,.
		\label{eq:164}
	\end{equation}
	To prove \eqref{eq:164} we first show that
	\begin{equation}
		\lim_{N \to \infty} \sup_{\beta \omega \gtrsim N^{-1/3}} \left| E^{\mathrm{H}}(N_0) - E^{\mathrm{NLS}}(\alpha) \right| = 0. 
		\label{eq:165}
	\end{equation}
	Since $\hat{v}_N(0) \geq \hat{v}_N(p)$ for all $p \in \mathbb{R}^3$ we have $E^{\mathrm{NLS}}(\alpha) \geq E^{\mathrm{H}}(N_0)$. Moreover,
	\begin{equation}
		E^{\mathrm{NLS}}(\alpha) \leq \mathcal{E}_{\alpha}\left( \phi^{\mathrm{H}}_{N_0} \right) = E^{\mathrm{H}}(N_0) + \frac{N_0}{2N } \int_{\mathbb{R}^3} \left| \hat{\varrho}_{\phi_{\alpha}}(p) \right|^2 \left[ \hat{v}(0) - \hat{v}(p/N^{\kappa}) \right] \de p.
		\label{eq:167}
	\end{equation}
	In combination with the bound in \eqref{eq:158} this implies \eqref{eq:165}. Using \eqref{eq:165}, we check that $\phi^{ \mathrm{H} }_{N_0}$ is a minimizing sequence for $\mathcal{E}^{\mathrm{NLS}}$. The minimizer of $\mathcal{E}^{\mathrm{NLS}}$ is unique, and standard arguments therefore imply \eqref{eq:164}.  This proves an upper bound for $e_0 - \lambda_0$ with the claimed asymptotic behavior. A lower bound is obtained with the same argument if the roles of $h_1$ and $h_2$ are interchanged. This proves \eqref{eq:155} for $j=0$.
\end{proof}

The second Lemma concerns a uniform bound for $\Vert \phi_{\alpha} \Vert_{L^{\infty}(\mathbb{R}^3)}$.
For a proof, see e.g., \cite[Lemma 2.1]{LSY2dGP}. 

\begin{lemma}
	\label{lem:convergenceofNLSint}
	The norm $\Vert \phi_{\alpha} \Vert_{L^{\infty}(\mathbb{R}^3)}$ is bounded uniformly in $\alpha \in [0,\int v(x) \de x]$.
\end{lemma}
The third Lemma provides us with pointwise bounds for $\varrho^{\mathrm{H},\mathrm{c}}(x)$ and
\begin{equation}
	\eta(x) = \frac{1}{\exp\left( \beta \left( h + N_0 v_N \ast \left| \phi_{N_0}^{\mathrm{H}} \right|^2(x) - \mu_N \right) \right) -1 }(x,x)
	\label{eq:eta}
\end{equation}
with $\mu_N$ defined below \eqref{eq:spectralgap5f}.
\begin{lemma}
	\label{lem:aprioridensity}
	Let $v_N$ satisfy the assumptions of Theorem~\ref{thm:freeenergyscaling1}. Then we have
	\begin{equation}
	\| \varrho^{\mathrm{H},\mathrm{c}}\|_{L^{\infty}(\mathbb{R}^3)} \lesssim \omega^{3/2} N \quad \text{ as well as } \quad \left\Vert \eta - N_0 \left|  \phi_{N_0}^{\mathrm{H}} \right|^2 \right\Vert_{L^{\infty}(\mathbb{R}^3)} \lesssim \beta^{-3/2}\left(1+ (\beta\omega)^{1/2}\right) \,.
	\label{eq:spectralgap45}
	\end{equation}
	The constants in the above inequalities are uniform in $N \geq 1$ and $\beta \omega \gtrsim N^{-1/3}$.
\end{lemma} 
\begin{proof}
	We write $\varrho^{\mathrm{H},\mathrm{c}}(x) = \sum_{j=0}^{\infty} \langle n_j \rangle_{G^{\mathrm{H},\mathrm{c}}} | \varphi_j(x) |^2$ with the eigenfunctions $\{ \varphi_j \}_{j=0}^{\infty}$ of the Hartree operator $h^{\mathrm{H}} = h +  v_N \ast \varrho^{\mathrm{H},\mathrm{c}}(x)$ and denote by $\langle \cdot \rangle_{\mathrm{gc}}$ the expectation w.r.t. the grand-canonical Gibbs state with one-particle Hamiltonian $h^\mathrm{H}$ (and expected particle number $N$). 
	From \cite[Remark~A.1]{me} we know that $\langle n_j \rangle_{G^{\mathrm{H},\mathrm{c}}} \lesssim \langle n_j \rangle_{\mathrm{gc}}$ for all $j \geq 0$, and hence 
	\begin{equation}
	\varrho^{\mathrm{H},\mathrm{c}}(x) \lesssim \sum_{j=0}^{\infty} \langle n_j \rangle_{\mathrm{gc}} | \varphi_j(x) |^2 = \frac{1}{e^{\beta \left( h + v_N \ast \varrho^{\mathrm{H},\mathrm{c}}- \mu^{\mathrm{H}} \right)} -1  }(x,x).
	\label{eq:149}
	\end{equation}
	The chemical potential $\mu^{\mathrm{H}}$ is chosen s.t. the integral over  the right-hand side of \eqref{eq:149} equals $N$.
	
	We pick $D>0$ and $n \in \mathbb{N}$ and use the identity $(e^x-1)^{-1} = \sum_{\alpha=1}^{\infty} e^{- \alpha x}$ for $x>0$ to write the right-hand side of \eqref{eq:149} as 
	\begin{equation}
	\sum_{\alpha=1}^{\infty} e^{-\beta \alpha \left( h^{\mathrm{H}} - \mu^{\mathrm{H}} \right)}(x,x) \label{eq:35}  = \sum_{1 \leq \alpha \leq \frac{D}{\beta \omega}} e^{-\beta \alpha \left( h^{\mathrm{H}} - \mu^{\mathrm{H}} \right) }(x,x) + \sum_{\alpha > \frac{D}{\beta \omega}} \sum_{j=0}^n e^{-\beta \alpha \left( e_j^{\mathrm{H}} - \mu^{\mathrm{H}} \right) } \left| \varphi_j(x) \right|^2 + \sum_{\alpha > \frac{D}{\beta \omega}} \sum_{j>n} e^{-\beta\alpha\left( e_j^{\mathrm{H}} - \mu^{\mathrm{H}} \right) } \left| \varphi_j(x) \right|^2. 
	\end{equation}
	Here $\{ e_j^{\mathrm{H}} \}_{j=0}^{\infty}$ denote the eigenvalues of the operator $h^{\mathrm{H}}$. 
	A simple variational argument using the harmonic oscillator ground state $\psi_0$ as trial state  shows that 
	\begin{equation}\label{sto}
	\mu^{\mathrm{H}} \leq e_0^{\mathrm{H}} 
	\leq \frac{3 \omega}{2} + \left\Vert v_N \ast \varrho^{\mathrm{H},\mathrm{c}} \right\Vert_{L^1(\mathbb{R}^3)} \left\Vert \psi_0 \right\Vert_{L^{\infty}(\mathbb{R}^3)}^2 \lesssim \omega + \left\Vert v_N \right\Vert_{L^1(\mathbb{R}^3)} \left\Vert \varrho^{\mathrm{H},\mathrm{c}} \right\Vert_{L^1(\mathbb{R}^3)} \lesssim \omega.
	\end{equation}
	The Feynman--Kac formula (see e.g. \cite{SimonFunct,BratelliRobinson2}) implies that we can drop the positive potential $\frac{\omega^2 x^2}{4} + v_{N} \ast \varrho^{\mathrm{H},\mathrm{c}}(x)$ in the exponentials in \eqref{eq:35} for an upper bound. 
	The first term on the right-hand side of \eqref{eq:35} is therefore bounded by 
	\begin{equation}
	\sum_{1 \leq \alpha \leq \frac{D}{\beta \omega}} e^{-\beta \alpha (h + v_{N} \ast \varrho^{\mathrm{H},\mathrm{c}} - \mu^{\mathrm{H}} )}(x,x) \lesssim \sum_{1 \leq \alpha \leq \frac{D}{\beta \omega}} e^{\beta \alpha \Delta }(x,x)  \lesssim \beta^{-3/2} 
	\label{eq:upperboundgc9b}
	\end{equation}
	for all $x \in \mathbb{R}^3$, with a constant depending only on $D$. 
	
	Next we consider the second term on the right-hand side of \eqref{eq:35}. We have
	\begin{equation}
	\sum_{j=0}^n \sum_{\alpha > \frac{D}{\beta \omega}} e^{-\beta \alpha\left( e_j^{\mathrm{H}} - \mu^{\mathrm{H}} \right) } \left| \varphi_j(x) \right|^2 \leq \sum_{j=0}^n \frac{1}{e^{\beta \left( e_j^{\mathrm{H}} - \mu^{\mathrm{H}} \right)}-1} \left| \varphi_j(x) \right|^2 \leq \left( \sup_{j \leq n } \left\Vert \varphi_j \right\Vert_{L^{\infty}(\mathbb{R}^3)}^2 \right) N.
	\label{eq:38}
	\end{equation} 
	We thus need a bound for the $L^{\infty}$-norm of the eigenfunctions of the Hartree operator. These can be obtained via 
	\begin{equation}
	\left| \varphi_j(x) \right| \leq \int_{\mathbb{R}^3} \frac{ e_j^{\mathrm{H}} + \omega }{- \Delta + \frac{\omega^2 x^2}{4} + v_N \ast \varrho^{\mathrm{H},\mathrm{c}}(x) + \omega }(x,y) \left| \varphi_j(y) \right| \de y \leq \int_{\mathbb{R}^3} \frac{ e_j^{\mathrm{H}} + \omega }{- \Delta + \omega}(x-y) \left| \varphi_j(y) \right| \de y.
	\label{eq:13a}
	\end{equation}
	To obtain \eqref{eq:13a} we used $a^{-1} = \int_0^{\infty} e^{-ax} \de x$ and the Feynman--Kac formula. The resolvent on the right-hand side of \eqref{eq:13a} is a bounded linear map from $L^2(\mathbb{R}^3)$ to $H^2(\mathbb{R}^3)$, and hence from $L^2(\mathbb{R}^3)$ to $L^{\infty}(\mathbb{R}^3)$. Moreover, an argument similar to \eqref{sto} that uses the min-max principle, see e.g. \cite[Theorem~12.1]{LiLo10}, can be used to bound $e_j^{\mathrm{H}}$ for $j\leq n$, and hence  
	\begin{equation}
	\sup_{j \leq n } \left\Vert \varphi_j \right\Vert_{L^{\infty}(\mathbb{R}^3)} \leq C_n \omega^{3/4}
	\label{eq:36}
	\end{equation}
	for some constant $C_n$ depending only on $n$. 
	
	It remains to consider the last term in  \eqref{eq:35}.
	We choose $n$ such that $e_j^{\mathrm{H}} - \mu^{\mathrm{H}} \geq \omega$ for $j>n$. Note that this can be done independently of $N \geq 1$ and $\beta \omega \gtrsim N^{-1/3}$ because $e_j^{\mathrm{H}} \geq e_j$, where $e_j$ denote the eigenvalues of $h$, and $\mu^{\mathrm{H}} \lesssim \omega$. Accordingly, we have
	\begin{align}
	\sum_{\alpha > \frac{D}{\beta \omega}} \sum_{j>n} e^{-\beta\alpha\left( e_j^{\mathrm{H}} - \mu^{\mathrm{H}} \right) } \left| \varphi_j(x) \right|^2 \leq e^{-D \left( h^{\mathrm{H}} - \mu^{\mathrm{H}} \right)/(2\omega)}(x,x) \sum_{\alpha > \frac{D}{\beta \omega}} e^{-\beta \omega \alpha/2} \lesssim \frac{\omega^{1/2}}{\beta }.
	\label{eq:37}
	\end{align}
	In the last step we used again the Feynman--Kac formula to bound the heat kernel of the Hartree operator by the one the Laplacian. In combination, \eqref{eq:upperboundgc9b}--\eqref{eq:37} prove the first claim in \eqref{eq:spectralgap45}.
	
	The bound for $\eta(x) - N_0 |  \phi_{N_0}^{\mathrm{H}}(x) |^2 $ can be obtained in a similar way using that
	\begin{align}
	\eta(x) - N_0 \left|  \phi_{N_0}^{\mathrm{H}}(x) \right|^2 &= \sum_{\alpha = 1}^{\infty} e^{-\beta \alpha \left( h + v_N \ast | \phi^{\mathrm{H}}_{N_0} |^2(x) - \mu_N \right)}(x,x) - \frac{1}{e^{ \beta( ( e_0 - \mu_N) ) } - 1} | \psi_0(x) |^2 \label{eq:91} \\
	&\leq \sum_{1 \leq \alpha \leq \frac{D}{\beta \omega}} e^{-\beta \alpha\left( h + v_N \ast | \phi^{\mathrm{H}}_{N_0} |^2(x) - \mu_N \right)}(x,x) + \sum_{\alpha > \frac{D}{\beta \omega}} \sum_{j=1}^{\infty} e^{-\beta \alpha\left( e_j - \mu_N \right)} \left| \psi_j(x) \right|^2, \nonumber
	\end{align}
	where  $\{ e_j \}_{j=0}^{\infty}$ and $\{ \psi_j \}_{j=0}^{\infty}$  denote the eigenvalues and eigenfunctions of $h + v_N \ast | \phi^{\mathrm{H}}_{N_0}|^2$, with $\psi_0 =  \phi_{N_0}^{\mathrm{H}}$. A simple trial state argument shows $\mu_N \lesssim \omega$ and allows us to use \eqref{eq:upperboundgc9b} to bound the first term in the second line of \eqref{eq:91}. To bound the last term we use \eqref{eq:37} with $n=0$. This is possible because Lemma~\eqref{lem:eigenvaluebound} shows that the spectral gap of the operator $h + v_N \ast | \phi^{\mathrm{H}}_{N_0}|^2$ above its unique ground state can be bounded from below uniformly in $N \geq 1$ and $\beta \omega \gtrsim N^{-1/3}$ by a constant times $\omega$. In combination, these considerations prove the claim.
\end{proof}

We define 
\begin{equation}
\Delta \varrho^{\mathrm{H},\mathrm{c}}(x) = \varrho^{\mathrm{H},\mathrm{c}}(x) - \eta(x),
\label{eq:deltarho}
\end{equation}
with $\eta$ in \eqref{eq:eta}, and write the Hartree potential as
\begin{equation}
v_N \ast \varrho^{\mathrm{H},\mathrm{c}}(x) = v_N \ast \Delta \varrho^{\mathrm{H},\mathrm{c}}(x) + v_N \ast \left( \eta - N_0 | \phi_{N_0}^{\mathrm{H}} |^2 \right)(x) + N_0 v_N \ast | \phi_{N_0}^{\mathrm{H}} |^2(x).
\label{eq:spectralgap44}
\end{equation}
Our goal is to show that the difference between the $j$-th eigenvalue of $h + v_N \ast \varrho^{\mathrm{H},\mathrm{c}}(x)$ and the $j$-th eigenvalue of $h + N_0 v_N \ast | \phi_{N_0}^{\mathrm{H}} |^2(x)$ converges to zero as $N \to \infty$ with $\beta \omega \gtrsim N^{-1/3}$. 
In order to prove this, it suffices to show that the operator norm of the difference of the resolvents of $h + v_N \ast \varrho^{\mathrm{H},\mathrm{c}}(x)$ and $h + N_0 v_N \ast | \phi_{N_0}^{\mathrm{H}} |^2(x)$ converges to zero in this limit (see, e.g.,
\cite[Theorem~VIII.23]{ReSi1980}). 

We start by noting that $\gamma^{\mathrm{H,c}}$ is an approximate minimizer of the Hartree free energy functional $\mathcal{F}^{\mathrm{H}}$ in the sense of \eqref{eq:spectralgap5e} for $\delta = O (N^{1/3} \ln N)$. This follows from $S(G^\mathrm{H,c}) \leq s(\gamma^\mathrm{H,c})$ (see \cite[Chapter~2.5.14.5]{Thirring_4}) and the free energy bound in Lemma~\ref{lem:freeenergyboundcgc}.
Eq.~\eqref{eq:spectralgap5f} and \cite[Lemma~5.1]{GriesHan2012} therefore imply  that
\begin{equation}
	\left\Vert \Delta \varrho^{\mathrm{H},\mathrm{c}} \right\Vert_{L^1(\mathbb{R}^3)} \lesssim N^{11/12}. 	
	\label{eq:147c}
\end{equation}
In combination with Lemma~\ref{lem:convergenceofNLSint} and Lemma~\ref{lem:aprioridensity}, we conclude that
\begin{align}
\left\Vert v_N \ast \Delta \varrho^{\mathrm{H},\mathrm{c}} \right\Vert_{L^1(\mathbb{R}^3)} &\leq \left\Vert v_N \right\Vert_{L^1(\mathbb{R}^3)} \left\Vert \Delta \varrho^{\mathrm{H},\mathrm{c}} \right\Vert_{L^1(\mathbb{R}^3)} \lesssim  N^{-1/12}, \label{eq:spectralgap53} \\
\left\Vert v_N \ast \Delta \varrho^{\mathrm{H},\mathrm{c}} \right\Vert_{L^{\infty}(\mathbb{R}^3)} &\leq \left\Vert v_N \right\Vert_{L^1(\mathbb{R}^3)} \left\Vert \Delta \varrho^{\mathrm{H},\mathrm{c}} \right\Vert_{L^{\infty}(\mathbb{R}^3)} \lesssim \omega^{3/2} \nonumber
\end{align}
and, in particular, 
\begin{equation}
\left\Vert v_N \ast \Delta \varrho^{\mathrm{H},\mathrm{c}} \right\Vert_{L^2(\mathbb{R}^3)} \lesssim \omega^{3/4} N^{-1/24}. 
\label{eq:spectralgap54}
\end{equation}
We also have 
\begin{equation}
\left\Vert v_N \ast \left( \eta - N_0 | \phi_{N_0}^{\mathrm{H}} |^2 \right) \right\Vert_{L^2(\mathbb{R}^3)} \leq \left\Vert v_N \right\Vert_{L^1(\mathbb{R}^3)} \left\Vert \eta - N_0 | \phi_{N_0}^{\mathrm{H}} |^2 \right\Vert_{L^{\infty}(\mathbb{R}^3)}^{1/2} \left\Vert \eta - N_0 | \phi_{N_0}^{\mathrm{H}} |^2 \right\Vert_{L^{1}(\mathbb{R}^3)}^{1/2} \lesssim \omega^{3/4} N^{-1/4},
\label{eq:spectralgap55}
\end{equation}
which follows from Lemma~\ref{lem:aprioridensity}
and $0 \leq \int_{\mathbb{R}^3} ( \eta(x) - N_0 | \phi_{N_0}^{\mathrm{H}} |^2 ) \de x \lesssim N$. In combination, \eqref{eq:spectralgap54}, \eqref{eq:spectralgap55} and \eqref{eq:spectralgap44} show that
\begin{equation}
\left\Vert v_N \ast \left( \varrho^{\mathrm{H},\mathrm{c}} - N_0 | \phi_{N_0}^{\mathrm{H}} |^2 \right) \right\Vert_{L^2(\mathbb{R}^3)} \lesssim \omega^{3/4} N^{-1/24}\,.
\label{eq:spectralgap56}
\end{equation}
To show that \eqref{eq:spectralgap56} implies the claimed convergence of the difference of the resolvents, we estimate
\begin{align}\label{eq:spectralgap57}
&\left\Vert \frac{1}{h+v_N \ast \varrho^{\mathrm{H},\mathrm{c}} + i} - \frac{1}{h+N_0 v_N* | \phi_{N_0}^{\mathrm{H}} |^2 + i} \right\Vert \\ \nonumber &= \left\Vert \frac{1}{h+v_N \ast \varrho^{\mathrm{H},\mathrm{c}} + i}  v_N \ast \left(\varrho^{\mathrm{H},\mathrm{c}} - N_0 | \phi_{N_0}^{\mathrm{H}} |^2 \right) \frac{1}{h+N_0 v_N* | \phi_{N_0}^{\mathrm{H}} |^2 + i} \right\Vert  \\
&\lesssim \left\Vert  v_N \ast \left( \varrho^{\mathrm{H},\mathrm{c}} - N_0 | \phi_{N_0}^{\mathrm{H}} |^2 \right) \frac{1}{h+N_0 v_N*| \phi_{N_0}^{\mathrm{H}} |^2 + \omega} \right\Vert_2 \nonumber \\
&= \left( \int_{\mathbb{R}^6} \left| v_N \ast \varrho^{\mathrm{H},\mathrm{c}}(x) - N_0 v_N*| \phi_{N_0}^{\mathrm{H}}(x) |^2 \right|^2 \frac{1}{h+N_0 v_N * | \phi_{N_0}^{\mathrm{H}} |^2 + \omega}(x,y)^2 \de(x,y) \right)^{1/2}. \nonumber
\end{align}
The Feynman--Kac formula implies that 
\begin{equation}
0 \leq \frac{1}{h+N_0 v_N*| \phi_{N_0}^{\mathrm{H}} |^2 + \omega}(x,y) \leq \frac{1}{-\Delta + \omega}(x-y)\,,
\label{eq:9}
\end{equation}
and hence  
\begin{equation}\label{7:31}
\left\Vert \frac{1}{h+v_N \ast \varrho^{\mathrm{H},\mathrm{c}} + i} - \frac{1}{h+N_0 v_N* | \phi_{N_0}^{\mathrm{H}} |^2 + i} \right\Vert \lesssim \left\Vert v_N \ast \left( \varrho^{\mathrm{H},\mathrm{c}} - N_0 | \phi_{N_0}^{\mathrm{H}} |^2\right) \right\Vert_{L^2(\mathbb{R}^3)}\,.
\end{equation}
In combination with \eqref{eq:spectralgap56} this proves the claimed bound for the difference of the resolvents, and therefore the closeness of the eigenvalues of $h+v_N \ast \varrho^{\mathrm{H},\mathrm{c}}$ and $h+N_0 v_N*| \phi_{N_0}^{\mathrm{H}} |^2$ in the limit  considered. 
In combination with Lemma~\ref{lem:eigenvaluebound}, this completes the proof of  Proposition~\ref{prop:spectralgap} (in the case $\kappa >0$).
\subsection{Bounds on the  free energy}
\label{sec:freeenergybounds2}
The proof of the lower bound for the free energy is literally the same as the one in Section~\ref{sec:prooflowerboundgc}. We therefore focus on the proof of the upper bound. 

The analysis in Section~\ref{sec:proofupperboundc} remains unchanged until \eqref{eq:upperboundcan6}. Eq.~\eqref{eq:upperboundcan7} needs to be replaced by the bound 
\begin{equation}
\int_{\mathbb{R}^6} v_N(x-y)  | \varphi_j(x) |^2 | \varphi_j(y) |^2 \de(x,y) \leq  N^{3 \kappa -1} \Vert v \Vert_{L^{\infty}(\mathbb{R}^3)}.
\label{eq:upperboundcan7a}
\end{equation}
The proof of Lemma~\ref{lem:momentbound} is based on Lemma~\ref{lem:semiclassicaltraceestimate}, whose proof uses Lemma~\ref{lem:spactralgapSCMF}. The proofs of the first two Lemmas can be applied with obvious adjustments also in the present situation if the spectral gap estimate for the Hartree operator in Lemma~\ref{lem:spactralgapSCMF} is replaced by the one in Proposition~\ref{prop:spectralgap}. The result of Lemma~\ref{lem:momentbound} translated to the present case reads
\begin{equation}
\sum_{j \geq 1} \langle n_j^2 \rangle_{G} \lesssim \left( \frac{1}{ \beta \omega } \right)^{3}.
\label{eq:lemnsquaredsums01}
\end{equation}

Using \eqref{eq:upperboundcan7a} and \eqref{eq:lemnsquaredsums01} we check that 
\begin{equation}
\sum_{j \geq 1} \langle n_j^2 \rangle_G \left( \int_{\mathbb{R}^6} v_N(x-y)  | \varphi_0(x) |^2 | \varphi_0(y) |^2 \de (x,y ) + \int_{\mathbb{R}^6}  v_N(x-y)  | \varphi_j(x) |^2 | \varphi_j(y) |^2 \de(x,y) \right) \label{eq:upperboundcan8b}  \lesssim N^{3 \kappa} \Vert v \Vert_{L^{\infty}(\mathbb{R}^3)} \,.
\end{equation}
Let
\begin{equation}
	\gamma_> = Q \frac{1}{ e^{\beta \left( h+ v_N \ast \varrho^{\mathrm{H},\mathrm{c}}(x) - \mu^{\mathrm{H}} \right)} - 1},
	\label{eq:150}
\end{equation}
with $Q$ denoting the spectral projection onto the complement of the ground state subspace of the Hartree operator $h + v_N \ast \varrho^{\mathrm{H},\mathrm{c}}(x)$. The chemical potential $\mu^{\mathrm{H}}$ is chosen such that the trace of the 1-pdm on the right-hand side of \eqref{eq:150} without $Q$ equals $N$. In order to estimate the exchange terms, we need the following Lemma providing us with an estimate for the integral kernel of $\gamma_>$. Its proof is based on the spectral gap estimate in Proposition~\ref{prop:spectralgap}. 

\begin{lemma}
	\label{lem:boundonoffdiagonalofDM1}
	Let $v_N$ satisfy the assumptions of Theorem~\ref{thm:freeenergyscaling1}. For fixed $D>0$ there exists a constant $C > 0$ such that
	\begin{equation}
	| \gamma_>(x,y) | \lesssim \sum_{1 \leq \alpha \leq \frac{D}{\beta \omega}} e^{-\alpha \beta h}(x,y) + \frac{1}{\beta \omega} \left(  \varphi_0(x) \varphi_0(y) + \omega^{3/2} e^{-C \omega \left( x^2 + y^2 \right)} \right) \,. \label{eq:upperboundgc6}
	\end{equation}
	Here $\varphi_0$ denotes the ground state of the Hartree operator $h + v_N \ast \varrho^{\mathrm{H},\mathrm{c}}(x)$.
\end{lemma}

	\begin{proof}
		We use the identity $(e^x-1)^{-1} = \sum_{\alpha=1}^{\infty} e^{- \alpha x}$ for $x>0$ to write $\gamma_>$ as
		\begin{equation}
		\gamma_> = Q \frac{1}{e^{\beta \left(h + v_{N} \ast \varrho^{\mathrm{H},\mathrm{c}}  - \mu^{\mathrm{H}} \right)}-1} = \sum_{\alpha=1}^{\infty} e^{-\beta \alpha \left(h + v_{N} \ast \varrho^{\mathrm{H},\mathrm{c}}  - \mu^{\mathrm{H}} \right)} - \frac{1}{e^{\beta (e_0^{\mathrm{H}}-\mu^{\mathrm{H}})} - 1} | \varphi_0 \rangle \langle \varphi_0 |.
		\label{eq:upperboundgc7}
		\end{equation}
		Here $e_0^{\mathrm{H}}$ denotes the lowest eigenvalue of $h + v_{N} \ast \varrho^{\mathrm{H},\mathrm{c}}(x)$. 
		By the Feynman--Kac formula, we can drop the positive potential $v_{N} \ast \varrho^{\mathrm{H},\mathrm{c}}(x) $ in the exponential in \eqref{eq:upperboundgc7} for an upper bound. Since $\mu^\mathrm{H} \lesssim \omega$ (see \eqref{sto}) we get, for fixed $D>0$, 
		\begin{equation}
		\sum_{1 \leq \alpha \leq \frac{D}{\beta \omega}} e^{-\beta \alpha \left(h +  v_{N} \ast \varrho^{\mathrm{H},\mathrm{c}} - \mu^{\mathrm{H}} \right)}(x,y) \lesssim \sum_{1 \leq \alpha \leq \frac{D}{\beta \omega}} e^{- \alpha \beta h}(x,y)
		\label{eq:upperboundgc9}
		\end{equation}
		for all $x,y \in \mathbb{R}^3$. 
		
		As above, we denote by $\lbrace e^{\mathrm{H}}_j \rbrace_{j=0}^{\infty}$ the eigenvalues of $h + v_{N} \ast \varrho^{\mathrm{H},\mathrm{c}}(x)$ and by $\{ \varphi_j \}_{j=0}^{\infty}$ the corresponding eigenfunctions. The remaining terms read
		\begin{equation}
		- \sum_{1 \leq \alpha \leq \frac{D}{\beta \omega}} e^{- \beta \alpha \left(e^{\mathrm{H}}_0 - \mu^{\mathrm{H}} \right) } \varphi_0(x) \varphi_0(y) + \sum_{\alpha > \frac{D}{\beta \omega}} \sum_{j\geq 1}  e^{- \beta \alpha \left(e_j^{\mathrm{H}} - \mu^{\mathrm{H}} \right)} \varphi_j(x) \varphi_j(y) \,.\label{eq:11}
		\end{equation}
		The absolute value of the first term is bounded from above by a constant times $(\beta \omega)^{-1}  \varphi_0(x) \varphi_0(y) $. The absolute value of the second term can be estimated by
		\begin{equation}
		\left( \sum_{\alpha > \frac{D}{\beta \omega}} \sum_{j \geq 1} e^{-\beta \alpha \left( e_j^{\mathrm{H}} - \mu^{\mathrm{H}} \right)} | \varphi_j(x) |^2 \right)^{1/2} \left( \sum_{\alpha > \frac{D}{\beta \omega}} \sum_{j \geq 1} e^{-\beta \alpha \left( e_j^{\mathrm{H}} - \mu^{\mathrm{H}} \right)} | \varphi_j(y) |^2 \right)^{1/2}.
		\label{eq:upperboundgcadjustments41}
		\end{equation}
		The latter sums can be bounded as 
		\begin{align}
		&\sum_{\alpha > \frac{D}{\beta \omega}} \sum_{j \geq 1} e^{-\beta \alpha \left( e_j^{\mathrm{H}} - \mu^{\mathrm{H}} \right)} | \varphi_j(x) |^2 \leq \sum_{\alpha > \frac{D}{\beta \omega}} e^{-\beta \alpha \left(e_1^{\mathrm{H}} - \mu^{\mathrm{H}}\right)/2} e^{-\beta \alpha \left( h + v_N \ast \varrho^{\mathrm{H},\mathrm{c}} - \mu^{\mathrm{H}} \right)/2}(x,x) \label{eq:upperboundgcadjustments51} \\
		&\hspace{0.5cm} \leq  e^{-D \left( h + v_N \ast \varrho^{\mathrm{H},\mathrm{c}} - \mu^{\mathrm{H}} \right)/(2 \omega)}(x,x) \sum_{\alpha > \frac{D}{\beta \omega}} e^{-\beta \alpha \left(e_1^{\mathrm{H}} - \mu^{\mathrm{H}} \right)/2} \leq 2 \frac{e^{- D \left( h + v_N \ast \varrho^{\mathrm{H},\mathrm{c}} - \mu^{\mathrm{H}} \right)/(2 \omega)}(x,x)}{\beta\left( e_1^{\mathrm{H}} - \mu^{\mathrm{H}} \right)} \lesssim \frac{ e^{-\frac{D h}{2 \omega}}(x,x)}{\beta \omega} \,. \nonumber
		\end{align}
		In the last step, we used Proposition~\ref{prop:spectralgap} which states that $e_1^\mathrm{H} - \mu^\mathrm{H} \gtrsim \omega$, as well as $\mu^\mathrm{H}\lesssim \omega$  and the Feynman--Kac formula to drop the effective potential $v_N* \rho^\mathrm{H,c}$. 
		From the explicit form of the Mehler kernel
		\begin{equation}
	e^{-t h}(x,y)  = \left( \frac{\omega}{2} \right)^{3/2} \left( \frac{1}{2 \pi \sinh( t \omega )} \right)^{3/2} \exp\left( \frac{-1}{\sinh(t \omega)} \frac{\omega}{2} \left\{ \cosh(t \omega) \left( x^2 + y^2 \right) - 2 x \cdot y  \right\} \right)
		\label{eq:upperboundgc12}
		\end{equation}
		we obtain the bound
		\begin{equation}
		e^{-\frac{D h}{2 \omega} }(x,x) \lesssim \omega^{3/2} e^{- C \omega x^2}
		\label{eq:upperboundgc13}
		\end{equation}
		 for some $C>0$ depending only on $D$. In combination with the considerations above, the completes the proof of \eqref{eq:upperboundgc6}
	\end{proof}

We apply Lemma~\ref{lem:boundonoffdiagonalofDM1} to estimate the first exchange term on the right-hand side of \eqref{eq:upperboundcan6} by 
\begin{equation}
\frac{1}{2} \int_{\mathbb{R}^6} v_N(x-y) | \gamma_{>}(x,y) |^2 \de(x,y) \lesssim \Vert v_N \Vert_{L^{\infty}(\mathbb{R}^3)} \left( \left\Vert \sum_{1 \leq \alpha \leq \frac{D}{\beta \omega}} e^{-\beta \alpha h } \right\Vert_2^2 + (\beta \omega)^{-2}  \right), \label{eq:upperboundgc17}
\end{equation}
where $\Vert \cdot \Vert_2$ denotes the Hilbert-Schmidt norm. The explicit form of  the Mehler kernel in \eqref{eq:upperboundgc12} allows to estimate it as 
\begin{equation}
 \left\Vert \sum_{1 \leq \alpha \leq \frac{D}{\beta \omega}} e^{-\beta \alpha h } \right\Vert_2 \lesssim 
 \frac{1}{(\beta \omega)^{3/2}}.
\label{eq:93}
\end{equation}
Accordingly, \eqref{eq:upperboundgc17} is bounded from above by a constant times $\omega N^{3 \kappa}$. 

In order to estimate the second exchange term on the right-hand side of \eqref{eq:upperboundcan6}, we apply Lemma~\ref{lem:boundonoffdiagonalofDM1} once more and find
\begin{align}\nonumber 
&N_0 \text{Re}  \int_{\mathbb{R}^6} v_N(x-y) \overline{ \varphi_{0}(x) } \varphi_0(y) \gamma_>(x,y) \de(x,y) \\ & \lesssim N_0 \int_{\mathbb{R}^6} v_N(x-y) | \varphi_{0}(x) | | \varphi_0(y) | \left( \sum_{1 \leq \alpha \leq \frac{D}{\beta \omega}} e^{-\beta \alpha h}(x,y) \right) \de(x,y) +  \frac{N_0 \Vert v \Vert_{L^1(\mathbb{R}^3)}}{N \beta \omega} \left(  \left\Vert \varphi_0 \right\Vert_{L^4(\mathbb{R}^3)}^4 + \omega^{3/2} \right). \label{eq:upperboundgc18}
\end{align}
The $L^4(\mathbb{R}^3)$-norm of the ground state $\varphi_0$ of the Hartree operator $h + v_N \ast \varrho^{\mathrm{H},{\mathrm{c}}}$ can be bounded uniformly in $N$ by a constant times $\omega^{3/8}$. This follows from \eqref{eq:36} and $\Vert \varphi_0 \Vert = 1$. To estimate the first term on the right-hand side of \eqref{eq:upperboundgc18}, we split the sum over $\alpha$ into terms with $1 \leq \alpha \leq N^{-2 \kappa}/(\beta \omega)$ and the complement with $N^{-2 \kappa}/(\beta \omega) < \alpha \leq D/(\beta \omega)$. In the first case, we bound $v_N$ by its $L^\infty$ norm and obtain 
\begin{equation}
N_0 \int_{\mathbb{R}^6} v_N(x-y) | \varphi_{0}(x) | | \varphi_0(y) | \left( \sum_{1 \leq \alpha \leq \frac{N^{-2\kappa}}{\beta \omega}} e^{-\beta \alpha h }(x,y) \right) \de(x,y) \leq \frac{N^{\kappa}}{\beta\omega}  \| v\|_{L^\infty(\mathbb{R}^3)}
\end{equation}
since $e^{-\beta \alpha h }\leq 1$ as on operator. In the second case, we bound $v_N$ by its $L^1$ norm and the exponential factor in \eqref{eq:upperboundgc12} by $1$ with the result that
\begin{equation}
N_0 \int_{\mathbb{R}^6} v_N(x-y) | \varphi_{0}(x) | | \varphi_0(y) | \left( \sum_{ \frac{N^{-2\kappa}}{\beta \omega} \leq \alpha \leq \frac{D}{\beta \omega}} e^{-\beta \alpha h }(x,y) \right) \de(x,y) \lesssim \|v\|_{L^1(\mathbb{R}^3)}  \sum_{ \frac{N^{-2\kappa}}{\beta \omega} \leq \alpha \leq \frac{D}{\beta \omega}} (\beta\alpha)^{-3/2}  \lesssim  \omega^{3/2} \|v\|_{L^1(\mathbb{R}^3)}  \frac{N^{\kappa}}{\beta \omega} \,.
\end{equation}
In combination with \eqref{eq:upperboundcan6}, \eqref{eq:upperboundcan8b} and \eqref{eq:upperboundgc17}--\eqref{eq:upperboundgc18}, we find
\begin{equation}
\tr\left[ \mathcal{V}_N G^\mathrm{H,c} \right] \leq D_N(\varrho^\mathrm{H,c},\varrho^\mathrm{H,c}) + \text{const. } \omega \ \left[ N^{3 \kappa} + \frac{N^{\kappa}}{\beta \omega}  \right] \label{eq:upperboundgc27a}
\end{equation}
with $D_N$ defined in \eqref{eq:abbreviationinteraction}. In combination with \eqref{eq:finalestimate1}, the representation of the Hartree free energy in terms of the minimizer $\gamma^{\mathrm{H},\mathrm{c}}$ in Lemma~\ref{lem:minimizersHartree} and Lemma~\ref{lem:freeenergyboundcgc} it implies
\begin{equation}
F^{\mathrm{c}}(\beta,N,\omega) \leq F^{\mathrm{H}}(\beta,N,\omega) + \text{const. } \omega N^{1/3} \left( N^{\kappa} + \ln N \right), \label{eq:upperboundgc32a}
\end{equation}
with a constant that is uniform in $\beta \omega \gtrsim N^{-1/3}$. This concludes the proof of \eqref{eq:mainresults1}. 
\begin{remark}
	The spectral gap estimate for the Hartree operator in Proposition~\ref{prop:spectralgap} is necessary in order to obtain the free energy bounds in the canonical ensemble. If one is only interested in the statement of Theorem~\ref{thm:freeenergyscaling1} for the grand-canonical ensemble the spectral gap estimate can be avoided. In order to do that, one chooses a trial state based on the grand-canonical Gibbs state but with the lowest $M = O(1)$ modes replaced by a coherent state. More precisely, these modes have to be removed from the Gibbs with a partial trace. The trial state is then defined as a tensor product of this Gibbs state with a coherent state. The latter is chosen such that the previously traced out modes have the same expected occupation as in the grand-canonical Gibbs state related to the Hartree minimizer. This leaves the kinetic energy unchanged, changes the entropy only by an additive term of the order $\ln N$ and allows one to obtain an effective spectral gap of the order $\omega$ for the thermal cloud if $M$ is chosen sufficiently large (independently of $N \geq 1$ and $\beta \omega \gtrsim N^{-1/3}$). In case of the canonical ensemble such a construction is not available and we need Proposition~\ref{prop:spectralgap}.
\end{remark}  
\subsection{Bound on the 1-pdm}
The analysis in Section~\ref{sec:boundgrandcanonicaldensitymatrix} applies with obvious adjustments to the present situation and proves \eqref{eq:mainresults1c} for any approximate minimizer of the Gibbs free energy functional in the sense of \eqref{eq:mainresults1b}. This concludes the proof of Theorem~\ref{thm:freeenergyscaling1}.

\subsection{Proof of Lemma~\ref{lem:Hartreeenergy}}\label{sec:Appendix}
We first derive appropriate upper and lower bounds for the Hartree free energy. Afterwards we prove the statement \eqref{eq:spectralgap5f} about the asymptotics of approximate minimizers of $\mathcal{F}^{\mathrm{H}}$.
\subsubsection*{Upper bound for $F^{\mathrm{H}}(\beta,N,\omega)$}
The proof of the upper bound is based on a trial state argument. As the result suggests, we need a trial state that looks like a non-interacting thermal cloud plus a condensate that is, as in the zero temperature case, described by the Hartree minimizer. 

\textit{The trial state:} Let $Q$ be the spectral projection onto the orthogonal complement of the ground state subspace of $h$ in \eqref{eq:h2}. Our trial state is given by
\begin{equation}
	\gamma_N = N_0 | \phi_{N_0}\rangle \langle \phi_{N_0}| + Q \gamma_{N,0},
	\label{eq:94}
\end{equation}
with $\phi_{N_0}$ the minimizer  of the Hartree energy functional \eqref{eq:spectralgapn1},  $\gamma_{N,0}$ in \eqref{eq:gc1-pdm} and $N_0$ the expected number of particles in the condensate of the ideal gas defined below \eqref{eq:gc1-pdm}. We dropped the superscript H in $\phi_{N_0}$ to simplify the notation. Note that $\trs[\gamma_N] = N$. We need to compute the free energy $\mathcal{F}^{\mathrm{H}}(\gamma_N)$. 

As in \eqref{eq:42} we use the monotonicity of $f$ in \eqref{eq:bosonicentropy} to see that
\begin{equation}
	-s(\gamma_N) \leq \trs[Qf(\gamma_{N,0})]\,.
	\label{eq:95}
\end{equation}
The kinetic energy reads
\begin{equation}
	\trs[h \gamma_N] = N_0 \langle \phi_{N_0}, h \phi_{N_0}  \rangle + \trs[Q h \gamma_{N,0}]
	\label{eq:96}
\end{equation}
and the interaction energy is bounded by
\begin{equation}
	D_N\left( \varrho_{\gamma_N},\varrho_{\gamma_N} \right) \leq  N_0^2 D_N( |\phi_{N_0}|^2, |\phi_{N_0}|^2) 
	+  N_0 \Vert v_N \Vert_{L^1(\mathbb{R}^3)}  \sup_{x \in \mathbb{R}^3} (Q \gamma_{N,0})(x,x) \label{eq:97} + \frac 12\Vert v_N \Vert_{L^1(\mathbb{R}^3)} N \sup_{x \in \mathbb{R}^3} (Q \gamma_{N,0})(x,x). 
\end{equation}
From Lemma~\ref{lem:aprioridensity} (second bound in \eqref{eq:spectralgap45} with $v=0$) we know that the suprema on the right-hand side are bounded from above by a constant times $\beta^{-3/2}(1+(\beta\omega)^{1/2})$. 
Hence
\begin{equation}
	D_N\left( \varrho_{\gamma_N},\varrho_{\gamma_N} \right) \leq N_0^2 D_N( |\phi_{N_0}|^2, |\phi_{N_0}|^2) + \text{const. } \omega^{3/2} \Vert v \Vert_{L^1(\mathbb{R}^3)} \left(  (\beta \omega)^{-3/2} + (\beta\omega)^{-1} \right) \,.
	\label{eq:99}
\end{equation}

In the final step we combine the term on the right-hand side of \eqref{eq:95} and the second term on the right-hand side of \eqref{eq:96} and find 
\begin{align}
	\trs[Q h \gamma_{N,0}] + \frac{1}{\beta} \trs[Qf(\gamma_{N,0})] = \frac{1}{\beta} \trs\left[ Q \ln\left( 1-\exp(-\beta(h-\mu_0)) \right) \right] + \mu_{0} (N-N_0) \label{eq:100} \\
	\leq F_0(\beta,N,\omega) + \frac{\text{const. } \ln N}{\beta} - \mu_0 N_0 \,,  \nonumber
\end{align}
where the free energy $F_0(\beta,N,\omega)$ of the ideal gas and $\mu_0$ are defined in \eqref{eq:101} and \eqref{eq:gc1-pdm}. Using the definition of $N_0$ below \eqref{eq:gc1-pdm} we see that
\begin{equation}
	\mu_0 = - \frac{1}{\beta} \ln\left( 1 + N_0^{-1} \right) + \frac{3 \omega}{2},
	\label{eq:102}
\end{equation}
and hence $- \mu_0 N_0 \leq \beta^{-1}$. In combination, the above considerations imply
\begin{equation}
	\mathcal{F}^{\mathrm{H}}(\gamma_N) \leq F_0(\beta,N,\omega) + N_0 E^{\mathrm{H}}(N_0) + \text{const. } \omega \left[ \frac{ 1 }{ (\beta \omega)^{3/2} } + \frac{\ln N}{\beta \omega} \right].
	\label{eq:103} 
\end{equation}
This concludes the upper bound.

\subsubsection*{Lower bound for $F^{\mathrm{H}}(\beta,N,\omega)$}
For $0 \leq \gamma \leq N$ with $\trs \gamma = N$ and $\eta \in L^1(\mathbb{R}^3)$ we have
\begin{align}
\mathcal{F}^{\mathrm{H}}(\gamma) \geq \trs[\left( h + v_N \ast \eta \right) \gamma] - \frac{1}{\beta} s(\gamma) - D_N (\eta,\eta) =& \frac{1}{\beta} \trs\left[ \ln\left( 1 - e^{-\beta\left( h + v_N \ast \eta - \mu \right)} \right) \right] + \mu N - D_N(\eta,\eta) \label{eq:spectralgap33} \\
&+ \frac{1}{\beta} \mathcal{S}\left( \gamma, \gamma^{\mathrm{G}}_{\eta} \right) \nonumber
\end{align}
for any $\mu$ such that $h+v_N \ast \eta - \mu > 0$. The bosonic relative entropy $\mathcal{S}$ is defined in \eqref{eq:asymptoticsgc3} and 
\begin{equation}
\gamma^{\mathrm{G}}_{\eta} = \frac{1}{e^{\beta \left( h + v_N \ast \eta - \mu \right)} - 1}.
\label{eq:spectralgap34}
\end{equation}
We denote
\begin{equation}
	\tilde{\mu}_{N} = E^{\mathrm{H}}(N_0) + \frac{N_0}{2}  \int_{\mathbb{R}^3} \left| \phi_{N_0}(x) \right|^2 v_N(x-y) \left| \phi_{N_0}(y) \right|^2 \!\de x \label{eq:Hartreechemicalptential}
\end{equation} 
and choose $\eta(x) = N_0 | \phi_{N_0} (x) |^2$ as well as $\mu = \tilde{\mu}_{N} + \mu_0$ with $\mu_0$ in \eqref{eq:102}. Note that with this choice we have $\text{inf\,spec\,} (h+v_N \ast \eta - \mu) = \text{inf\,spec\,} (h - \mu_0) = \beta^{-1} \ln ( 1 + N_0^{-1})$. 

Pick $M \in \mathbb{N}$ such that $e_j(h) - \mu \geq \omega$ for $j > M$, where $e_j(A)$ denotes the $j$-th eigenvalue of the operator $A$ in increasing order. Note that $M$ can be chosen independently of $N \geq 1$ and $\beta \omega \gtrsim N^{-1/3}$ because $\tilde{\mu}_{N} \leq 2 E^{\text{H}}(N_0) \lesssim \omega$. Hence
\begin{equation}
\sum_{j < M} \ln\left( 1 - e^{-\beta \left( e_j\left( h + v_N \ast \eta \right) - \mu \right) } \right) \gtrsim - \ln N, \label{eq:spectralgap35}
\end{equation}
which implies that
\begin{align}
\frac{1}{\beta} \trs\left[ \ln\left( 1 - e^{-\beta\left( h + v_N \ast \eta - \mu \right)} \right) \right] \geq \frac{1}{\beta} \sum_{j \geq M} \ln\left( 1 - e^{-\beta \left( e_j( h )  - \mu \right) } \right) - \frac{\text{const. } \ln N }{\beta}.
\label{eq:spectralgap36}
\end{align}
Using the concavity of the function $x \mapsto \ln(1-e^{-x})$ we see that
\begin{equation}
\sum_{j \geq M} \ln\left( 1 - e^{-\beta \left( e_j( h )  - \tilde{\mu}_{N} - \mu_0 \right) } \right) \geq \sum_{j \geq 0} \ln\left( 1 -e^{-\beta \left( e_j(h) - \mu_0 \right)} \right) - \sum_{j \geq M} \frac{ \beta \tilde{\mu}_{N} }{e^{\beta \left( e_j(h) - \tilde{\mu}_{N} - \mu_0 \right)}-1}.
\label{eq:spectralgap37}
\end{equation}
A similar argument based on the concavity of $x \mapsto -1/(e^x-1)$ and $\tilde{\mu}_{N} \leq 2 E^{\mathrm{H}}(N_0) \lesssim \omega$ implies 
\begin{equation}
- \sum_{j \geq M} \frac{ \tilde{\mu}_{N} }{e^{\beta \left( e_j(h) - \tilde{\mu}_{N} - \mu_0 \right)}-1} \geq - \sum_{j \geq M} \frac{ \tilde{\mu}_{N} }{e^{\beta \left( e_j(h) - \mu_0 \right)}-1} - \frac{ \text{const. } \omega}{(\beta \omega)^2}.
\label{eq:spectralgap38}
\end{equation}
In combination, \eqref{eq:spectralgap36}--\eqref{eq:spectralgap38} show that 
\begin{equation}
\frac{1}{\beta} \trs\left[ \ln\left( 1 - e^{-\beta\left( h + v_N \ast \eta - \mu \right)} \right) \right] \geq \frac{1}{\beta} \trs\left[ \ln\left( 1 - e^{-\beta\left( h - \mu_0 \right)} \right) \right] - (N-N_0) \tilde{\mu}_{N} - \text{const. } \left( \frac{\ln N}{\beta} + \frac{\omega}{(\beta \omega)^2} \right)\,,
\label{eq:spectralgap39}
\end{equation}
and hence 
\eqref{eq:spectralgap33}  yields
\begin{equation}
\mathcal{F}^{\mathrm{H}}(\gamma) \geq F_0(\beta,N,\omega) + N_0 E^{\mathrm{H}}(N_0) + \frac{1}{\beta} \mathcal{S}\left( \gamma, \gamma^{\mathrm{G}}_{\eta} \right) - \text{const. } \omega \left( \frac{\ln N}{\beta \omega} + \frac{1}{(\beta \omega)^2} \right).
\label{eq:spectralgap41}
\end{equation}
Together  with \eqref{eq:103} this proves \eqref{eq:spectralgap5d}. 
\subsubsection*{Asymptotics of the 1-pdm of approximate minimizers of the Hartree free energy functional}
Let $\gamma_N$ be an approximate minimizer of $\mathcal{F}^{\mathrm{H}}$ in the sense that \eqref{eq:spectralgap5e} holds. Using \eqref{eq:spectralgap41}, we see that its 1-pdm $\gamma_N$ obeys
\begin{equation}
\mathcal{S}\left( \gamma_N, \gamma^{\mathrm{G}}_{\eta} \right) \lesssim \beta \omega \left( N^{2/3} + \delta \right)\,.
\label{eq:spectralgap42}
\end{equation}
The same arguments as in Section~\ref{sec:boundgrandcanonicaldensitymatrix}  then imply \eqref{eq:spectralgap5f}. We note that the spectral gap estimate in Lemma~\eqref{lem:eigenvaluebound} is needed during this analysis. This concludes the proof of Lemma~\ref{lem:Hartreeenergy}.
\vspace{1cm}
\newline
\textbf{Acknowledgments.}
Funding from the European Union’s Horizon 2020 research and innovation programme under the ERC grant agreement No~694227~(R.~S.) and under the Marie Sklodowska-Curie grant agreement No~836146~(A.~D.) is gratefully acknowledged. A.~D. acknowledges support of the Swiss National Science Foundation through the Ambizione grant PZ00P2~185851.

\vspace{0.5cm}

(Andreas Deuchert) Institute of Mathematics, University of Zurich \\ 
Winterthurerstrasse 190, 8057 Zurich, Switzerland \\ 
E-mail address: \texttt{andreas.deuchert@math.uzh.ch}

(Robert Seiringer) Institute of Science and Technology Austria (IST Austria)\\ Am Campus 1, 3400 Klosterneuburg, Austria\\ E-mail address: \texttt{robert.seiringer@ist.ac.at}

\end{document}